\documentclass[%
 reprint,
 amsmath,amssymb,
 aps,
 nofootinbib,
]{revtex4-2}

\usepackage{enumitem}
\usepackage{amsfonts}
\usepackage[matrix,arrow,curve]{xy}
\usepackage{braket}

\usepackage{bm}

\usepackage{graphicx}
\usepackage{hyperref}
\usepackage{xcolor}
\hypersetup{
    colorlinks,
    linkcolor={red!80!black},
    citecolor={blue!50!black},
    urlcolor={blue!70!black}
}

\usepackage{tikz}
\tikzset{oplus/.style={draw,circle,inner sep=1pt, draw=black!100,fill = white!100,line width=1.5pt,minimum size = 15pt,path picture={%
      \draw[black]
       (path picture bounding box.south) -- (path picture bounding box.north) 
       (path picture bounding box.west) -- (path picture bounding box.east);
      }}}

\usepackage{CJKutf8}

\definecolor{dkpurple}{rgb}{0.6,0,0.6}
\newcommand{\zhuan}[1]{{\color{orange}\ifmmode\text{\footnotesize(ZL) #1}\else\footnotesize{(ZL) #1}\fi}}
\newcommand{\roger}[1]{{\color{dkpurple}\ifmmode\text{\footnotesize(RM) #1}\else\footnotesize{(RM) #1}\fi}}
\newcommand{\toERASE}[1]{{\small\color{blue} #1}}
	\renewcommand{\toERASE}[1]{} \renewcommand{\zhuan}[1]{} \renewcommand{\roger}[1]{}	

\usepackage{amsthm}
\theoremstyle{definition}     
\newtheorem{definition}{Definition}[section]
\newtheorem{theorem}[definition]{Theorem}
\newtheorem{proposition}[definition]{Proposition}
\newtheorem{lemma}[definition]{Lemma}
\newtheorem{corollary}[definition]{Corollary}

\newcommand{\defeq}{\stackrel{\mathrm{def}}{=}}
\newcommand*{\D}[2][]{\mathinner{\mathrm{d}^{#1}\mkern-1mu{#2}}}
\newcommand{\bigXp}[1]{\mathinner{\bigl\langle #1 \bigr\rangle}}
\newcommand{\ketbra}[2]{\mathinner{\ket{#1}\mkern-5mu\bra{#2}}}

\newcommand{\rrangle}{\rangle\mkern-4mu\rangle}
\newcommand{\E}{\mathcal{E}}        
\newcommand{\M}{\mathcal{M}}
\newcommand{\R}{\mathcal{R}}

\newcommand{\Z}{\mathbb{Z}}
\newcommand{\Tr}{\operatorname{Tr}} 
\newcommand{\tr}{\operatorname{tr}} 

\newcommand{\Prob}{\mathrm{Pr}}
\newcommand{\auxX}{\hat{X}}
\newcommand{\auxU}{\hat{U}}
\newcommand{\auxV}{\hat{V}}
\newcommand{\Th}{\mathbf{h}}
\newcommand{\Tv}{\mathbf{v}}
\newcommand{\Td}{\mathbf{d}} 
\newcommand{\Ti}{\mathbf{i}}
\newcommand{\Tj}{\mathbf{j}}
\newcommand{\Tk}{\mathbf{k}}
\newcommand{\Tl}{\mathbf{l}}
\newcommand{\Ts}{\mathbf{s}}
\newcommand{\Tt}{\mathbf{t}}
\newcommand{\DMxS}[1]{\mathcal{S}^{(#1)}}
\newcommand{\DMxV}[1]{\mathcal{V}^{(#1)}}
\newcommand{\phd}{{\vphantom{\dag}}}
\newcommand{\ccc}{{\bm c}}

\newcommand{\gPath}{\mathfrak{q}}
\newcommand{\JHh}{\mathord{\diamond\mkern-1.2mu\diamond}}
\newcommand{\JHv}{\mathord{\substack{\textstyle\diamond\\[-1.4pt]\textstyle\diamond}}}

\begin{document}

\title{Replica topological order in quantum mixed states and quantum error correction}

\author{Zhuan Li (李專)}
 \affiliation{
 Department of Physics and Astronomy, University of Pittsburgh, Pittsburgh, Pennsylvania 15260, USA
}%
\affiliation{Pittsburgh Quantum Institute,  Pittsburgh, Pennsylvania 15260, USA}

\author{Roger S. K. Mong (蒙紹璣)}
\affiliation{
 Department of Physics and Astronomy, University of Pittsburgh, Pittsburgh, Pennsylvania 15260, USA
}%
\affiliation{Pittsburgh Quantum Institute, Pittsburgh, Pennsylvania 15260, USA}

\date{\today}

\begin{abstract}
Topological phases of matter offer a promising platform for quantum computation and quantum error correction.
Nevertheless, unlike its counterpart in pure states, descriptions of topological order in mixed states remain relatively under-explored.
Our work give two definitions for replica topological order in mixed states, which involve $n$~copies of density matrices of the mixed state.
Our framework categorizes topological orders in mixed states as either quantum, classical, or trivial, depending on the type of information that can be encoded.
For the case of the toric code model in the presence of decoherence, we associate for each phase a quantum channel and describes the structure of the code space.
We show that in the quantum-topological phase, there exists a postselection-based error correction protocol that recovers the quantum information, while in the classical-topological phase, the quantum information has decohere and cannot be fully recovered.
We accomplish this by describing the mixed state as a projected entangled pairs state (PEPS) and identifying the symmetry-protected topological order of its boundary state to the bulk topology.
We discuss the extent that our findings can be extrapolated to $n \to 1$ limit.
\end{abstract}

\begin{CJK*}{UTF8}{bsmi}
\maketitle
\end{CJK*}

\hbadness=3000

\section{Introduction}
\label{sec:intro}

Topological phases of matter lie beyond the traditional Landau-Ginzburg symmetry-breaking paradigm, where phases are characterized by local order parameters.
Their unique properties, such as ground state degeneracy, long-range entanglements, and nontrivial excitations, originated from the topological nontrivial structures of the Hamiltonian and thus remain stable when subjected to local perturbations. 
This inherent stability positions topological phases as prime candidates for error correction and fault-tolerant quantum computation~\cite{KITAEV20032, RevModPhys.NayakTopologicalQuantumComputation}. 
For instance, their degenerate ground states can be used to encode quantum information, and the braiding of nontrivial excitations can act as quantum gates.
Multiple characteristics have been used to capture the topological order in two-dimensional (2D) pure states, including topology-dependent ground state degeneracy~\cite{wen1990TopologicalOrder},
topological entanglement entropy~\cite{TEEkitaev,TEELevinWen,YiZhang2012GroundStateEntanglement,DMRG2012} and nontrivial ground state modular transform~\cite{YiZhang2012GroundStateEntanglement,PhysRevLettCincio,PhysRevB.91.035127Over,PhysRevBZhuan}.

In reality, quantum devices are noisy, prone to environmental effects, which means we must inevitably need to deal with quantum mixed states.
Recently, the interplay between mixed states and topology has drawn increased attention due to its connections with quantum error correction and state preparation. 
While it is argued that 2D topological order is destroyed by thermal effects~\cite{KITAEV2006Honeycomb,PhysRevLett_Hastings_NozeroTemperature}, recent works have noted that they are stable against small local quantum decoherence~\cite{2023BaoBoundarySPT,DiagnosticsRelativeEntropyNegativity,wang2023fcondensation,sang2023mixedstateRG}.
For example, it has been observed that the Toric code (TC)~\cite{KITAEV20032} model---a paragon of topological order---has a finite error threshold against bit-flip and phase errors~\cite{dennis2002topologicalMemory,Jean_Marie_Maillard_2003_RBIM}.
This ensures that, when the error rate is within this threshold, the quantum information stored in the initial pure state can still be retrieved from the mixed state, suggesting the existence of mixed-state topological order.
Subsequently, error thresholds for various other channels have also been found~\cite{KatzgraberPhysRevLett2009ColorCode,PhysRevX2012DepolarizingThreshold,PhysRevA2015ZnToricCodeThreshold,PhysRevLett2017Ampdamping,KubicaPhysRevLett3DcolorcodeThreshold,UltrahighErrorThresholdYNoise,chubb2021MapsToStatistical,bonilla2021xzzx,xiao2024exactXZZX}, implying that topological order remains robust beyond the realm of pure states.

Noisy topologically ordered states have also enjoyed a renewed surge of interest due to recent research regarding topological state preparation~\cite{Satzinger2021TCsuperconducting,verresen2022efficientlyPreparation,tantivasadakarn2022LREfromSRE,bravyi2022PreparationOfNonAbelianState,fossfeig2023adaptiveQCpreparation,iqbal2023D4TrappedIon,Tantivasadakarn2023PreparationOfNonAbelianState,Tantivasadakarn2023D4State}.
Finite depth adaptive circuits have been designed and implemented on various quantum architectures with a moderate number of qubits.
As the current generation of devices are limited by local decoherence---e.g.\ lifetimes of the individual components---it is natural to ask: in what sense are the prepared states topological?
These concerns motivate a detailed study into the topological order of mixed states.

In contrast to the well-defined pure state topological order, a precise definition for mixed state topological order is still in development.
One approach is to generalize entropy measurements, such as topological entanglement entropy, coherent information, and topological quantum negativity, to the mixed-state density matrix~\cite{Castelnovo_finiteTemperatureTopologicalEntropy, Castelnovo_negativity_TC_finiteTemperature,Lu_negativity_finiteTemperature,DiagnosticsRelativeEntropyNegativity,wang2023fcondensation,colmenarez2023thresholdsCoherentInformation}.
As with the pure state case, the sharp transition of these quantities in mixed states can reveal changes in the system's topological order.

An alternative approach is to describe a mixed state by ``vectorizing'' it to a pure state living in the double space~\cite{DiagnosticsRelativeEntropyNegativity,2023BaoBoundarySPT,PRXQuantumLeeBoundarySPT}.
In this framework, the double-space pure-state topological order depicts the mixed-state topological order.
This approach allows us to characterize mixed-state topological orders, conceptually and numerically, using the more familiar tools for pure states, such as anyon condensation, bipartite entanglement entropy, etc.
However, the error threshold for the quantum channel, as determined using the vectorization method, typically differs from that found in the error correction context.
For example, the error threshold for the the bit- and phase-flip channels are $\approx 18\%$ in Ref.~\onlinecite{2023BaoBoundarySPT}, greater than the established value of $\approx 11\%$~\cite{dennis2002topologicalMemory,Jean_Marie_Maillard_2003_RBIM},
We understand this discrepancy arises because the vectorization method, utilizes two copies of the mixed state in characterizing and computing its properties.
By extension, the error thresholds changes as the number copies of the mixed state increases~\cite{Jean_Marie_Maillard_2003_RBIM,DiagnosticsRelativeEntropyNegativity}.

This raises several critical questions: What are the physical interpretations of these error thresholds?
Can we define a unified theory that encompasses the different interpretations of topological order?
How do the various methods for detecting topological orders relate to one another? 
Notably, recent studies suggest that the mixed-state topological order exhibits a richer structure than its counterpart in the pure state.
There are so-called classical topologically ordered mixed states within which only classical information can be encoded~\cite{2023BaoBoundarySPT,wang2023fcondensation}. Such states have no corresponding pure state topological order and, as a result, require new descriptions.
What are the fundamental differences between these states and the more traditional topological states?
We will address some of these questions in this paper. 

There are four main objectives of this paper:
(1) To give an intrinsic, process-independent definition of topological orders in mixed states, distinguishing between \emph{quantum} and \emph{classical} topological order.
(2) To classify possible topological phases constructed from the TC model subject to local decoherence.
(3) To give an interpretation to replica mixed-state topological orders in the quantum error correction paradigm.
(4) To establish tensor network framework for characterizing mixed state topological order.

In Sec.~\ref{sec:TopoMixed}, we present two definitions of \emph{replica topological order} for mixed states.
Our definitions involve $n$ replicas of a state (which shares the same conceptual basis as the R\'enyi entropy), and are based on properties of the mixed state density matrix itself rather than the process which produces it.
The doublefield/vectorized methods~\cite{DiagnosticsRelativeEntropyNegativity,2023BaoBoundarySPT,PRXQuantumLeeBoundarySPT} can be regarded as a special case when $n=2$.
Our first definition establishes the mixed-state topological order through the geometric structure of density matrix space.
The second definition characterizes topological order via nonlocal Wilson loop operators, and their algebraic properties.
Both definitions capture the types of information encoded in mixed-state space, allowing us to classify states either trivial, \emph{classical topologically ordered}, or \emph{quantum topologically ordered}.

\begin{figure}
	\centering
	\includegraphics[width = \linewidth]{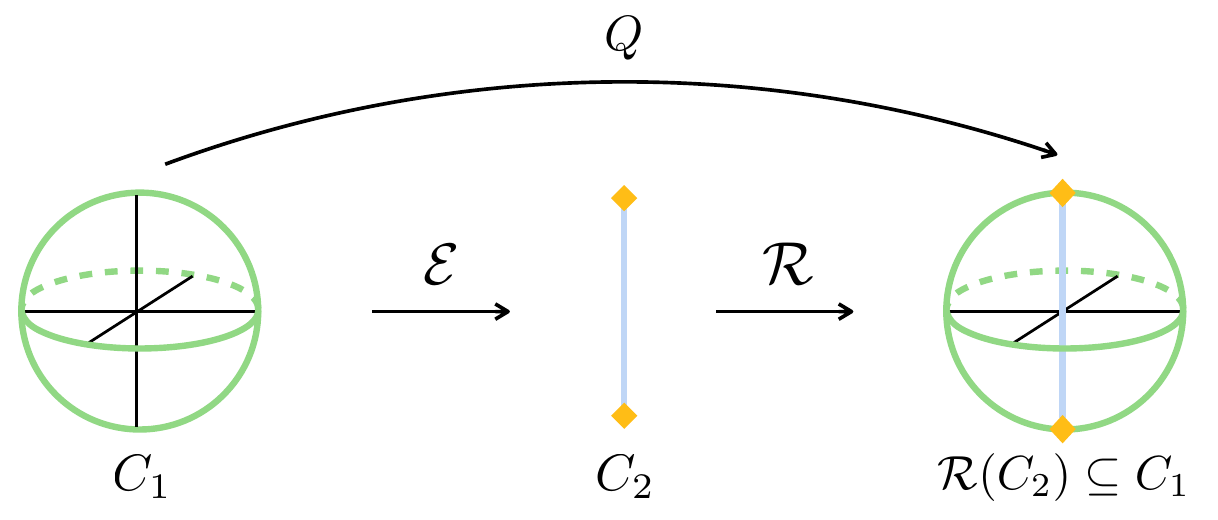}
	\caption{Example of a quantum map on a logical qubit encoded within a topological phase.
	$C_1$ is the Bloch sphere (along with its interior), representing the density matrix space of the qubit.
	Applying an error channel $\E$ to the topological state turns it into a mixed state.
	In this example, $\E$ acts as a dephasing channel which maps $C_1$ to $C_2$, a line segment representing the logical space space of the mixed state.
	The optimal recovery channel $\R$ embeds the mixed state logical space $C_2$ back to the original topological state space $C_1$.
	From the effective quantum channel of entire process $Q = \R \circ \E$, we can infer the structure of the mixed state space $C_2$.
	} \label{fig:qubit_geo}
\end{figure}

In Sec.~\ref{sec:GaugeModel}, we present examples of classical topologically ordered states, using both Abelian and non-Abelian lattice gauge theories.
These unconventional phases are constructed as incoherent sums of gauge configurations.
These examples provide first supports for our definitions of mixed-state topological orders.

In Sec.~\ref{sec:boundarySPT}, we present a postselection-based error correction protocol as a concrete physical interpretation of the mixed-state topological order resulting from a local noisy quantum channel.
Our protocol starts with $n$ replicas of a topologically ordered state $\rho_\text{init} = \rho_1^{\otimes n}$.
Subsequently external noises and syndrome measurements are applied on each copy of $\rho_1$.
The key step of our protocol is to postselect the same error configurations among different replicas.

To study this error correction protocol, we use the powerful projected entangled pair state (PEPS) tensor network method, where the error channel, syndrome measurement, and postselection are represented by tensor elements.
The PEPS tensor network relates the quantum channel acting on the logical space with the \emph{boundary symmetry topological (SPT) order} of the network.
This allows us to classify the descendants---mixed states that arises from local decoherence channels---of the toric code based on their optimal recovery map.
With this formulism, we also prove that there are exactly $p+3$ phases that are descended from the $\Z_p$ TC model in Sec.~\ref{sec:enum_boundary_SPT}.

In Sec.~\ref{sec:ClassifyMixed}, we classify the mixed state topological orders of the descendants of the $\Z_p$ toric code in accordance to our two definitions.
For the geometric definition, we determine the structure of the recovery space by explicitly writing down the quantum channel of the code space; this idea is illustrated in Fig.~\ref{fig:qubit_geo}.
Regarding the Wilson loop definition, we infer the algebra of Wilson loop operators in the noisy mixed-state space from the boundary SPT order of the tensor network.
For each descendant phase identified in Sec.~\ref{sec:enum_boundary_SPT}, we find the same mixed-state topological order classification under both definitions, thereby supporting the equivalence between our two definitions of mixed-state topological orders.
We close this section with various remarks regarding the anyon condensation picture of decoherence~\cite{2023BaoBoundarySPT}, and chiral topological order.

In Sec.~\ref{sec:numerics}, we numerically detect the mixed state topological order of the $\Z_2$ TC model that subjects to different error channels using the PEPS tensor network.
These calculations support our theoretical definition and demonstrate PEPS’s power in simulating the replica topological order.
Unlike previous research that map the TC model to exactly solvable statistical models~\cite{dennis2002topologicalMemory,Jean_Marie_Maillard_2003_RBIM,KatzgraberPhysRevLett2009ColorCode,chubb2021MapsToStatistical,KubicaPhysRevLett3DcolorcodeThreshold}, the PEPS approach establishes a framework for dealing with general error channels.   

\begin{figure}[tb]
\begin{align*}
	\begin{tikzpicture}
	[baseline={([yshift=-.5ex]current bounding box.center)},every text node part/.style={align=center}],
	\node[rectangle,draw] (3) at (70pt,-20pt) { Examples\\[-4pt] \S\ref{sec:GaugeModel}};
	\node[rectangle, draw] (2a) at (0pt,20pt) { Geometric definition \\[-4pt] \S\ref{sec:TopoMixedGeo}};
	\node[rectangle, draw] (2b) at (140pt,20pt) { Wilson loop definition \\[-4pt] \S\ref{sec:TopoMixedWL}};
	\node[rectangle, draw] (4a) at (140pt,-60pt) { Boundary SPT order \\[-4pt]  \S\ref{sec:PEPS}};
	\node[rectangle, draw] (4b) at (0pt,-60pt) { Logical space \\[-4pt] quantum channel \\[-4pt] \S\ref{sec:QEC_quantum_channel}};
	\node[rectangle, draw] (6) at (140pt,-100pt) { Numerics \\[-5pt]
	\S\ref{sec:numerics}};
	\draw [<->,shorten >= 2pt, shorten <= 3pt,line width=1pt] (2a) edge (3) (2b) edge (3);
	\draw [<->,shorten >= 2pt, shorten <= 2pt,line width=1pt] (4a) edge (6);
	\draw [<->,shorten >= 2pt, shorten <= 2pt,line width=1pt] (2a) edge (4b);
	\draw [<->,shorten >= 2pt, shorten <= 2pt,line width=1pt] (2b) edge (4a);
	\draw [<->,shorten >= 2pt, shorten <= 2pt,line width=1pt] (4b) edge (4a);
	\node[] () at (-14pt,-20pt) {\S\ref{sec:ClassificationGeo}};
	\node[] () at (154pt,-20pt) {\S\ref{sec:ClassificationWL}};
	\node[] () at (67pt,-52pt) {\S\ref{sec:QEC_quantum_channel}};
	\end{tikzpicture} 
\end{align*}
\vspace{-4ex}
\caption{A schematic describing what the main ideas of this work are and how they are connected.}
\label{fig:schematic}
\end{figure}
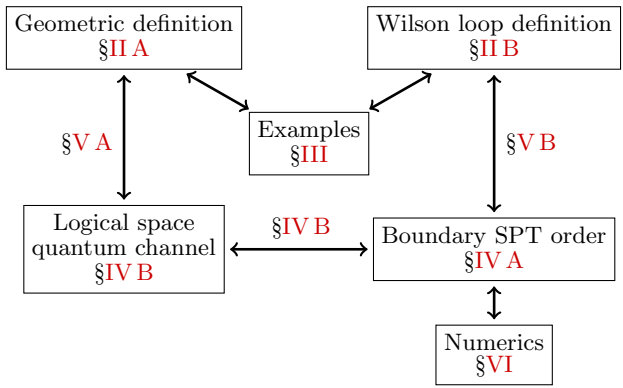

The main concepts and overall logical flow of this paper is illustrated in Figure~\ref{fig:schematic}.

\section{Topological order in mixed states}
\label{sec:TopoMixed}

In this section, we define the topological ordered mixed state.
Let us start by reviewing the properties of topologically ordered pure state.

A pure state $\ket{\varphi}$ on a closed orientable topologically not-simply-connected manifold (e.g.\ torus) is (nontrivially) topologically ordered if there exists a different state $\ket{\psi}$ such that  $\ket{\varphi}$ and $\ket{\psi}$ are \emph{locally indistinguishable}.
This condition guarantees the degenerate Hilbert space and allows quantum information to be encoded therein, hidden and protected against local perturbations.
The set of normalized locally indistinguishable wavefunctions forms the complex projective space $\mathbb{CP}^{d-1}$;  where $d$ is the dimension of the Hilbert space.
(E.g.\ $\mathbb{CP}^{1}$ is the Bloch sphere for qubits.)
The geometry of the manifold $\mathbb{CP}^{d-1}$ reveals information in regards to the type of topological order $\ket{\varphi}$ possesses.

Another defining characteristic of topological pure states are the existence of the Wilson loop operators.
The Wilson loop operator $W_\ell(a)$ is defined from the process of creating a particle-antiparticle pair $(a,\bar{a})$, moving one of them along the loop $\ell$, and then annihilating the particle pair back into the vacuum.
Around noncontractible loops, Wilson loop operators acts non-trivially within the degenerate Hilbert space, acting as logical operators in a topological encoding.

Inspired by these properties of topological pure states, we asks the analogous questions: what information can be hidden in a mixed state, what are the operators that we can use to manipulate this information?
This leads us to two definitions for topological mixed state: a geometric definition and a Wilson loop definition.
We conjecture that these two definitions are equivalent, and we explicitly show that they are the same in the toric code model under local quantum channels in Sec.~\ref{sec:ClassifyMixed}.

We first define $n$-replica (in)distinguishability in context of mixed states; when two states should be considered the different/same and when they belong to the same topological phase.

\vspace{1ex}\noindent\textbf{Definitions}:
Let $n$ be a positive integer.

Two density matrices $\rho$ and $\sigma$ are \emph{$n$-replica locally indistinguishable}, if 
\begin{align} \label{eq:n_indisting}
    \frac{\Tr(A_1\rho A_2\rho \cdots A_n \rho)}{\Tr(\rho^n)}
    = \frac{\Tr(A_1\sigma A_2\sigma \cdots A_n \sigma)}{\Tr(\sigma^n)}
\end{align}
for any set of operators $\{A_1,\dots,A_n\}$ with support $\mathcal{R} = \bigcup \mathrm{Supp}(A_i)$ on a contractible region.

We use the $n$-Schatten norm to measure state normalization and distances between states.
\begin{align} \label{eq:Schatten_norm}
    \left\lVert \alpha \right\rVert_n \defeq \frac{1}{\mathcal{N}} \Bigl[ \Tr\bigl(|\alpha|^n\bigr) \Bigr]^{1/n} .
\end{align}
$|\alpha| = \sqrt{\alpha^\dag \alpha}$ denote the ``absolute value'' of $\alpha$.
(If $\alpha$ is positive semidefinite, then $|\alpha| = \alpha$.)
$\mathcal{N}$ is a normalization constant to be discussed shortly.

Two density matrices $\rho$ and $\sigma$ are \emph{$n$-replica globally indistinguishable}, if their distance---defined as follows---vanishes 
\begin{align} \label{eq:Schatten_dist}
    \mathrm{dist}_n(\rho,\sigma) \defeq \frac{1}{2^{1/n}} \bigl\lVert\rho-\sigma\bigr\rVert_n \,.
\end{align}
When $n=1$, this reduces to the trace distance, which equals to zero if and only if the fidelity $F(\rho,\sigma) = 1$.

It is understood that we are characterizing states in the thermodynamic limit.
For example, typically topological grounds states in finite systems are locally indistinguishable up to an exponentially small correction in system size; we expect the analogous behaviour to occur for mixed states.
Crucially, it is possible for two states $\rho$ and $\sigma$ to be globally distinguishable for one $n$ but not another.%
    \footnote{This is clearly impossible for finite systems where $\lVert \rho-\sigma \rVert_n = 0 \Leftrightarrow \rho=\sigma$.}
Consequently the topological classification of a state may depends on the replica index $n$.

Generically, mixed states have an exponentially large number of eigenvalues and $\Tr \rho^n$ vanishes in the thermodynamic limit for $n > 1$.
To give meaning to our state normalization and distance measure, we need to scale $\mathcal{N}$ with the system size.
An appropriate choice would be to let $\mathcal{N} = \exp\bigl[ -(1-\tfrac{1}{n}) s_n N \bigr]$, where $s_n$ is the $n$\textsuperscript{th}-R\'enyi entropy density, $N$ is the system size.

The notion of global indistinguishability allows us to define equivalent classes of density matrices.
Simply put, we group together states that are $n$-replica globally indistinguishable from one another into a set, which we call \emph{$n$-replica states}.%
	\footnote{The construction of $\DMxS{n}$ and $\DMxV{n}$ rely on the conjecture that globally indistinguishable states are also locally indistinguishable.  We are not aware of any counterexamples to this conjecture.}
\begin{itemize}
\item Let $\DMxS{n}(\rho)$ denote the set of $n$-replica states that are $n$-replica locally indistinguishable from $\rho$.
\item Let $\DMxV{n}(\rho)$ denote the complex vector space generated by $\DMxS{n}(\rho)$.
\end{itemize}
The elements of $\DMxS{n}(\rho)$ must be positive semidefinite with unit trace.
In contrast, elements in $\DMxV{n}(\rho)$ may not even be Hermitian.

For what follows we require that $\DMxS{n}(\rho)$ be a \emph{convex set}, that is, if $\sigma, \sigma' \in \DMxS{n}(\rho)$ then $x \sigma + (1-x) \sigma' \in \DMxS{n}(\rho)$ for $x\in[0,1]$.
This must be true for $n=1$ since Eq.~\eqref{eq:n_indisting} is linear.
We provide an argument for the $n=2$ case as follows.  Let $|\rho\rrangle$ denote the vectorized form of density matrix $\rho$.%
\footnote{For a density matrix $\sigma$ in the Hilbert space $\mathcal{H}$, its vectorized form $|\sigma\rrangle$ can be considered as a pure state in $\mathcal{H} \times \mathcal{H}^\ast$.}
Then Eq.~\eqref{eq:n_indisting} implies that $|\rho\rrangle$ and $|\sigma\rrangle$ are topologically degenerate, that is, locally indistinguishable as pure states.
Both are ground states of some Hamiltonian on 2-layers,
and so the state $x|\rho\rrangle + (1-x)|\sigma\rrangle$ must also be topologically degenerate.

\subsection{Geometric definition}
\label{sec:TopoMixedGeo}
We classify mixed states based on the geometry of $\DMxS{n}(\rho)$ \textbf{on the torus}.
Analogous to the pure state case, the set of states topologically degenerate with $\rho$ on a manifold (with genus greater than zero) encodes characteristics of topological order.
An extreme point of $\DMxS{n}(\rho)$ is a state that \emph{cannot} be written as a positive linear combination of two different states.
\begin{itemize}
\item [(i)]
    $\rho$ is called $n$-replica \textbf{trivial} if $\DMxS{n}(\rho)$ is a single point.
\item [(ii)]
    $\rho$ is called $n$-replica \textbf{classical topologically ordered} (CTO) if $\DMxS{n}(\rho)$ has a finite number of extreme points.
\item [(iii)]
    $\rho$ is called $n$-replica \textbf{quantum topologically ordered} (QTO) if the extreme points of $\DMxS{n}(\rho)$ form a submanifold with dimension $\geq 1$.
\end{itemize}
Here ``classical" and ``quantum" reflect the type of information that can be stored the set $\DMxS{n}(\rho)$.

Now define $\mathcal{B}(d)$ to be the set density operators in a $d$-dimensional Hilbert space, that is, the set of matrices
\begin{align} \label{eq:qudit_space}
	\mathcal{B}(d) \defeq \bigl\{ M \in \mathbb{C}^{d \times d} \,\big|\, M^\dag = M ,\, M \succeq 0 ,\, \tr M = 1 \bigr\} .
\end{align}
$\mathcal{B}(d)$ has dimensions $d^2-1$, and its extreme points is the set of rank-1 matrices, diffeomorphic to $\mathbb{CP}^{d-1}$.

If a pure state $\ket{\psi}$--say the ground state of some Hamiltonian--is topologically ordered in the conventional sense,
then its corresponding mixed state ensemble is always QTO.
Indeed, $\DMxS{1}(\ketbra\psi\psi)$ is the set of density matrices of $\mathcal{H}$ with $\mathcal{H}$ the Hilbert space of states locally indistinguishable to $\ket{\psi}$, and is isomorphic to $\mathcal{B}(\dim\mathcal{H})$.
The extreme points of $\DMxS{1}\ketbra\psi\psi)$ are the pure states $\ketbra\psi\psi$ of $\mathcal{H}$.
Importantly, $\mathbb{CP}^{d-1}$ is a manifold with infinite number of points (for $d>1$) and hence ``quantum''.

The same is true for $\DMxS{n}(\ketbra\psi\psi)$, with the following proof:
First note that if $\rho $ and $\sigma $ are pure states, then $\mathrm{dist}_n(\rho,\sigma) = \mathrm{dist}_1(\rho,\sigma) = \sqrt{1-\Tr(\rho\sigma)}$ (when $\mathcal{N}$ is set to 1).

Any state $\sigma \in \DMxS{1}(\ketbra\psi\psi)$ can be diagonalized: $\sigma = \sum_i p_i \ketbra{\varphi_i}{\varphi_i}$ such that $\ket{\varphi_i}$ form an orthonormal basis for $\mathcal{H}$.
The ratio~\eqref{eq:n_indisting} 
$\Tr(A_1 \sigma A_2 \sigma \cdots A_n \sigma) / \Tr(\sigma^n)$
has denominator $\sum_i p_i^n$.
Its numerator consists of terms of the form $\braket{\varphi_{i_n}|A_1|\varphi_{i_1}} p_{i_1} \braket{\varphi_{i_1}|A_2|\varphi_{i_2}} p_{i_2} \braket{\varphi_{i_2}|A_3|\varphi_{i_3}}  p_{i_3} \cdots$, which vanishes unless $i_1 = i_2 = \dots = i_n$ as the orthogonal states of a topologically-degenerate Hilbert space $\mathcal{H}$ have vanishing matrix elements for all local operators.
Furthermore, each braket $\braket{\varphi_i|A_k|\varphi_i} = \braket{\psi|A_k|\psi}$ since the $\ketbra{\varphi_i}{\varphi_i}$ is (1-replica) indistinguishable from $\ketbra{\psi}{\psi}$.
Hence
\begin{align} \notag
 &   \frac{\Tr(A_1\sigma A_2\sigma \cdots A_n \sigma)}{\Tr(\sigma^n)}
= \frac{\sum_i p_i^n \braket{\varphi_i|A_1|\varphi_i}\braket{\varphi_i|A_2|\varphi_i} \cdots}
    {\sum_i p_i^n}
\\ &= \braket{\psi|A_1|\psi} \braket{\psi|A_2|\psi} \cdots \braket{\psi|A_n|\psi} ,
\end{align}
from which we conclude that $\DMxS{1}(\ketbra\psi\psi) \in S^{(n)}(\ketbra\psi\psi)$

On the other hand, if $\sigma \in \DMxS{n}(\ketbra\psi\psi)$, then for any local operator $A$,
\begin{align}
    \frac{\Tr(A\sigma^n)}{\Tr(\sigma^n)}
    = \frac{\Tr\bigl[ A (\ketbra{\psi}{\psi})^n \bigr]}{\Tr\bigl[ (\ketbra{\psi}{\psi})^n \bigr]}
    = \braket{\psi|A|\psi} .
\end{align}
Therefore, $\sigma^n/\Tr(\sigma^n) \in \DMxS{1}(\ketbra\psi\psi)$.
Diagonalizing $\sigma^n$ into
$\sigma^n/\Tr(\sigma^n) = \sum_i q_i \ketbra{\varphi_i}{\varphi_i}$,
it is evident that $\ket{\varphi_i} \in \mathcal{H}$
and so $\sigma \propto \sum_i \sqrt[n]{q_i} \ketbra{\varphi_i}{\varphi_i} \in S^{(1)}(\ketbra\psi\psi)$
and hence $\DMxS{n}(\ketbra\psi\psi) \in \DMxS{1}(\ketbra\psi\psi)$.
Therefore $\DMxS{n}(\ketbra\psi\psi) = \DMxS{1}(\ketbra\psi\psi)$ and thus $\ketbra\psi\psi$ is $n$-replica QTO.

In contrast, the phase space of a probability distribution forms a simplex and is CTO.
A probability distribution over $N$ items is characterized by $\{p_1,p_2,\dots,p_N\}$ constrained to $p_1+p_2+\dots+p_N=1$.
The geometry of this space is an $(N-1)$-simplex, its extreme points are the $N$ vertices corresponding the cases when one of the items have unit probability.
Density matrix spaces with such structure do not exhibit quantum characteristics (such as destructive interference) and hence is ``classical''.

\subsection{Wilson loop definition}
\label{sec:TopoMixedWL}
Alternately, we classify mixed state topological order using the algebra of Wilson loop operators acting on the vector space $\DMxV{n}(\rho)$.

We call $W$ a \emph{non-identity operator} if it acts as an automorphism of $\DMxV{n}(\rho)$ via both left- and right-multiplication and is not proportional to the identity map.
That is, (a) $W R, R W \in \DMxV{n}(\rho)$ for all $R \in \DMxV{n}(\rho)$, and (b) for all $\alpha \in \mathbb{C}$ there exists some $Q \in \DMxV{n}(\rho)$ such that $W Q \neq \alpha Q \neq Q W$.
The operators are classified as follows.
Suppose $W$ is a non-identity operator.
\begin{itemize}[itemsep=0pt, parsep=4pt, topsep=2pt]
\item	$W$ is \textbf{classical} if it commutes with every element of $\DMxV{n}(\rho)$.
\item	$W$ is \textbf{quantum} if there exists $R \in V^{(n)}(\rho)$ such that $W R \neq R W$.
\end{itemize}

\vspace{1ex}
A mixed state $\rho$ on a topological non-trivial manifold is called
\begin{itemize}
\item [(i)]
    $\rho$ is $n$-replica \textbf{trivial} if there are no such non-identity operator; i.e., all operators act trivially: $W R \propto R$ for all $W$ and $R \in \DMxV{n}(\rho)$.
\item [(ii)]
    $\rho$ is $n$-replica \textbf{classical topologically ordered} (CTO) if (a) there exists at least one non-identity operator, and that (b) all non-identity operators are classical: $W R = R W$ for all $W$ and $R \in \DMxV{n}(\rho)$.
\item [(iii)]
    $\rho$ is $n$-replica \textbf{quantum topologically ordered} (QTO) if there exists a quantum non-identity operator. 
\end{itemize}
Physically, the non-identity operators acting on the subspace describes Wilson loop operators (or combinations of such) of the state.

The intuition behind these definitions is captured by the following statement.

\begin{proposition} \label{prop:noncommuting_W}
Suppose there exists a pair of non-identity operators $W_1$ and $W_2$ that do not commute, then both operators $W_1$ and $W_2$ are quantum and the phase is QTO.
\end{proposition}

\begin{proof}
We prove this claim by contradiction.
Let $W_1$ and $W_2$ be operators as stated in the claim,
but suppose $W_1$ is classical.

Let $W_3 = W_1 W_2 - W_2 W_1$. Since $W_3 \neq 0$,  we can always find a $R_\ast \in \DMxV{n}(\rho)$ such that $W_3 R_\ast \neq 0$.
Then,
\begin{align} \begin{aligned}
&	W_2 R_\ast W_1 = W_1 W_2 R_\ast = (W_3 + W_2 W_1) R_\ast \\&\quad = W_3 R_\ast + W_2 W_1 R_\ast = W_3 R_\ast + W_2 R_\ast W_1 \,,
\end{aligned} \end{align}
which leads to $W_3 R_\ast = 0$, a contradiction.  Hence $W_1$---and by the same argument, $W_2$---is quantum.
Therefore the phase is QTO.
\end{proof}

In a pure state topological phase, states are transformed between one another via Wilson loop operators.
Since intersecting Wilson loop operators generally do not commute,%
	\footnote{Let $W_\Th(a)$, $W_\Tv(b)$ be Wilson loop operators along intersecting simple loops $\Th$, $\Tv$.  The operators commute only if $S_{a,b} = S_{0,a} S_{0,b} / S_{0,0}$, where $0$ is the vacuum anyon and $S$ is the $S$-matrix of the TQFT.  In a TQFT, for every anyon $a \neq 0$ there exists anyon $b$ such that $S_{a,b} \neq S_{0,a} S_{0,b} / S_{0,0}$ and hence all (non-identity) Wilson loops are quantum.}
	the phase is rightfully classified as quantum.
In fact, the non-commuting nature of logical operators is critical to manipulation of quantum information and operation of a quantum computer!

In a CTO phase, $\DMxV{n}(\rho)$ are spanned by states $\rho_1, \rho_2, \dots$ (that are locally indistinguishable from one another) which are orthogonal operators: $\rho_i \rho_j = \delta_{i,j} \rho_i$.
The orthogonal states can be distinguished by Wilson loop operators, taking the form $W \rho_i = \alpha_i \rho_i = \rho_i W$,%
	\footnote{To explain why $W\rho_i$ must be a multiple of $\rho_i$, consider hypothetically $W \rho_1 = \alpha_1 \rho_1 + \alpha_2 \rho_2$.  Then $(W-\alpha_1) \rho_1 \rho_2 = \alpha_2 \rho_2^2 \neq 0$ despite $\rho_1 \rho_2 = 0$, leading to a contradiction.}
but cannot be transformed into one another via conjugation $W \rho_i W^{-1} = \rho_i$.
Hence all non-identity operators in a CTO phase are classical.

A phase may have both classical and quantum Wilson loop operators.
For example, the product of a CTO phase with a QTO phase, or an error channel acting on a non-Abelian phase.
In such cases, the phase is still classified as being QTO.

\section{Classical Topological Order example: incoherent gauge model}
\label{sec:GaugeModel}

In this section, we explicitly construct examples of classical topologically order mixed state from a discrete gauge theory.
In a discrete gauge theory, the wavefunction of a topological pure state is determined uniquely by local constraints of the gauge configuration.

The discrete gauge theory is dual to the loop gas construction, which we explain.
In the loop gas construction, the wavefunction is a coherent sum of loop configurations.
Loop configurations that differ by local deformations---such as the creation and annihilation of loops, and loop surgery---are given the same amplitude.

We begin by a brief review of the toric code, realizing $\Z_2$ gauge theory.
For the toric code on a square lattice, there is a qubit degree of freedom on each edge of the square lattice.
It is realized as a loop gas on the dual lattice (also a square lattice), with spin up/down corresponding to the presence/absence of loop crossing the edge.
The toric code wavefunctions are all $+1$ eigenstate of the operators
\begin{align} \label{eq:toric_code_stab}
A_+ &=
	\begin{tikzpicture}
	[baseline={([yshift=-.5ex]current bounding box.center)},rnode1/.style={circle,inner sep=1pt, draw=black!100,fill = white!100,line width=1.5pt,minimum size = 8pt},
	rnode2/.style={circle,inner sep=1pt, draw=black!100,fill = black!100,line width=1.5pt,minimum size = 8pt}],
	\node[] (1) at (-12pt,0pt) {$\scriptstyle X$};
	\node[] (2) at (12pt,0pt) {$\scriptstyle X$};
	\node[] (3) at (0pt,-12pt) {$\scriptstyle X$};
	\node[] (4) at (0pt, 12pt) {$\scriptstyle X$};
	\draw[color=black!100,line width=.6pt] (1) -- (2);
	\draw[color=black!100,line width=.6pt] (3) -- (4);
	\end{tikzpicture},
&
B_\square &= 
	\begin{tikzpicture}
	[baseline={([yshift=-.5ex]current bounding box.center)},rnode1/.style={circle,inner sep=1pt, draw=black!100,fill = white!100,line width=1.5pt,minimum size = 8pt},
	rnode2/.style={circle,inner sep=1pt, draw=black!100,fill = black!100,line width=1.5pt,minimum size = 8pt}],
	\node[] (1) at (-12pt,0pt) {$\scriptstyle Z$};
	\node[] (2) at (12pt,0pt) {$\scriptstyle Z$};
	\node[] (3) at (0pt,-12pt) {$\scriptstyle Z$};
	\node[] (4) at (0pt, 12pt) {$\scriptstyle Z$};
	\draw[color=black!100,line width=.6pt] (1) -- (-12pt,12.25pt);
	\draw[color=black!100,line width=.6pt] (4) -- (-12.25pt,12pt);
	\draw[color=black!100,line width=.6pt] (1) -- (-12pt,-12.25pt);
	\draw[color=black!100,line width=.6pt] (3) -- (-12.25pt,-12pt);
	\draw[color=black!100,line width=.6pt] (2) -- (12pt,-12.25pt);
	\draw[color=black!100,line width=.6pt] (3) -- (12.25pt,-12pt);
	\draw[color=black!100,line width=.6pt] (2) -- (12pt,12.25pt);
	\draw[color=black!100,line width=.6pt] (4) -- (12.25pt,12pt);
	\end{tikzpicture},
\end{align}
for every vertex $+$ and plaquette $\square$.

On the torus, the Hilbert space of wavefunctions is four-fold degenerate.
Let $\Th$ and $\Tv$ be two simple loops on the torus that wraps around horizontally and vertically respectively.
Denote $L_{n_\Th,n_\Tv}$ as the set of loop configurations that intersect
	the $\Th$ loop an $\Bigl\{ \begin{array}{@{}ll} \text{even} & n_\Th=0 \\[-0.5ex] \text{odd} & n_\Th=1 \end{array}$ number of times,
	the $\Tv$ loop an $\Bigl\{ \begin{array}{@{}ll} \text{even} & n_\Tv=0 \\[-0.5ex] \text{odd} & n_\Tv=1 \end{array}$ number of times.
The Hilbert space is spanned by the four wavefunctions
\begin{align} \label{eq:TC_wavefunctions} \begin{aligned}
    \ket{\Psi_{00}} \propto \sum_{g\in L_{00}} \ket{g} = \left|  \begin{tikzpicture}
        [baseline={([yshift=-.5ex]current bounding box.center)}],
        \draw[color=black!100,line width=0.5 pt] (-9pt,3pt)-- (9pt,3pt);
         \draw[color=black!100,line width=0.5 pt] (-9pt,-3pt) -- (9pt,-3pt);
         \draw[color=black!100,line width=0.5 pt] (3pt,-9pt)-- (3pt,9pt);
         \draw[color=black!100,line width=0.5 pt] (-3pt,-9pt)-- (-3pt,9pt);
    \end{tikzpicture} \right \rangle+\left|  \begin{tikzpicture}
        [baseline={([yshift=-.5ex]current bounding box.center)}],
        \draw[color=black!100,line width=0.5 pt] (-9pt,3pt)-- (9pt,3pt);
         \draw[color=black!100,line width=0.5 pt] (-9pt,-3pt) -- (9pt,-3pt);
         \draw[color=black!100,line width=0.5 pt] (3pt,-9pt)-- (3pt,9pt);
         \draw[color=black!100,line width=0.5 pt] (-3pt,-9pt)-- (-3pt,9pt);
         \draw[color=red!100,line width=.7pt] plot [smooth cycle, tension=.8] coordinates{(-3pt,0pt) (0pt,3pt)(-3pt,6pt)  (-6pt,3pt)};
    \end{tikzpicture} \right \rangle+\left| 
    \begin{tikzpicture}
        [baseline={([yshift=-.5ex]current bounding box.center)}],
        \draw[color=black!100,line width=0.5 pt] (-9pt,3pt)-- (9pt,3pt);
         \draw[color=black!100,line width=0.5 pt] (-9pt,-3pt) -- (9pt,-3pt);
         \draw[color=black!100,line width=0.5 pt] (3pt,-9pt)-- (3pt,9pt);
         \draw[color=black!100,line width=0.5 pt] (-3pt,-9pt)-- (-3pt,9pt);
         \draw[color=red!100,line width=.7pt] plot [smooth cycle, tension=.8] coordinates{(-3pt,0pt) (3pt,0pt) (6pt,3pt)(3pt,6pt)  (-3pt,6pt)  (-6pt,3pt)};
    \end{tikzpicture} \right \rangle + \dots,
\\
    \ket{\Psi_{01}} \propto \sum_{g\in L_{01}} \ket{g} = \left|  \begin{tikzpicture}
        [baseline={([yshift=-.5ex]current bounding box.center)}],
        \draw[color=black!100,line width=0.5 pt] (-9pt,3pt)-- (9pt,3pt);
         \draw[color=black!100,line width=0.5 pt] (-9pt,-3pt) -- (9pt,-3pt);
         \draw[color=black!100,line width=0.5 pt] (3pt,-9pt)-- (3pt,9pt);
         \draw[color=black!100,line width=0.5 pt] (-3pt,-9pt)-- (-3pt,9pt);
         \draw[color=red!100,line width=0.7 pt] (-9pt,-6pt) -- (9pt,-6pt);
    \end{tikzpicture} \right \rangle+\left|  \begin{tikzpicture}
        [baseline={([yshift=-.5ex]current bounding box.center)}],
        \draw[color=black!100,line width=0.5 pt] (-9pt,3pt)-- (9pt,3pt);
         \draw[color=black!100,line width=0.5 pt] (-9pt,-3pt) -- (9pt,-3pt);
         \draw[color=black!100,line width=0.5 pt] (3pt,-9pt)-- (3pt,9pt);
         \draw[color=black!100,line width=0.5 pt] (-3pt,-9pt)-- (-3pt,9pt);
         \draw[color=red!100,line width=.7pt] plot [smooth cycle, tension=.8] coordinates{(-3pt,0pt) (0pt,3pt)(-3pt,6pt)  (-6pt,3pt)};
         \draw[color=red!100,line width=0.7 pt] (-9pt,-6pt) -- (9pt,-6pt);
    \end{tikzpicture} \right \rangle+\left| 
    \begin{tikzpicture}
        [baseline={([yshift=-.5ex]current bounding box.center)}],
        \draw[color=black!100,line width=0.5 pt] (-9pt,3pt)-- (9pt,3pt);
         \draw[color=black!100,line width=0.5 pt] (-9pt,-3pt) -- (9pt,-3pt);
         \draw[color=black!100,line width=0.5 pt] (3pt,-9pt)-- (3pt,9pt);
         \draw[color=black!100,line width=0.5 pt] (-3pt,-9pt)-- (-3pt,9pt);
         \draw[color=red!100,line width=.7pt] plot [smooth cycle, tension=.8] coordinates{(-3pt,0pt) (3pt,0pt) (6pt,3pt)(3pt,6pt)  (-3pt,6pt)  (-6pt,3pt)};
         \draw[color=red!100,line width=0.7 pt] (-9pt,-6pt) -- (9pt,-6pt);
    \end{tikzpicture} \right \rangle + \dots,
\\
    \ket{\Psi_{10}} \propto \sum_{g\in L_{10}} \ket{g} = \left|  \begin{tikzpicture}
        [baseline={([yshift=-.5ex]current bounding box.center)}],
        \draw[color=black!100,line width=0.5 pt] (-9pt,3pt)-- (9pt,3pt);
         \draw[color=black!100,line width=0.5 pt] (-9pt,-3pt) -- (9pt,-3pt);
         \draw[color=black!100,line width=0.5 pt] (3pt,-9pt)-- (3pt,9pt);
         \draw[color=black!100,line width=0.5 pt] (-3pt,-9pt)-- (-3pt,9pt);
         \draw[color=red!100,line width=0.7 pt] (6pt,-9pt)-- (6pt,9pt);
    \end{tikzpicture} \right \rangle+\left|  \begin{tikzpicture}
        [baseline={([yshift=-.5ex]current bounding box.center)}],
        \draw[color=black!100,line width=0.5 pt] (-9pt,3pt)-- (9pt,3pt);
         \draw[color=black!100,line width=0.5 pt] (-9pt,-3pt) -- (9pt,-3pt);
         \draw[color=black!100,line width=0.5 pt] (3pt,-9pt)-- (3pt,9pt);
         \draw[color=black!100,line width=0.5 pt] (-3pt,-9pt)-- (-3pt,9pt);
         \draw[color=red!100,line width=.7pt] plot [smooth cycle, tension=.8] coordinates{(-3pt,0pt) (0pt,3pt)(-3pt,6pt)  (-6pt,3pt)};
         \draw[color=red!100,line width=0.7 pt] (6pt,-9pt)-- (6pt,9pt);
    \end{tikzpicture} \right \rangle+\left| 
    \begin{tikzpicture}
        [baseline={([yshift=-.5ex]current bounding box.center)}],
        \draw[color=black!100,line width=0.5 pt] (-9pt,3pt)-- (9pt,3pt);
         \draw[color=black!100,line width=0.5 pt] (-9pt,-3pt) -- (9pt,-3pt);
         \draw[color=black!100,line width=0.5 pt] (3pt,-9pt)-- (3pt,9pt);
         \draw[color=black!100,line width=0.5 pt] (-3pt,-9pt)-- (-3pt,9pt);
         \draw[color=red!100,line width=.7pt] plot [smooth, tension=.8] coordinates{(6pt,9pt) (6pt,5pt)(3pt,6pt)  (-3pt,6pt)  (-6pt,3pt) (-3pt,0pt) (3pt,0pt) (6pt,1pt)(6pt,-9pt)};
    \end{tikzpicture} \right \rangle + \dots,
\\
    \ket{\Psi_{11}} \propto \sum_{g\in L_{11}} \ket{g} = \left|  \begin{tikzpicture}
        [baseline={([yshift=-.5ex]current bounding box.center)}],
        \draw[color=black!100,line width=0.5 pt] (-9pt,3pt)-- (9pt,3pt);
         \draw[color=black!100,line width=0.5 pt] (-9pt,-3pt) -- (9pt,-3pt);
         \draw[color=black!100,line width=0.5 pt] (3pt,-9pt)-- (3pt,9pt);
         \draw[color=black!100,line width=0.5 pt] (-3pt,-9pt)-- (-3pt,9pt);
         \draw[color=red!100,line width=0.7 pt] (6pt,-9pt)-- (6pt,9pt);
         \draw[color=red!100,line width=0.7 pt] (-9pt,-6pt) -- (9pt,-6pt);
    \end{tikzpicture} \right \rangle+\left|  \begin{tikzpicture}
        [baseline={([yshift=-.5ex]current bounding box.center)}],
        \draw[color=black!100,line width=0.5 pt] (-9pt,3pt)-- (9pt,3pt);
         \draw[color=black!100,line width=0.5 pt] (-9pt,-3pt) -- (9pt,-3pt);
         \draw[color=black!100,line width=0.5 pt] (3pt,-9pt)-- (3pt,9pt);
         \draw[color=black!100,line width=0.5 pt] (-3pt,-9pt)-- (-3pt,9pt);
         \draw[color=red!100,line width=.7pt] plot [smooth cycle, tension=.8] coordinates{(-3pt,0pt) (0pt,3pt)(-3pt,6pt)  (-6pt,3pt)};
         \draw[color=red!100,line width=0.7 pt] (6pt,-9pt)-- (6pt,9pt);
         \draw[color=red!100,line width=0.7 pt] (-9pt,-6pt) -- (9pt,-6pt);
    \end{tikzpicture} \right \rangle+\left| 
    \begin{tikzpicture}
        [baseline={([yshift=-.5ex]current bounding box.center)}],
        \draw[color=black!100,line width=0.5 pt] (-9pt,3pt)-- (9pt,3pt);
         \draw[color=black!100,line width=0.5 pt] (-9pt,-3pt) -- (9pt,-3pt);
         \draw[color=black!100,line width=0.5 pt] (3pt,-9pt)-- (3pt,9pt);
         \draw[color=black!100,line width=0.5 pt] (-3pt,-9pt)-- (-3pt,9pt);
         \draw[color=red!100,line width=.7pt] plot [smooth, tension=.8] coordinates{(6pt,9pt) (6pt,5pt)(3pt,6pt)  (-3pt,6pt)  (-6pt,3pt) (-3pt,0pt) (3pt,0pt) (6pt,1pt)(6pt,-9pt)};
         \draw[color=red!100,line width=0.7 pt] (-9pt,-6pt) -- (9pt,-6pt);
    \end{tikzpicture} \right \rangle + \dots.
\end{aligned} \end{align}
Here the black grid denote the square lattice (with qubits along each edge), the red lines/loops denote the loop gas on the dual lattice signifying the spin state of each bond.

\begin{figure}
	\centering
	\includegraphics[width = \linewidth]{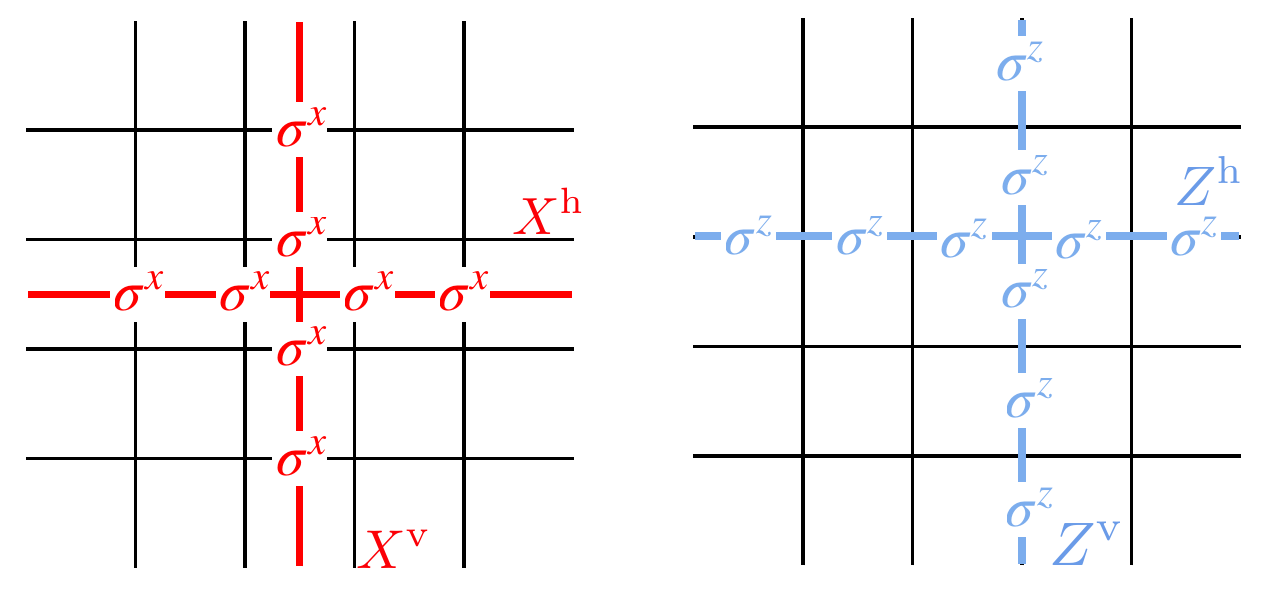}
	\caption{The basis of the Wilson loop operators of the TC model.
$X$-type strings transports $m$ (magnetic) anyons, $Z$-type strings transports $e$ (electric) anyons.}
	\label{fig:TC_Wilson}
\end{figure}
The Wilson loop algebra is generated by $Z$-strings and $X$-strings, shown in Fig.~\ref{fig:TC_Wilson}.
The $Z^{\Th/\Tv}$ operator is a product of $\sigma^z$ along a loop in the (black) square lattice, wrapping around the torus along $\Th/\Tv$ direction.
The $X^{\Th/\Tv}$ operator is a product of $\sigma^x$ along a loop in the \emph{dual} lattice, again wrapping around the torus along $\Th/\Tv$ direction.
The ground states basis $\ket{\Psi_{st}}$ are all eigenstates of the $Z^\Th$ and $Z^\Tv$ with eigenvalues $(-1)^s$ and $(-1)^t$ respectively, while the $X^\bullet$ operators permutes between the different basis states.

\subsection{Incoherent gauge model for \texorpdfstring{$\Z_2$}{Z2}}
To construct a CTO state, we consider the incoherent sum of the loop configuration:
\begin{align} \label{eq:Z2_incoherent_rho} \begin{aligned}
	\rho_{00} &\propto \sum_{g\in L_{00}} \ketbra{g}{g},
&	\rho_{01} &\propto \sum_{g\in L_{01}} \ketbra{g}{g},
\\	\rho_{10} &\propto \sum_{g\in L_{10}} \ketbra{g}{g},
&	\rho_{11} &\propto \sum_{g\in L_{11}} \ketbra{g}{g}.
\end{aligned} \end{align}
(Each state is normalized to have unit trace.)
We want to emphasize that $\rho_{00} \neq \ketbra{\Psi_{00}}{\Psi_{00}}$.

The states are all stabilized by plaquette operators $B_\square$ from Eq.~\eqref{eq:toric_code_stab},
	but not the vertex operators.
In fact, the states are locally indistinguishable from $\rho_e = \lim_{\beta\to+\infty} \prod_\square \exp\bigl[-\beta B_\square\bigr]$;
	the ensemble where vertex stabilizers $A_+$ are at infinite temperature while the plaquette stabilizers $B_\square$ are at zero temperature.

States~\eqref{eq:Z2_incoherent_rho} are also locally indistinguishable from their linear combinations
\begin{align} \label{eq:Z2_incoherent_rho_sum}
	\rho = \sum_{l,m \in \{0,1\}} q_{l,m} \rho_{lm}
\end{align}
subject to normalization condition $\tr \rho = 1$.
Unlike in the pure state case, states $\rho$ must be positive semidefinite and so the coefficients $q_{l,m}$ must be non-negative.%
	\footnote{The states of~\eqref{eq:Z2_incoherent_rho} are (proportional to) orthogonal projectors.}
We assert that states of the form~\eqref{eq:Z2_incoherent_rho_sum} constitute all possible states that are $n$-locally indistinguishable from $\rho_{lm}$'s,
	that is,
\begin{align}
	\DMxS{n}(\rho_{00}) = \Bigl\{ \text{Eq.~\eqref{eq:Z2_incoherent_rho_sum}} \,\Big|\, 0 \leq q_{l,m}, \,{\textstyle\sum}_{l,m} q_{l,m} = 1 \Bigr\} .
\end{align}
$\DMxS{n}$ is a convex set diffeomorphic to the 3-simplex, with the $\rho_{lm}$'s as its extreme points.

We now we examine the Wilson loop operators of this model.
In the pure state toric code, both $Z$-strings and $X$-strings (see Fig.~\ref{fig:TC_Wilson}) act as non-identity operators on the space of states.
For the mixed states in this section, we find that these operators behave differently than in the pure case.

The vector space of mixed states $\DMxV{n}$ is spanned by $\rho_{lm}$:
\begin{align}
	\DMxV{n} &= \bigl\{ q_{00}\rho_{00} + q_{01}\rho_{01} + q_{10}\rho_{10} + q_{11}\rho_{11} \,\big|\, q_{lm} \in \mathbb{C} \bigr\} .
\end{align}
Denote $\mathcal{H}(L_{lm}) \bigl\{ \ket{g} \mathrel{|} g \in L_{lm} \bigr\}$ as the Hilbert space spanned by loop configurations in sector $L_{lm}$.
Observe that $\rho_{lm} \in \mathcal{H}(L_{lm}) \times \mathcal{H}(L_{lm})^\ast$; the bra and ket of each $\rho$ must come from the same sector.
The $X$-strings permutes states between different sectors, for example, 
\begin{align}
	X^\Th \rho_{00} \in \mathcal{H}(L_{01}) \times \mathcal{H}(L_{00})^\ast .
\end{align}
So we see that $X^\Th \rho_{00}$ cannot be written as a sum of $\rho_{lm}$'s; $X^\Th \rho_{00} \notin \DMxV{n}$.
Similar arguments show that $X^\Tv$ and $X^{\Th+\Tv} = X^\Th X^\Tv$ are also not operators on the space of density matrices.%
	\footnote{One may notice that $\DMxV{n}(\rho)$ is closed when conjugated by $X^\Ts$, e.g., $X^\Th \rho_{l,m} X^\Th = \rho_{l,1-m}$.  We can think of $\{X^\Ts\}$ as elements of the automorphism group of $\DMxV{n}(\rho)$; transforming between classical states.}

On the other hand the $Z$-strings act on the states
\begin{align} \begin{aligned}
	Z^\Th \rho_{lm} &= (-1)^l\rho_{lm} \,,
\\	Z^\Tv \rho_{lm} &= (-1)^m\rho_{lm} \,, 
\\	Z^{\Th+\Tv} \rho_{lm} &= (-1)^{l+m}\rho_{lm} \,.
\end{aligned} \end{align} 
($Z^{\Th+\Tv}$ denote $Z^\Th Z^\Tv$.)
We can see that $\DMxV{n}$ is closed under left (and right) multiplication of $Z$-string operators.
Furthermore, it is easy to check that $Z^\Ts \rho_{lm} = \rho_{lm} Z^\Ts$, and therefore, this is a CTO phase.
(This state is in fact an example of Case~2 discussed in future sections.)

\subsection{Non-Abelian generalization}

For the non-Abelian generalization, it is more convenient to describe the state via discrete gauge theory.
The model is defined for a finite group $\mathcal{G}$ on a ``planar'' graph (e.g. a triangulation) on an orientable 2D surface.
Our set up parallels Kitaev's non-Abelian toric code~\cite{KITAEV20032}, where states are given by sums of gauge configurations.
Here instead, we will be considering incoherent sums to define the mixed state.

The technical description is as follows.
On the planar graph, each edge has $|\mathcal{G}|$ degrees of freedom, spanned by $\ket{g}$ for $g \in \mathcal{G}$.
Each edge is labeled with a direction towards one of its endpoints.
(Reversing the direction of an edge is equivalent to a change of basis $\ket{g} \leftrightarrow \ket{g^{-1}}$.)

A path $\gPath$ consists of sequence of adjacent vertices, traversing a sequence of edges $(b_1,b_2,\dots,b_{|\gPath|})$.
(Denote $|\gPath|$ as the number of edges in $\gPath$.)
For each edge along $\gPath$, let $\sigma_\gPath(e) \in \{\pm1\}$ depending on its direction relative to the path:
	$\sigma_\gPath(e) = +1$ if the path and edge direction are the same, and $\sigma_\gPath(e) = -1$ if the two directions are opposing.
Denote $\bar{\gPath}$ as the reverse path of $\gPath$, consisting of edges $(b_{|\gPath|}, b_{|\gPath|-1}, \dots, b_1)$.
Evidently, $\sigma_{\bar{\gPath}}(e) = -\sigma_{\gPath}(e)$.
A path whose starting and ending point coincide is call a loop.
We say a loop is contractible if it can be deformed into a null loop (single point) when embed in the torus.

\begin{figure}
	\centering
	\includegraphics[width = 0.9\linewidth]{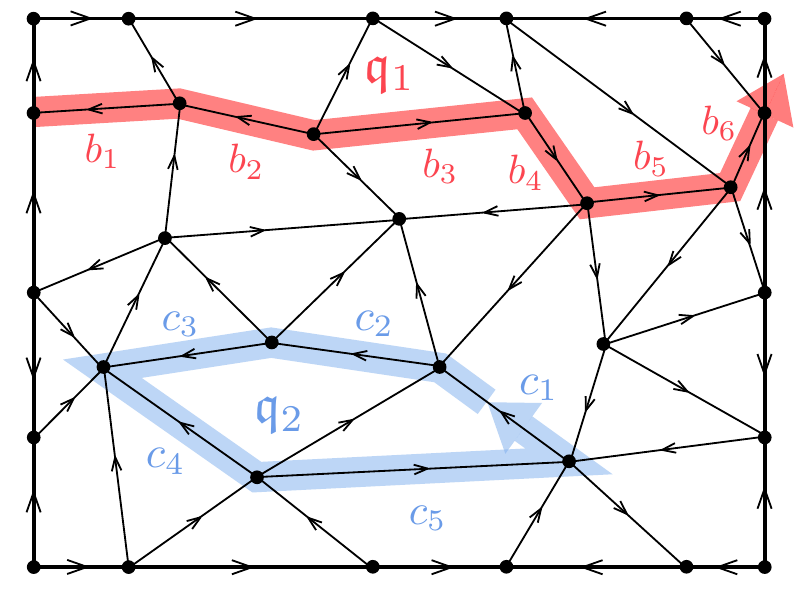}
	\caption{%
An example of a ``planar'' graph on the torus with directed edges.
Illustrated here are $\gPath_1 = (b_1,b_2,b_3,b_4,b_5,b_6)$ and $\gPath_2 = (c_1,c_2,c_3,c_4,c_5)$ as examples of loops, the latter of which is contractible.
A gauge configuration $\alpha$ is an assignment of a group element on each edge of the graph,
with the requirement that $\gPath(\alpha) = 1$ (product of group elements along the loop) for every contractible loop $\gPath$.
For example, $\gPath_2(\alpha) = \alpha(c_5)\alpha(c_4)^{-1}\alpha(c_3)\alpha(c_2)\alpha(c_1) = 1$.
	} \label{fig:gauge_model_graph}
\end{figure}

Let $\alpha$ be a map from the set of edges to $\mathcal{G}$ (i.e., an assignment of group element on every edge),
	and $\ket{\alpha}$ be the quantum product state where edge $e$ takes on state $\ket{\alpha(e)}$.
Given such map $\alpha$ and a path $\gPath$, we can define $\gPath(\alpha) \in \mathcal{G}$ as the product
\begin{align}
	\gPath(\alpha) \defeq \alpha(b_{|\gPath|})^{\sigma_\gPath(b_{|\gPath|})} \cdots \alpha(b_2)^{\sigma_\gPath(b_2)} \alpha(b_1)^{\sigma_\gPath(b_1)} \,.
\end{align}
That is, $\gPath(\alpha)$ is the product of group elements along path $\gPath$, taking into account of the direction of each edge relative to the path.
We have the relation $\bar{\gPath}(\alpha) = \big(\gPath(\alpha)\big)^{-1}$.

We call $\alpha$ a \textbf{gauge configuration} if for every contractible loop $\gPath$, $\gPath(\alpha) = 1$ is identity.
(See Fig.~\ref{fig:gauge_model_graph}.)

Let $\Gamma$ be the set of all gauge configurations.
Then the sum $\sum_\alpha \ket{\alpha}$ yields the $\mathcal{G}$-toric code~\cite{KITAEV20032}, a non-Abelian (pure) state,
whose topological quantum field theory (TQFT) is given by $\mathcal{Z}(\mathcal{G})$, the quantum double of $\mathcal{G}$.

Our mixed state is simply an incoherent sum of gauge configurations
\begin{align}
	\rho = \frac{1}{|\Gamma|} \sum_{\alpha \in \Gamma} \ketbra{\alpha}{\alpha} .
\end{align}
The mixed state can be constructed by dephasing the $\mathcal{G}$-toric code in the ``computational'' basis $\{\ket{g}|g\in\mathcal{G}\}$.

We now construct a basis of states on the torus.
Pick an origin, and let $\Th$ and $\Tv$ to be two loops that begin and end at the origin, forming a basis of loops on the torus.
For $g_1, g_2 \in \mathcal{G}$ a pair of commuting group elements,
let $\Gamma_0[g,h] \subseteq \Gamma$ be the set of gauge configurations such that its holonomy around the two fundamental cycles $\Th$ and $\Tv$ are $g_1$ and $g_2$ respectively.%
	\footnote{If $g_1, g_2$ do not commute, then $\Gamma_0[g_1,g_2]$ is empty; it is impossible for the gauge holonomy around two cycles of the torus to not commute.
	This is because $\pi_1(T^2) = \Z^2$ is Abelian; the loops $\Tl_1\Tl_2$ is homotopic to $\Tl_2\Tl_1$ for any pair of loops (with same base point).}
Formally,
\begin{subequations} \begin{align}
	\Gamma_0[g_1,g_2] &= \bigl\{ \alpha \in \Gamma \,\big|\, \Th(\alpha) = g_1 ,\, \Tv(\alpha) = g_2 \bigr\} ,
\\	\Gamma[g_1,g_2] &= \bigcup_{h \in \mathcal{G}} \Gamma_0\bigl[ hg_1h^{-1}, hg_2h^{-1} \bigr] .
\end{align} \end{subequations}
Finally, let
\begin{align}
	\rho[g_1,g_2] \propto \sum_{\alpha \in \Gamma[g_1,g_2]} \ketbra{\alpha}{\alpha}.
\end{align}

From the definition of $\Gamma[g_1,g_2]$, we see that $\rho[g_1,g_2] = \rho[hg_1h^{-1},hg_2h^{-1}]$ for any $h \in \mathcal{G}$.
Hence the distinct $\rho[g_1,g_2]$ is the number of commuting pairs modulo conjugacy equivalence.
Unsurprisingly, this is in fact the same as $|\mathcal{Z}(\mathcal{G})|$, the number of anyons in $\mathcal{Z}(\mathcal{G})$~\cite{beigi2011quantumDoubleCondensation}, which is the Hilbert space dimension of the $\mathcal{G}$-toric code~\cite{KITAEV20032}.

Each of the $\rho[g_1,g_2]$ are locally indistinguishable from $\rho$, and hence also locally indistinguishable with each other.
Since the set of distinct $\rho[g_1,g_2]$ are proportional to orthogonal projectors, the space of density matrices $\DMxS{n}(\rho)$ are simply nonnegative sums of $\rho[g_1,g_2]$.
Thus, $\DMxS{n}(\rho)$ is a simplex $\Delta^{|\mathcal{Z}(\mathcal{G})|-1}$, its extreme points are the set of $\rho[g_1,g_2]$.
From our geometric definition, this is a classically ordered phase.

From the Wilson loop point of view, the extreme points $\rho[g_1,g_2]$ can be distinguished by operators that ``measures'' $\Th(\alpha)$ and $\Tv(\alpha)$.
For instance, for a simple noncontractible loop $\Ts$, the following is a non-identity operator on the mixed states,
\begin{align}
	Y_\Ts(g) \ket{\alpha} = \ket{\alpha} \begin{cases} 1 & \Ts(\alpha) \equiv g, \\ 0 & \text{otherwise}. \end{cases}
\end{align}
Here $g \equiv h$ means $g$ and $h$ are in the same conjugacy class.
Clearly, $Y_\Ts(g)$ commutes with each $\rho[g_1,g_2]$.

\section{Phases induced from quantum channels \& quantum error correction}
\label{sec:boundarySPT}

\begin{figure*}
    \centering
    \includegraphics[width = 0.98\textwidth]{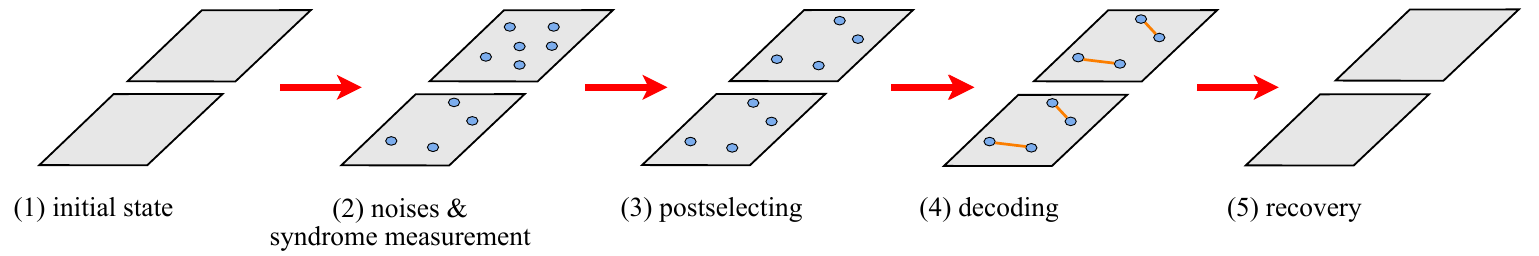}
    \caption{A schematic of the $n$-replica quantum error correction, illustrated here for $n=2$.
    Step (1): Encode quantum information in two copies of TC ground state.
    Step (2): Introduce some quantum channel (e.g. decoherence) independently on each layer, then measure the stabilizers to determine the location of the errors.
    Step (3): Postselect states that have the same error configuration for both copies.
    Step (4): Decoding the error syndrome by matching up anyons.
    Step (5): Apply recovery map reverting back to a pure TC ground state.
    }
    \label{fig:QEC}
\end{figure*}

In this section we study mixed states in the context of topological quantum error correction.
In a topological QEC code, the code space are the ground states of some topologically-ordered Hamiltonian comprised of local terms.
When the topological code begin to decoheres, error syndromes appear as violations of local Hamiltonian terms, indicative of pairs of anyons being created and moved around.
In principle, if the anyon pair density remains low, then the original ground state can be recovered and the encoded quantum information are preserved.
In the topological mixed state formalism that we have establish, we associate $\DMxS{n}$ as the logical code space.
We view the noisy system below the error threshold as remaining in the same topological phase as that of the pure topological QEC, and their code space are isomorphic.
As the system continues to decohere, the anyon density reaches a critical threshold and the quantum information is no longer recoverable.
We seek to understand what happens to the code space across this transition and its effect on the encoded quantum information.

Two factors determine the structure of a mixed state space: $n$-replica local indistinguishability and $n$-distance (defining global indistinguishability).
We expect local quantum channels to preserve local indistinguishability, i.e., locally indistinguishable states remain locally indistinguishable after decoherence.  However, as the decoherence error surpasses the threshold, the  initially distinct states may no longer be globally distinguishable;
the $n$-distance between states to reduces to zero across the transition.
When this happens, states become identified with each other resulting in a structural change of the logical code space, such that the noisy mixed-state code space is no longer isomorphic to the original code space.
The goal of this section is to develop a method utilizing tensor networks to diagnose structural changes to the code space.

For a concrete realization of these concepts, we consider the postselection-based QEC protocol on the toric code shown in Fig.~\ref{fig:QEC}.
As the $n$-Schatten norm~\eqref{eq:Schatten_norm} involves $n$ copies of a state: $\lVert \rho \rVert_n \propto (\Tr \rho^n)^{1/n}$, we will need to work with $n$ copies of the QEC code.
We start by preparing $n$ copies of the toric code pure state on the torus as our initial state.
\begin{align}
    \rho_\text{init} = \underbrace{ \rho_1 \otimes \rho_1 \otimes \cdots \otimes \rho_1 }_\text{$n$ copies} .
\end{align}
Each copy encodes two qubit of information in the basis of $\ket{\Psi_{00}}$, $\ket{\Psi_{01}}$, etc.
(Each copy may encode a different state, but for notational simplicity we assume the copies start off identical and unentangled with one another.)
Each individual copy undergo identical quantum channels $\E$, composed of local channels: $\E = \bigotimes_\mathbf{r} \epsilon_\mathbf{r}$ where $\mathbf{r}$ denote physical qubits.%
    \footnote{More generally, $\E$ can be any finite-depth local quantum circuit comprised of quantum channels with bounded range.}
(For instance, $\epsilon$ may be the on-site amplitude damping or dephasing channel.)

We then perform error syndrome measurements---which are the plaquette and star stabilizers~\eqref{eq:toric_code_stab}---on each copy $\E(\rho_1)$.
Viewed as a quantum channel $\M$, the measurement process transforms the state such that it becomes block diagonal in the syndrome measurement basis.
Let $\theta$ denote the possible results of these error syndrome measurements.
On a single copy,
\begin{align}
    \mathcal{E}(\rho_1) &\mapsto \mathcal{M}\E(\rho_1) = \sum_{\theta} \Prob(\theta) \, \rho(\theta),
\end{align}
where $\text{Pr}(\theta)$ is the probability that the error configuration $\theta$ is detected, and $\rho(\theta)$ is the resulting density matrix measurement when $\theta$ is detected.

Although the copies began as identical states, the outcome of the measurements may differ.
We postselect outcomes for which the error syndromes are identical among all $n$ copies, illustrated in Fig.~\ref{fig:QEC} step (3).
The postselected quantum state is
\begin{align}
    [\M\E(\rho_1)]^{\otimes n} \mapsto
    \mathcal{S}\bigl( [\M\E(\rho_1)]^{\otimes n} \bigr)
    = \frac{\sum_{\theta} \Prob(\theta)^n\rho (\theta) ^{\otimes n }}{\sum_{\theta} \displaystyle\Prob(\theta)^n} \,.
\end{align}
The denominator of the expression---i.e., the probability that the postselection succeeds---is in fact equals to the $n$-Schatten norm of the state $\M\E(\rho_1)$ (assuming $\rho_1$ began as a pure state).
Indeed, since the set of $\rho(\theta)$ are mutually orthogonal, the trace $\Tr\bigl[ [ \M\E(\rho_1) ]^n \bigr]$ filters out terms where the syndrome differs among copies.

Following the postselection, we attempt to deduce the most likely source of errors that lead to the observed syndrome.
For the toric code this means finding the correct pairing of anyons (syndromes that appears incorrect) to annihilate.
The recovery step $\R$ is a unitary operation which transforms each copy back into the toric code pure state.
The entire process is described by the sequence of operations
\begin{align}
    \rho_\text{init} \mapsto \rho_\text{f} = \R^{\otimes n} \mathcal{S}\bigl( (\M\E)^{\otimes n} \cdot \rho_\text{init} \bigr)
\end{align}
On each individual copy, this process can be characterized by a quantum channel $Q$, such that $Q^{\otimes n} \rho_\text{init} = \rho_\text{f}$.%
\footnote{This statement is not obvious, since postselection couples together different copies and is also not trace preserving.}
The measure of success of this protocol (as a QEC) is how close $Q$ is to the identity (or a unitary) map.

The primary clutch relating the structure of the code space to quantum channels is the use of Projected Entangled Pair States (PEPS)~\cite{verstraetePhysRevA2003Valence-bondQC,verstraete2004renormalization} and their generalization to operators.
PEPS is a type of tensor network useful both as a theoretical and numerical tool to classify and diagnose the topological order in two-dimensions~\cite{SCHUCH20102153PEPS,SCHUCH2011PEPS2,PhysRevB17Norbert,Norbert_anyonCondensation,PhysRevX21Norbert}.
One of the reason why PEPS is a powerful tool is because they encode the local indistinguishability conditions very naturally.
The PEPS formalism allows us to write down $\M\E(\rho_1)$ as a tensor network and evaluate its $n$\textsuperscript{th} power trace.

The main result of this section is follows.
We find that, starting from the $\Z_2$ toric code, the error state $\M\E(\rho_1)$ can be classified into one of five possible topological phases.
Associated with each of these phases is a quantum channel $Q$.
Among the five channels, one is a complete erasure, three takes the form of measurement projection, and one of them is a unitary map.
The latter map is indicative to a successful QEC process---corresponding to the regime where the error rate falls below the threshold.
Note that the initial state do not have to be identical among the copies or even unentangled; this protocol can encode and correct $2n$ qubits on the torus. 
(In following section we show that their corresponding phases can be identified as trivial/CTO/QTO.)
We also generalize this result to the $\Z_p$ toric code for arbitrary prime number $p$, noting the subtle (but significant) differences with when $p > 2$.

Before we delve into the explicit calculations, Sec.~\ref{sec:PEPS} briefly reviews the various tools of PEPS that we will employ, notably using the symmetry of the boundary state to classify topological orders.
Sec.~\ref{sec:QEC_quantum_channel} sets up the QEC protocol in detail, relating the quantum channel within the code space to the boundary SPT order, culminating in the key formula Eq.~\eqref{eq:xp_XXXX_Qq}.

In Sec.~\ref{sec:enum_boundary_SPT}, we prove Thm.~\ref{thm:TCdesc_classification}, that there are exactly $p+3$ possible classes of solutions that can arise from the $\Z_p$ toric code under decoherence, and compute their associated quantum channels.
This final subsection is purely mathematical, providing a complete classification of the possible descendant phases of TC, with the results summarized in Table~\ref{tab:TCcondensation}.

\subsection{PEPS and Symmetry}
\label{sec:PEPS}

\begin{figure*}
	\centering
	\includegraphics[width = 0.99\linewidth]{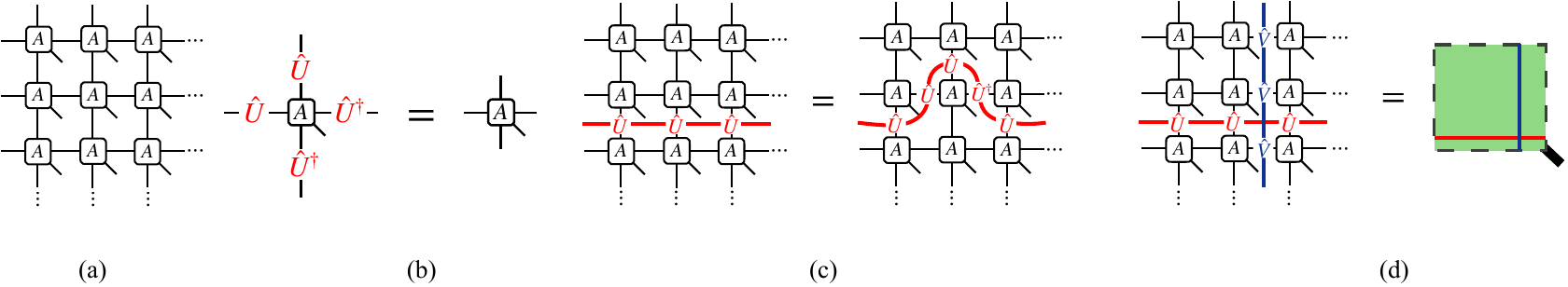}
	\caption{%
(a) An illustration of a 2D PEPS network on the square lattice, comprised of an array of 5-leg tensors.
(b) Each tensor $A$ admits an auxiliary symmetry $\mathcal{G}$; being invariant when conjugated by $\auxU$.
(c) Due to the auxiliary symmetry, a string operator $\prod \auxU$ can be topologically deformed within the PEPS network.
(d) The ground state on the torus $\ket{0}$ is constructed by tiling $A$ in a grid with periodic boundary conditions.
Strings of symmetry generators inserted along the horizontal (red) and vertical (blue) directions gives different ground states.
We use a green plane (with a thick line denoting all the physical degrees of freedom) to represent a layer of the pure state PEPS.
	} \label{fig:PEPS_sym}
\end{figure*}

In two dimensions, PEPS is a versatile class of tensor networks capable of capturing a variety of phases~\cite{verstraetePhysRevA2003Valence-bondQC,verstraete2004renormalization}.
Figure~\ref{fig:PEPS_sym}(a) shows a typical PEPS as a tensor network, where each individual tensor of the PEPS have a physical leg capturing the local degrees of freedom, as well as multiple auxiliary legs which are interconnected in a network.
In a PEPS, topological order is manifested as some form of symmetry acting on the auxiliary legs.
In particular, states realizing the topological order $\mathcal{Z}(\mathcal{G})$---where $\mathcal{Z}(\mathcal{G})$ is the quantum double of the finite group $\mathcal{G}$---can be written as a PEPS that admit local symmetries on auxiliary degrees of freedom~\cite{SCHUCH20102153PEPS,SCHUCH2011PEPS2}.
Figure~\ref{fig:PEPS_sym}(b) illustrates the action of such symmetry.
Each tensor is invariant when its auxiliary legs are conjugated by $\auxU$, where $\auxU$ is a representation of an element of the group $\mathcal{G}$.
As a result, a string these auxiliary symmetry operators can be passed freely through the network of tensors, as exemplified in Fig.~\ref{fig:PEPS_sym}(c).

Let's consider the toric code as our example where $\mathcal{G} = \Z_2$.
(The quantum double $\mathcal{Z}(\Z_2)$ is the toric code topological order.)
The toric code consists of two qubits per unit cell, and its ground states can be constructed as a PEPS via the unit cell tensor
\begin{align} \label{eq:toric_code_PEPS_A}
\begin{tikzpicture}
	[baseline={([yshift=-.5ex]current bounding box.center)},rnode1/.style={circle,inner sep=1pt, draw=orange!100,fill = white!100,line width=1.5pt,minimum size = 10pt}],
	\draw[color=black!100,line width=1.6pt] (-20pt,0pt) -- (25pt,0pt);
	\draw[color=black!100,line width=1.6pt] (0pt,-25pt) -- (0pt,15pt);
	\draw[color=black!100,line width=2pt] (0pt,0pt) -- (18pt,-18pt);
	\draw[color=black!100,fill=white!100,rounded corners=4pt,line width=1.5pt] (-10pt,-10pt) rectangle (10pt,10pt);
	\node[] at (0pt, 0pt) {$A$};
\end{tikzpicture}
\mkern10mu = \mkern10mu
\begin{tikzpicture}
[baseline={([yshift=-.5ex]current bounding box.center)},rnode1/.style={circle,inner sep=1pt, draw=orange!100,fill = white!100,line width=1.5pt,minimum size = 10pt,path picture={\draw[orange]
       (path picture bounding box.south) -- (path picture bounding box.north) 
       (path picture bounding box.west) -- (path picture bounding box.east);
      }}],
\draw[color=black!100,line width=1.6pt] (-20pt,0pt) -- (25pt,0pt);
\draw[color=black!100,line width=1.6pt] (0pt,-25pt) -- (0pt,15pt);
\draw[color=orange!100,line width=1.6pt] (12.5pt,0pt) -- (22.5pt,-10pt);
\draw[color=orange!100,line width=1.6pt] (0pt,-12.5pt) -- (10pt,-22.5pt);
\node[rnode1] (0) at (12.5pt,0pt) {};
\node[rnode1] (1) at (0pt,-12.5pt) {};
\end{tikzpicture} \ ,
\end{align}
where every leg (physical or auxiliary) has dimension 2.
The tensor is built from the XOR tensor
\begin{subequations} \begin{align}
\begin{tikzpicture}
	[baseline={([yshift=-.5ex]current bounding box.center)},rnode1/.style={circle,inner sep=1pt, draw=orange!100,fill = white!100,line width=2pt,minimum size = 20pt}],
	\node[rnode1] (0) at (0pt,0pt) {};
	\draw[color=orange!100,line width=2pt] (-20pt,0pt) -- (20pt,0pt);
	\draw[color=orange!100,line width=2pt] (0pt,10pt) -- (0pt,-20pt);
	\node[] () at (-18pt,6pt) {$i$};
	\node[] () at (18pt,6pt) {$j$};
	\node[] () at (6pt,-15pt) {$k$};
\end{tikzpicture} 
 = \begin{cases} 1 & i+j+k \equiv 0 \pmod{2}, \\ 0 & \text{otherwise}, \end{cases}
\end{align}
and the delta tensor
\begin{align}
\begin{tikzpicture}
	[baseline={([yshift=-.5ex]current bounding box.center)}],
	\draw[color=black!100,line width=2pt] (-20pt,0pt) -- (20pt,0pt);
	\draw[color=black!100,line width=2pt] (0pt,20pt) -- (0pt,-20pt);
	\node[] () at (-18pt,6pt) {$i$};
	\node[] () at (18pt,6pt) {$j$};
	\node[] () at (6pt,-15pt) {$k$};
	\node[] () at (6pt,15pt) {$l$};
\end{tikzpicture} 
 = \delta_{ij}\delta_{ik}\delta_{il} \,.
\end{align} \end{subequations}
The PEPS representation of the the toric code state admits the a $\Z_2$ symmetry on the auxiliary space.
Group elements $1, x \in \Z_2$ are represented by the identity and Pauli-$X$ operators respectively.
We can visually verify that the tensors~\eqref{eq:toric_code_PEPS_A} are invariant when conjugated by $X$.
Observe that 
\begin{subequations} \begin{align}
\begin{tikzpicture}
	[baseline={([yshift=-.5ex]current bounding box.center)},rnode1/.style={circle,inner sep=1pt, draw=orange!100,fill = white!100,line width=2pt,minimum size = 15pt},
	rnode2/.style={circle,inner sep=0pt, draw=white!100,fill = white!100,line width=0pt,minimum size = 10pt}],
	\node[rnode1] (0) at (0pt,0pt) {};
	\draw[color=orange!100,line width=2pt] (-30pt,0pt) -- (20pt,0pt);
	\draw[color=orange!100,line width=2pt] (0pt,8pt) -- (0pt,-20pt);
	\node[rnode2] (1) at (-20pt,1pt) {$\auxX$};
\end{tikzpicture}
\mkern12mu &= \mkern12mu
\begin{tikzpicture}
	[baseline={([yshift=-.5ex]current bounding box.center)},rnode1/.style={circle,inner sep=1pt, draw=orange!100,fill = white!100,line width=2pt,minimum size = 15pt},
	rnode2/.style={circle,inner sep=0pt, draw=white!100,fill = white!100,line width=0pt,minimum size = 10pt}],
	\node[rnode1] (0) at (0pt,0pt) {};
	\draw[color=orange!100,line width=2pt] (-20pt,0pt) -- (30pt,0pt);
	\draw[color=orange!100,line width=2pt] (0pt,8pt) -- (0pt,-20pt);
	\node[rnode2] (1) at (20pt,1pt) {$\auxX$};
\end{tikzpicture} \;,
\\
\begin{tikzpicture}
	[baseline={([yshift=-.5ex]current bounding box.center)},
	rnode2/.style={circle,inner sep=0pt, draw=white!100,fill = white!100,line width=0pt,minimum size = 10pt}],
	\draw[color=black!100,line width=2pt] (-30pt,0pt) -- (30pt,0pt);
	\draw[color=black!100,line width=2pt] (0pt,30pt) -- (0pt,-30pt);
	\node[rnode2] (1) at (-20pt,1pt) {$\auxX$};
	\node[rnode2] (1) at (0pt,-20pt) {$\auxX$};
	\node[rnode2] (1) at (0pt,20pt) {$\auxX$};
	\node[rnode2] (1) at (20pt,1pt) {$\auxX$};
\end{tikzpicture}
\mkern12mu &= \mkern12mu
\begin{tikzpicture}
	[baseline={([yshift=-.5ex]current bounding box.center)}],
	\draw[color=black!100,line width=2pt] (-20pt,0pt) -- (20pt,0pt);
	\draw[color=black!100,line width=2pt] (0pt,20pt) -- (0pt,-20pt);
\end{tikzpicture} \;.
\end{align} \end{subequations}
Thus, strings of $\auxX$'s cutting across the PEPS network can be topologically deformed as shown in Fig.~\ref{fig:PEPS_sym}(c), reminiscent of Wilson loop operators.

On the torus, the different ground states are constructed by tiling the unit tensors with various boundary conditions.
For a general Abelian group $\mathcal{G}$, the state $\ket{\Psi_{00}}$ results from periodic boundary conditions on the torus.
(Henceforth, we assume $\mathcal{G}$ is Abelian for simplicity.)
The remaining ground states $\ket{\Psi_{\bullet\bullet}}$ can be obtained by inserting non-contractible strings into the network of tensors, as shown in Fig.~\ref{fig:PEPS_sym}(d).
The fact that all ground states can be constructed from the same PEPS unit cell means that \textbf{these states are local indistinguishable}.

\begin{figure*}
    \centering
    \includegraphics[width = 0.9\linewidth]{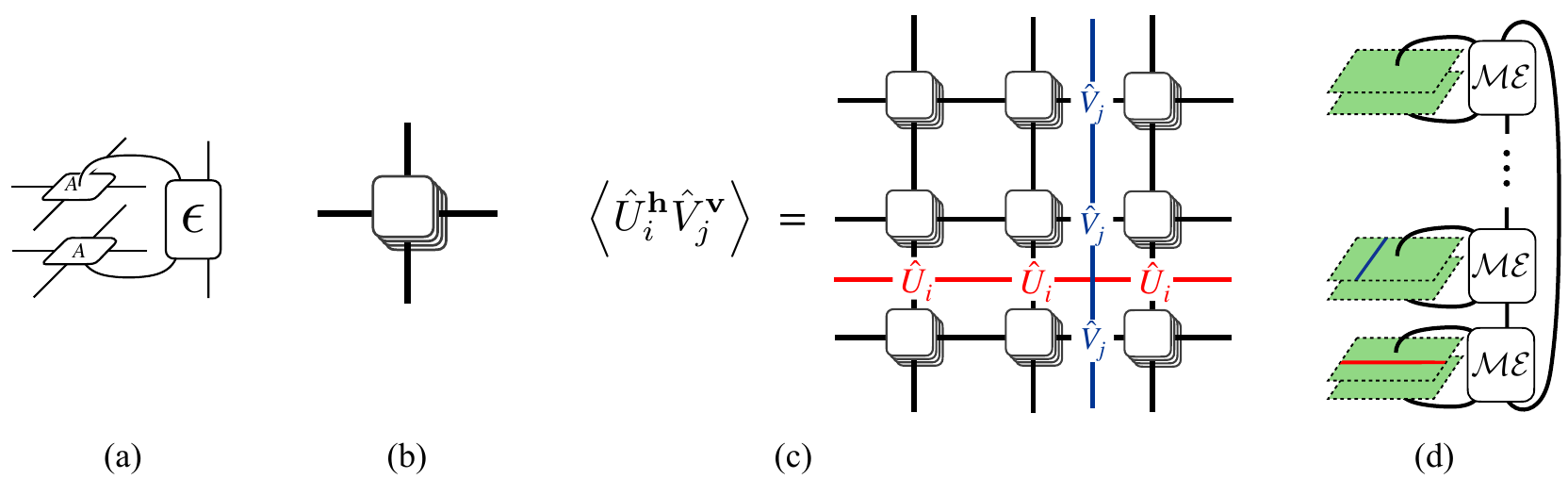}
    \caption{%
(a) The unit cell of the mixed state PEPS obtained by applying a channel $\epsilon$ on every individual site, comprised of two copies of $A$ and the $\epsilon$ channel written as a superoperator.
(b) Illustration of a unit cell of the $n$-replica network computing $\Tr \M\E(\ketbra{0}{0})^{n}$.
It is built from $n$ copies of (a) (along with the syndrome measurement not illustrated).
(c) The top view of an $n$\textsuperscript{th} trace tensor network with a $U$ string on the $i$\textsuperscript{th} layer and a $V$ string on the $j$\textsuperscript{th} layer.
(d) An illustration of the same tensor network viewed from the side.
Here we can see that the tensor network is comprised of $2n$ copies of the original pure state.
Each pair of pure states PEPS forms a density matrix, which is passed through the channels $\M\E$.  The $n$-Schatten norm (with auxiliary strings) is obtained by contracting $n$ copies of the density operator.
	} \label{fig:PEPS_channel}
\end{figure*}

The density matrix corresponding to the topological pure state is $\rho = \ketbra{\Psi}{\Psi}$, and its tensor network representation is simply two decoupled copies of the PEPS wavefunction.
Hence the density matrix has $\mathcal{G}^2$ auxiliary symmetry, a copy of $\mathcal{G}$ for the ket and bra.
Because the error channel $\E$ act locally on the physical degrees of freedom, they can be realized as a product of tensors connecting the ket and bra side of the density matrix.
The unit tensor for the mixed state $\E(\rho_1)$ is visualized in Fig.~\ref{fig:PEPS_channel}(a).
Notably, the maps $\E$ and $\M$ do not alter the symmetry of the auxiliary degrees of freedom;%
	\footnote{The syndrome measurements $\M$ do not act on a single site--and so they generally will increase the number of auxiliary degrees of freedom (bond dimension).}
	the mixed state inherits the $\mathcal{G}^2$ auxiliary symmetry from the pure state.

To compute the $n$-Schatten norm of our state, we contract $n$ copies of $\M\E(\rho_1)$ tensor network.
Hence the trace $\Tr \M\E(\rho_1)^{n}$ is characterized by the tensor network involving $n$ copies of $\M\E$ channels and $2n$ copies of the topological pure state, illustrated in Fig.~\ref{fig:PEPS_channel}(d).
Its unit cell---shown in Fig.~\ref{fig:PEPS_channel}(b)---has no remaining physical degree of freedom, but possesses $\mathcal{G}^{2n}$ symmetry acting on the auxiliary degrees of freedom.
The combined tensor network carries no physical degree of freedom---akin to a 2D statistical mechanics model---and can be evaluated into a number.

In order to discuss the role of auxiliary symmetry in the evaluation of $n$-replica tensor network, we need to establish some notation.
Suppose $\auxU$ is the representation of $u \in \mathcal{G}$ acting on (one single copy of) the auxiliary legs of the PEPS,
then we let $\auxU_i^\Th$ for $1 \leq i \leq 2n$ denote a horizontal string of $\auxU$'s acting on the $i$\textsuperscript{th} layer of the tensor network.
Similar, let $\auxU_i^\Tv$ denote a vertical string of $\auxU$'s on the $i$\textsuperscript{th} copy.
We denote
\begin{align}
	\Braket{ \hat{A}_{i_1}^\Th \hat{B}_{i_2}^\Tv \hat{C}_{i_3}^\Tv \hat{D}_{i_4}^\Th \cdots }
\end{align}
as the evaluation of the tensor network with a horizontal string of $\hat{A}$'s on layer $i_1$, vertical string of $\hat{B}$'s on layer $i_2$, etc.
The layer indices can repeat; there can be multiple strings on any layer.
For example, the evaluation of network shown in Fig.~\ref{fig:PEPS_channel}(c) is $\bigXp{\auxU_i^\Th \auxV_j^\Tv} = \bigXp{\auxV_j^\Tv \auxU_i^\Th}$,
the evaluation of the network with no auxiliary string is $\braket{1}$.

The trace tensor network can be evaluated via the standard transfer operator method.
For a torus with dimension $L_x \times L_y$, the transfer operator is a row of the tensor network---consisting $L_x$ unit cells, shown in Fig.~\ref{fig:PEPS_transferOp}(a)---whose $L_y$\textsuperscript{th} power trace is the tensor network evaluation.
Inserting a horizontal string along the auxiliary degrees of freedom:
\begin{align}
	\bigXp{\auxU_i^\Th} = \tr\Bigl[ \bigl(\hat{T}_0\bigr)^{\mkern-2mu L_y} \auxU_i^\Th \Bigr] .
\end{align}
Here $\hat{T}_0$ is the transfer operator, which acts on the \emph{auxiliary space}, a row of auxiliary legs of the tensor network.
In the auxiliary space, $\auxU_i^\Th$ can be written as an operator and commutes with $\hat{T}_0$, as shown in Fig.~\ref{fig:PEPS_transferOp}(b).
In the thermodynamic limit,
\begin{align}
	\bigXp{\auxU_i^\Th} = \sum_\mu b^{(T)\mu} \bigl(\hat{T}_0\bigr)^{\mkern-2mu L_y} \auxU_i^\Th \, b^{(B)\mu} ,
\end{align}
where $b^{(T/B)\mu}$ are the dominant top/bottom eigenvector pairs for the transfer operator $\hat{T}_0$.
Thus
\begin{align}
	\frac{\bigXp{ \auxU_{i_1}^\Th \auxU_{i_2}^\Th \cdots \auxU_{i_m}^\Th }}{\braket{1}}
	= \frac{1}{\#_\mu} \sum_\mu b^{(T)\mu} \auxU_{i_1}^\Th \cdots \auxU_{i_m}^\Th b^{(B)\mu} ,
\end{align}
where $\#_\mu = \sum_\mu b^{(T)\mu} b^{(B)\mu}$ is the number of dominant eigenvectors.
We call $b^{(T/B)\mu}$ the \textbf{boundary states} of the $n$-replica tensor network.

\begin{figure}
	\centering
	\includegraphics[width = \linewidth]{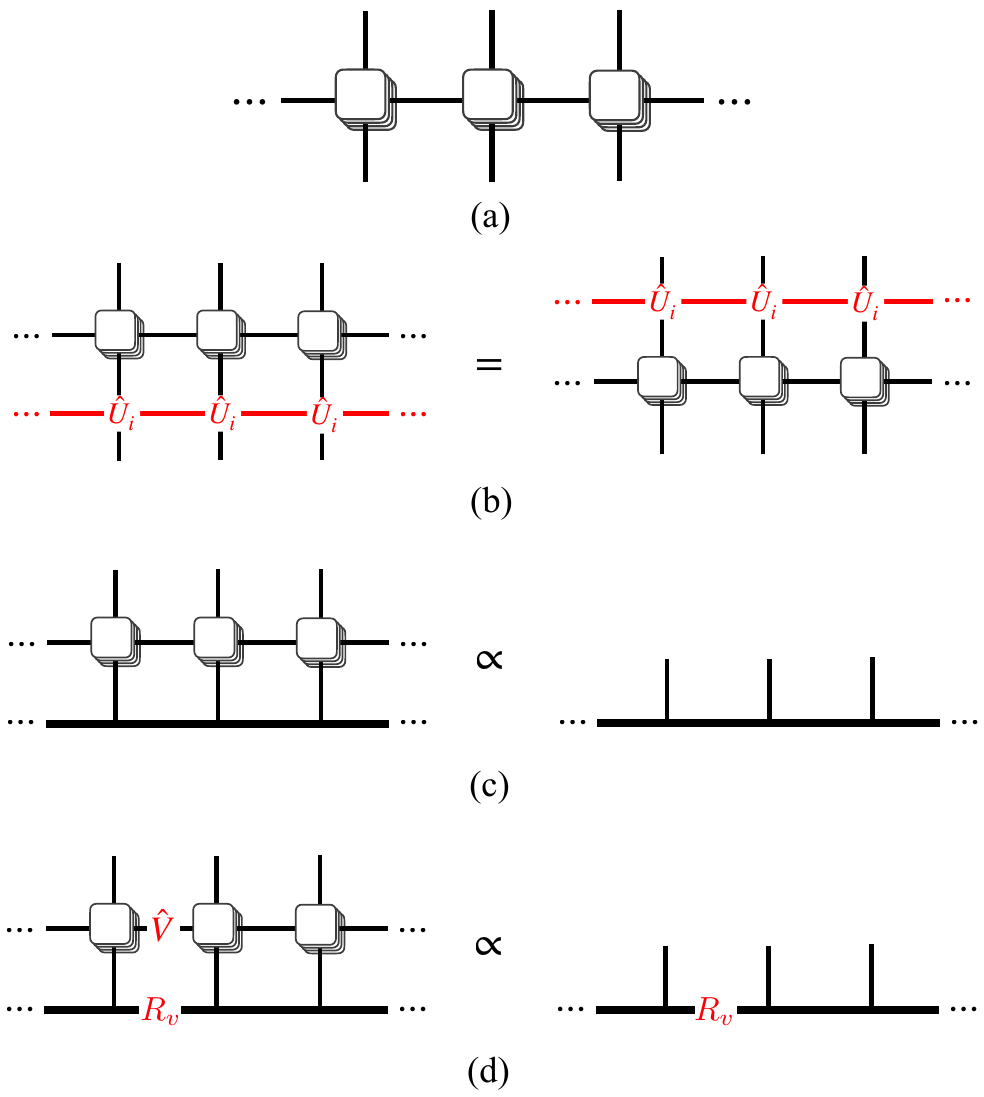}
	\caption{%
(a) The tensor network representation of the the transfer operator $\hat{T}_0$, a slice of the $n$-replica tensor network.
(b) Because auxiliary strings can be moved around freely, the transfer operator commutes with the string $\auxU_i^\Th$.
(c) The boundary state(s) of the tensor network are the dominant eigenvector(s) of the transfer operator.
(d) The boundary state(s) of the tensor network with $\auxV^\Tv$ string, which are related to boundary states of (c) by insertion of the projective representation element of $v$: $R_v$.
	} \label{fig:PEPS_transferOp}
\end{figure}

While the bulk tensors (and transfer operators) possess $\mathcal{G}^{2n}$ symmetry,
it may be spontaneously broken by the boundary states $b^{(T/B)\mu}$~\cite{PhysRevB17Norbert,Norbert_anyonCondensation, PhysRevX21Norbert}.
\textbf{Let $G$ denote the symmetry subgroup of the boundary states} (so $G \subseteq \mathcal{G}^{2n}$).
If the symmetry action of $u_i$ is unbroken, then the eigenvectors are invariant the action: $\auxU_i^\Th b^{(B)\mu} = e^{i\varphi(U_i)} b^{(B)\mu}$ (we can absorb the phase $e^{i\varphi(U_i)}$ into $\auxU_i^\Th$).
On the other hand, if the symmetry action of $u_i$ is spontaneously broken, then $\bigXp{\auxU_i^\Th}$ vanishes.\footnote{If $u \notin G$, then generically the eigenvectors $b^{(B)\mu}$ belong to the regular representation of $\mathcal{G}^{2n}/G$.}
Succinctly, the expectation values of the horizontal strings are
\begin{align}
	\frac{\bigXp{ \auxU_{i_1}^\Th \auxU_{i_2}^\Th \cdots \auxU_{i_m}^\Th }}{\braket{1}} = \begin{cases} 1 & u_{i_1} u_{i_2} \cdots u_{i_m} \in G , \\ 0 & \text{otherwise} . \end{cases}
\end{align}
For example, for the 1-replica tensor network (with a single ket and bra) of the pure state TC,
$x_1$ and $x_2$ are both individually spontaneously broken by the boundary state but $x_1^\phd x_2^{-1}$ remains preserved~\cite{Norbert_anyonCondensation}.

We can visualize the boundary state(s) $b^{(B)\mu}$ as a matrix product state (MPS).
When $u \in G$ ($u$ is a product of group generators $x_i$), then 
\begin{align}
	\begin{tikzpicture}
	[baseline={([yshift=-.5ex]current bounding box.center)}],
	\draw[color=red!70,line width=3pt] (-64pt,-10pt) -- (14pt,-10pt);
	\draw[color=black!100,line width=1.6pt] (-64pt,-25pt) -- (14pt,-25pt);
	\draw[color=black!100,line width=1.6pt] (-48pt,0pt) -- (-48pt,-25pt);
	\draw[color=black!100,line width=1.6pt] (-24pt,0pt) -- (-24pt,-25pt);
	\draw[color=black!100,line width=1.6pt] (0pt,0pt) -- (-0pt,-25pt);
    \node [circle, fill = white!100,minimum size=12pt] (0) at (-48pt,-10pt) {};
    \node [color = red] () at (-48pt,-10pt) {$\auxU$};
    \node [circle, fill = white!100,minimum size=12pt] (0) at (-24pt,-10pt) {};
    \node [color = red] () at (-24pt,-10pt) {$\auxU$};
    \node [circle, fill = white!100,minimum size=12pt] (0) at (0pt,-10pt) {};
    \node [color = red] () at (0pt,-10pt) {$\auxU$};
	\end{tikzpicture} 
\mkern12mu = \mkern12mu
	\begin{tikzpicture}
	[baseline={([yshift=1.3ex]current bounding box.center)}],
	\draw[color=black!100,line width=1.6pt] (-64pt,-25pt) -- (14pt,-25pt);
	\draw[color=black!100,line width=1.6pt] (-48pt,-15pt) -- (-48pt,-25pt);
	\draw[color=black!100,line width=1.6pt] (-24pt,-15pt) -- (-24pt,-25pt);
	\draw[color=black!100,line width=1.6pt] (0pt,-15pt) -- (-0pt,-25pt);
	\end{tikzpicture} \;.
\end{align}
In the usual contexts, the vertical legs of the MPS represent physical degrees of freedom while the horizontal links represent auxiliary legs capturing entanglement among sites.
Here instead, each of the vertical leg represents an auxiliary bond of the $n$-replica tensor network, the horizontal bonds of the MPS now captures the ``boundary state entanglement'' of the network,\footnote{The horizontal bonds of the MPS are the \emph{auxiliary of the auxiliary legs} of the 2D network, which is mouthful.} which we refer to as the ``boundary backbone''.

Since the boundary state can be viewed as a 1D system with $u$-symmetry, the action of $\auxU$ on any one site can also be pushed into the boundary backbone.
\begin{align} \label{eq:MPS_U_projective}
	\begin{tikzpicture}
	[baseline={([yshift=-.5ex]current bounding box.center)}],
	\draw[color=black!100,line width=1.6pt] (0pt,25pt) -- (-0pt,0pt);
	\draw[color=black!100,line width=1.6pt] (-20pt,0pt) -- (20pt,0pt);
	\node [circle, fill = white!100,minimum size=12pt] (0) at (0pt,15pt) {};
	\node[] () at (0pt,15pt) {$\auxU$};
	\end{tikzpicture}
\mkern12mu = \mkern12mu
	\begin{tikzpicture}
	[baseline={([yshift=1ex]current bounding box.center)}],
	\node[] (1) at (-20pt,0pt) {$R_{u}$};
	\node[] (2) at (20pt,0pt) {$R^{-1}_{u}$};
	\draw[color=black!100,line width=1.6pt] (-2pt,20pt) -- (-2pt,0pt);
	\draw[color=black!100,line width=1.6pt] (0.5pt,0pt) -- (1);
	\draw[color=black!100,line width=1.6pt] (-40pt,0pt) -- (1);
	\draw[color=black!100,line width=1.6pt] (-0.5pt,0pt) -- (2) ;
	\draw[color=black!100,line width=1.6pt] (40pt,0pt) -- (2) ;
	\end{tikzpicture},
\end{align}
where $u \mapsto R_{u}$ is a projective representation of $G$.
When $\auxU$ is applied on every site, the operators $R_{u}, R_{u}^{-1}$ appears in pairs and cancel.

In a projective representation, the operators $R_u$ obey group multiplication up to a multiplicative factor:
\begin{align}
	R_u R_v = \omega(u,v) R_{uv} \,,
\end{align}
for some set $\omega(u,v) \in \mathbb{C}$.
This makes projective representations a generalization of ordinary representations where $\omega = 1$.
Two projective representations are ``gauge equivalent'' if they are related by a ``gauge transformation'' $R_u \to \sigma(u) R_u$.
Notice that Eq.~\eqref{eq:MPS_U_projective} only determines the projective representation up to a gauge transformation.
Since the $\omega$'s (called 2-cocycles) are gauge-dependent, transforming as $\omega(u,v) \to \omega(u,v) \sigma(u) \sigma(v) / \sigma(uv)$, it is convenient to define
\begin{align}
	\chi(u,v) = \frac{\omega(u,v)}{\omega(v,u)}
\end{align}
which are gauge-invariant (for an Abelian group $G$).
$\chi$ is a bimultiplicative function: $\chi(u,vw) = \chi(u,v) \chi(u,w)$ and $\chi(uv,w) = \chi(u,w) \chi(v,w)$.

Projective representations of $G$ are classified by the group cohomology $H^2(G,\mathrm{U}(1))$.%
	\footnote{Specifically, $H^2(G,\mathrm{U}(1))$ characterizes the obstruction to finding a gauge transformation such that $\omega(u,v) = 1$ for all $u,v \in G$.}
The group $G$, along with the cohomology class $\in H^2(G,\mathrm{U}(1))$ that the (bottom) boundary state belongs to, constitute 
the \textbf{boundary SPT order} of the $n$-replica tensor network.

Now let's consider the expectation value with two intersecting string operators, (similar to that in Fig.~\ref{fig:PEPS_channel}(c), but omitting layer subscripts)
\begin{align}
	\bigXp{ \auxU^\Th \auxV^\Tv } .
\end{align}
Just as before, this ca be evaluated with the aid of transfer operators.
On a finite torus, $\bigXp{ \auxU^\Th \auxV^\Tv } = \tr\bigl[ \bigl(\hat{T}_V\bigr)^{L_y} U^\Th \bigr]$,
where $\hat{T}_V$ is the transfer operator with a single $V$ operator applied to one of its auxiliary legs.
Again, in the thermodynamic limit,
\begin{align}
	\frac{\bigXp{ \auxU^\Th \auxV^\Tv }}{\bigXp{ \auxV^\Tv }} = \frac{1}{\#_\mu} \sum_\mu b_V^{(T)\mu} \, \auxU^\Th \, b_V^{(B)\mu} ,
\end{align}
where $b_V^{(T/B)\mu}$ are the dominant top/bottom eigenvector pairs for the transfer operator $\hat{T}_V$.
Remarkably, the dominant eigenvector(s) of $\hat{T}_V$ are related to those of $\hat{T}_0$ by the insertion of a projective representation element $R_v$ into the backbone of the MPS, illustrated in Fig.~\ref{fig:PEPS_transferOp}(d).
Finally, since $v \in G$, we can replace $\bigXp{\auxV^\Tv} = \braket{1}$.

Putting it all together: for $u, v \in G$, the expectation value of a pair of intersecting strings is given by 
\begin{align} \begin{aligned}
&	\bigXp{ \auxU^\Th \auxV^\Tv } / \braket{1}
\\ &\; = \mkern12mu
\begin{tikzpicture}
	[baseline={([yshift=-.5ex]current bounding box.center)}],
	\node[] (1) at (-36pt, 20pt) {$R_v^{-\mathrm{T}}$};
	\draw[color=black!100,line width=1.6pt] (-64pt,20pt) -- (-45pt, 20pt);
	\draw[color=black!100,line width=1.6pt] (-25pt, 20pt) -- (25pt,20pt);
	\draw[color=black!100,line width=1.6pt] (-53pt,-20pt) -- (-53pt,20pt);
	\draw[color=black!100,line width=1.6pt] (-19pt,-20pt) -- (-19pt,20pt);
	\draw[color=black!100,line width=1.6pt] (15pt,-20pt) -- (15pt,20pt);
	\draw[color=red!70,line width=3pt] (-64pt,0pt) -- (25pt,0pt);
	\node [circle, fill = white!100,minimum size = 12pt] (0) at (-53pt,0pt) {};
	\node [color = red] () at (-53pt,0pt) {$\auxU$};
	\node [circle, fill = white!100,minimum size = 12pt] (0) at (-19pt,0pt) {};
	\node [color = red] () at (-19pt,0pt) {$\auxU$};
	\node [circle, fill = white!100,minimum size = 12pt] (0) at (15pt,0pt) {};
	\node [color = red] () at (15pt,0pt) {$\auxU$};
	\node[] (1) at (-36pt, -20pt) {$R_v$};
	\draw[color=black!100,line width=1.6pt] (-64pt,-20pt) -- (-43pt, -20pt);
	\draw[color=black!100,line width=1.6pt] (-30pt, -20pt) -- (25pt,-20pt);
\end{tikzpicture}
\begin{tikzpicture}
	[baseline={([yshift=0ex]current bounding box.center)}],
	\draw[color=black!100,line width=1.6pt] (7pt,23pt) -- (-7pt,-23pt);
\end{tikzpicture}
\begin{tikzpicture}
	[baseline={([yshift=1ex]current bounding box.center)}],
	\draw[color=black!100,line width=1.6pt] (-64pt,20pt) --  (14pt,20pt);
	\draw[color=black!100,line width=1.6pt] (-53pt,10pt) -- (-53pt,20pt);
	\draw[color=black!100,line width=1.6pt] (-19pt,10pt) -- (-19pt,20pt);
	\draw[color=black!100,line width=1.6pt] (0pt,10pt) -- (-0pt,20pt);
	\draw[color=black!100,line width=1.6pt] (-64pt,-10pt) -- (14pt,-10pt);
	\draw[color=black!100,line width=1.6pt] (-53pt,10pt) -- (-53pt,-10pt);
	\draw[color=black!100,line width=1.6pt] (-19pt,10pt) -- (-19pt,-10pt);
	\draw[color=black!100,line width=1.6pt] (0pt,10pt) -- (-0pt,-10pt);
	\node[circle, fill = white!100,minimum size = 19pt] () at (-36pt, 20pt) {};
	\node[circle, fill = white!100,minimum size = 12pt] () at (-36pt, -10pt) {};
	\node[] () at (-36pt, 20pt) {$R_v^{-\mathrm{T}}$};
	\node[] () at (-36pt, -10pt) {$R_v$};
\end{tikzpicture}
\\ &\; = \mkern12mu
\begin{tikzpicture}
	[baseline={([yshift=-.5ex]current bounding box.center)}],
	\node[] (1) at (-25pt, 20pt) {$R_v^{-\mathrm{T}}$};
	\draw[color=black!100,line width=1.6pt] (-64pt,20pt) -- (-40pt, 20pt);
	\draw[color=black!100,line width=1.6pt] (-10pt, 20pt) -- (10pt,20pt);
	\node[] (2) at (-25pt, -20pt) {$R_u^{-1} R_v R_u$};
	\draw[color=black!100,line width=1.6pt] (-64pt,-20pt) -- (-50pt, -20pt);
	\draw[color=black!100,line width=1.6pt] (-0pt, -20pt) -- (10pt,-20pt);
	\draw[color=black!100,line width=1.6pt] (-64pt,-20.5pt) -- (-64pt,20.5pt);
	\draw[color=black!100,line width=1.6pt] (10pt,-20.5pt) -- (10pt,20.5pt);
\end{tikzpicture} 
\hspace{4pt}
\begin{tikzpicture}
	[baseline={([yshift=0ex]current bounding box.center)}],
	\draw[color=black!100,line width=1.6pt] (7pt,23pt) -- (-7pt,-23pt);
\end{tikzpicture}
\hspace{4pt}
\begin{tikzpicture}
	[baseline={([yshift=1.5ex]current bounding box.center)}],
	\draw[color=black!100,line width=1.6pt] (12pt,12pt) -- (-12pt,12pt) -- (-12pt, -12pt) -- (12pt,-12pt) -- (12pt,12pt);
	\end{tikzpicture}
\\	&\; = \chi(v,u) \,.
\end{aligned} \end{align}
Hence we can summarize
\begin{align}
	\frac{\bigXp{ \auxU^\Th \auxV^\Tv }}{\braket{1}} = \frac{\bigXp{ \auxV^\Tv \auxU^\Th }}{\braket{1}}
	= \begin{cases} \chi(v,u) & u,v \in G , \\ 0 & \text{$u \notin G$ or $v \notin G$}. \end{cases}
\end{align}
For Abelian groups $G$, the bimultiplicative function $\chi$ uniquely identifies the group cohomology class.
Therefore, the evaluation of the $n$-replica tensor network with auxiliary string operators completely determines the boundary SPT order of the network.
As we will be argued in Sec.~\ref{sec:ClassifyMixed}, the boundary SPT order determines the type of Wilson loop operators that exists in the mixed state phase, and thus provide a classification of the mixed state topological orders.

A few comments to cap off our discussion of PEPS.
In this section we explain how the boundary SPT order is used to classify the effect of auxiliary symmetry in the bulk.
However, since we are evaluating the tensor network on a torus, no where in our discussion is there a physical boundary.
Rather, the ``boundary states''---i.e., the dominant eigenvectors of the transfer operator---formally appear when evaluating powers of $\hat{T}_0$ in the thermodynamic limit.
We consider the symmetry characterization of the boundary states as a concrete tool to determine expectations such as the intra- and inter-layer correlation of auxiliary strings in the network.
Because the boundary SPT order is a (tensor network) renormalization invariant, it is no surprise that it can also gives us a classification of $n$-replica mixed states.

We also assume that $G$, the boundary symmetry subgroup, is independent of the direction (of the boundary states).
That is, we assume that $\bigXp{ U^\Th } = \bigXp{ U^\Tv }$.
We find this to be true for generic error channels, but the assumption fails for extreme anisotropic channels.\footnote{For example, if errors only occur on horizontal but not vertical bonds, or if anyon pairs are only created along one direction of the lattice.}
In this work, we assume that topological properties of our states are isotropic; the boundary SPT order is independent of the lattice directions.

Finally, our discussion implicitly relies on a gap in the transfer operator spectrum (such that the boundary states are well-defined and finite in number), and that the boundary states are also short-ranged correlated (such that we can apply the 1D SPT classification results).
These conditions reflect the criteria that we think of when describing a (mixed) state as being ``gapped'' or possessing ``exponential decay in correlations'', which must be violated at topological phase transitions.
Just as it has been done for pure states,\roger{cite}
	it would be useful for future work to formalize these ideas and provide concrete characterizations for a gapped mixed state.

\subsection{Quantum Channels and Boundary SPT}
\label{sec:QEC_quantum_channel}

In this section, we relate the boundary SPT phases of the tensor network with quantum channels in the $\Z_p$ TC model, for $p$ a positive integer.

We begin by establishing our notation for states and Wilson loop operators.
Denote $Z^{\Th}/X^{\Th}$ as the Wilson loop operators wrapping around the torus comprising of $\Z_p$-Pauli $Z/X$ operators along the original/dual lattice.
Similarly $Z^{\Tv}/X^{\Tv}$ are defined along the vertical direciton.
Of course, $(Z^{\Th})^p = (X^{\Th})^p = (Z^{\Tv})^p = (X^{\Tv})^p = 1$.
For a general loop $\Ts$ in the torus wrapping $a$ around horizontally and $b$ times vertically, we denote
\begin{align} \begin{aligned}
	Z^\Ts = Z^{a\Th + b\Tv} &= \bigl(Z^\Th\bigl)^a \bigl(Z^{\Tv}\bigr)^b ,
\\	X^\Ts = X^{a\Th + b\Tv} &= \bigl(X^\Th\bigl)^a \bigl(X^{\Tv}\bigr)^b .
\end{aligned} \end{align}
Here we think of the loops themselves as vectors in a vector space spanned by $\{\Th,\Tv\}$, such that $p\Th = p\Tv = 0$.%
	\footnote{Formally, the loops are elements of the homology group $H_1(T^2;\Z_p)$ (which is Abelian).}
(That is, wrapping around any direction $p$ times is equivalent to a trivial loop.)
This allows us to take advantage of the notation and write statements such as $Z^\Ts Z^\Tt Z^\Ts = Z^{2\Ts+\Tt}$, $Z^{p\Ts} = Z^0 = 1$.
The Wilson loops obey the algebra
\begin{align} \label{eq:TC_WL_algebra}
	Z^\Ts X^\Tt &= \zeta_p^{\Ts \times \Tt} X^\Tt Z^\Ts ,
&	X^\Ts Z^\Tt &= \zeta_p^{\Ts \times \Tt} Z^\Tt X^\Ts ,
\end{align}
where $\zeta_p$ denotes the $p$\textsuperscript{th} root of unity
\begin{align}
	\zeta_p &\defeq \exp\bigl[ 2 \pi i / p \bigr] ,
\end{align}
and $\Ts \times \Tt$ is the intersection product defined by
\begin{align} \label{eq:cross_product}
	(a_1 \Th + a_2 \Tv) \times (b_1 \Th + b_2 \Tv) = a_1 b_2 - a_2 b_1 \,.
\end{align}
The intersection product $\Ts \times \Tt = -\Tt \times \Ts$ counts the number%
	\footnote{The intersection product is only defined modulo $p$.  However $\zeta_p^{\Ts\times\Tt}$ is unambiguous.}
	of oriented crossings between loops $\Ts$ and $\Tt$.

\begin{figure}
	\centering
	\includegraphics[width = \linewidth]{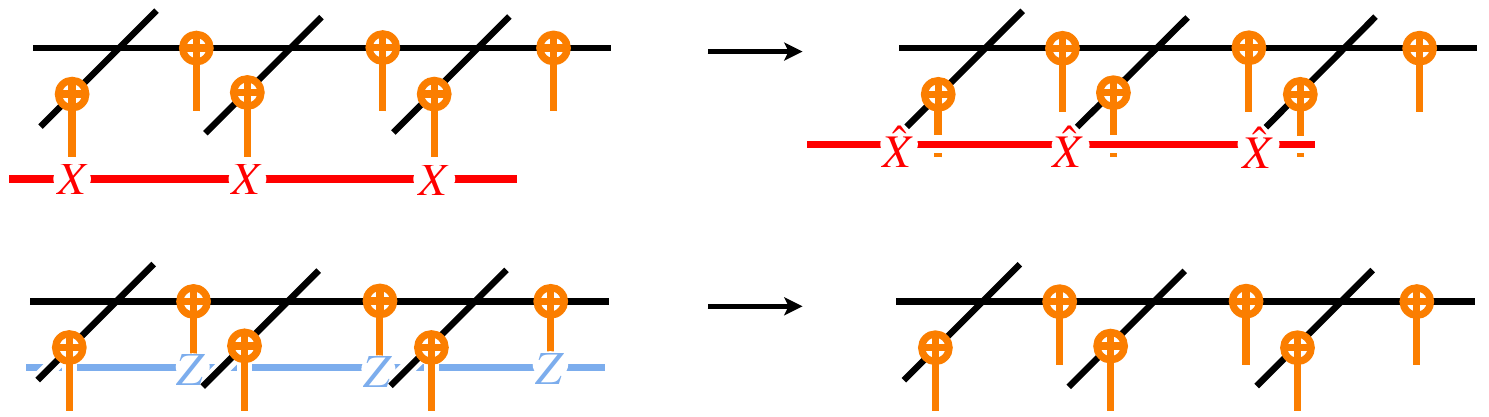}
	\caption{
For the (pure state) toric code PEPS, an $X$-Wilson loop applied to the physical degrees of freedom can be pushed into a string operator $\auxX$ acting on the auxiliary degrees of freedom.
In contrast, a $Z$-Wilson loop can be absorbed into the network (they detect the boundary conditions of the PEPS).
	} \label{fig:TC_Wilson_aux}
\end{figure}

There are $p^2$ orthogonal logical states---that is, the ground state Hilbert space of the $\Z_p$ TC Hamiltonian has dimension $p^2$, or equivalently its topological quantum field theory (TQFT) has $p^2$ anyons.
Denote $\ket{0}$ as the eigenstate of $Z^\Th$ and $Z^\Tv$ with eigenvalue 1.\footnote{$\ket{0}$ is not a minimum entangled state (MES), and should not be confused with the MES corresponding to the vacuum anyon.}
We label the pure ground states via loop labels $\Ts$, such that $\ket{\Ts} = X^\Ts \ket{0}$.
In this basis, the Wilson loop operators behave as
\begin{align} \label{eq:Xpsi_basis} \begin{aligned}
	Z^{\Ts} \ket{\Tt} &= \zeta_p^{\Ts\times\Tt} \ket{\Tt} ,
\\	X^{\Ts} \ket{\Tt} &= \ket{\Ts+\Tt} .
\end{aligned} \end{align}
The state $\ket{0}$ is of particular importance because it results from the PEPS on a torus with periodic boundary conditions,
i.e., the ($\Z_p$ generalization of the) state $\ket{\Psi_{0,0}}$ described in Eq.~\eqref{eq:TC_wavefunctions}.
Indeed, as shown in Fig.~\ref{fig:TC_Wilson_aux}, $Z$-type Wilson loops can be ``absorbed'' into the PEPS network; meaning $Z^\Ts \ket{0} = \ket{0}$.
Due to the $\Z_p$ auxiliary symmetry of the PEPS, $X$-type Wilson loop operators can be pushed between the physical and auxiliary degrees of freedom, also shown in Fig.~\ref{fig:TC_Wilson_aux}.
In other words, acting $X^{a\Th+b\Tv}$ on the (physical legs of) state $\ket{0}$ is equivalent to inserting a horizontal $\auxX^a$-string and a vertical $\auxX^b$-string into the auxiliary legs of the tensor network.
In other words,
\begin{align} \label{eq:TC_Xloop_Xp} \begin{aligned}
&	\bigXp{ \auxX_1^{\Tt_1} \auxX_2^{\Tt_2} \cdots \auxX_{2n}^{\Tt_{2n}} }_\text{TC}
\\	&\quad = \Tr X^{\Tt_1} \rho_{00} X^{\Tt_2} X^{\Tt_3} \rho_{00} X^{\Tt_4} X^{\Tt_5} \cdots \rho_{00} X^{\Tt_{2n}} ,
\end{aligned} \end{align}
where we denote $\rho_{00} = \ketbra{0}{0}$.
Similar to the notation used for Wilson loops, $\auxX_i^{a\Th+b\Tv}$ means that we insert on the $i$\textsuperscript{th} layer, $a$ horizontal $\auxX$-strings and $b$ vertical $\auxX$ strings.
This is reflected in the notational similarity between physical Wilson loop operators $X^\Tt$ and auxiliary strings $\auxX^\Tt$.

Applying Eq.~\eqref{eq:Xpsi_basis} to simplify Eq.~\eqref{eq:TC_Xloop_Xp},
\begin{align} \label{eq:xp_TC} \begin{aligned}
&	\bigXp{ \auxX_1^{\Tt_1} \auxX_2^{\Tt_2} \cdots \auxX_{2n}^{\Tt_{2n}} }_\text{TC}
\\	&= \braket{ 0 | X^{\Tt_2} X^{\Tt_3} | 0 } \braket{ 0 | X^{\Tt_4} X^{\Tt_5} | 0 } \cdots \braket{ 0 | X^{\Tt_{2n}} X^{\Tt_1} | 0 }
\\	&= \delta(\Tt_2+\Tt_3) \delta(\Tt_4+\Tt_5) \cdots \delta(\Tt_{2n}+\Tt_1) .
\end{aligned} \end{align}
This expression completely captures the boundary SPT order of the toric code.
For example, $\bigXp{\auxX_1^\Th} = \bigXp{\auxX_2^\Th} = 0$ means that the $\Z_p$ symmetry are spontaneously broken in the first and second layer.
However, $\bigXp{\auxX_1^\Th \auxX_2^{-\Th}} = 1$ means that a residue $\Z_p$ (generated by $x_1^\phd x_2^{-1}$) remains unbroken.
Finally, the fact that the expectation value is always 0 or 1 means that the boundary state is in the trivial SPT phase of its symmetry group $G$.

Here
\begin{align}
	\delta(\Ts) = \begin{cases} 1 & \Ts=0 , \\ 0 & \Ts\neq0 , \end{cases}
\end{align}
is the delta function.

We now analyze the postselection-based QEC protocol shown in Fig.~\ref{fig:QEC} on the toric code. 
The initial state is $n$ copies of some TC logical state $\rho_1$.
External noises are introduced by quantum channel $\E$, followed by measurement of all the syndrome operators [cf.\ Eq.~\eqref{eq:toric_code_stab}].
Explicitly, the quantum channel
\begin{align} \label{eq:M_channel}
	\M = \prod_\text{vertices $+$}\mkern-16mu \Phi[A_+] \mkern8mu \circ \prod_\text{plaquettes $\square$}\mkern-24mu \Phi[B_\square] \,,
\end{align}
a product over all the plaquette/star stabilizers with $\Phi[P](\sigma) \defeq \frac{1}{2}\sigma + \frac{1}{2} P \sigma P^\dag$.
Crucially, $\M$ is a product of local channels and therefore can be written as a (superoperator) tensor network.
After measuring the stabilizers (star and plaquette operators), the density matrix becomes an incoherent sum of different error configuration states:
\begin{align}
	\mathcal{M}\E(\rho_1) = \sum_{\theta} \Prob(\theta) \, Q_\theta(\rho_1) ,
\end{align}
where $\Prob(\theta) = \Tr(M_\theta \E(\rho_1) M_\theta)$ is the probability that the syndrome $\theta$ is detected,
$M_\theta$ is the projector into the set of states compatible with the syndrome $\theta$,
and $Q_\theta(\rho_1) = M_\theta \E(\rho_1) M_\theta / \Tr[M_\theta \E(\rho_1) M_\theta]$ is the density matrix after measurement when $\theta$ is detected.
Notably, the probabilities $\Prob(\theta)$ do not depend on the initial state $\rho_1$.

Regardless of $\theta$, the rank of $M_\theta$ is $p^2$.
Simply put, once the syndrome is determined, $Q_\theta(\rho_1)$ belong to some $p^2$-dimensional Hilbert space which is isomorphic to the logical space of the TC.
Denote $M_0$ as the projector into the set of states with no error syndromes; i.e., logical states of the TC.
Then for each $\theta$ we can find a unitary $E_\theta$ that maps between the two space:%
	\footnote{$E_\theta$ can always be chosen to a product of $\Z_p$ Pauli operators.}
	$M_\theta = E_\theta^\phd M_0 E_\theta^\dag$.
Therefore, we can express the map from $\rho_1$ to the post-measurement state as a quantum channel
\begin{align} \label{eq:def_Ptheta}
	Q_\theta(\rho) = \sum_{\Ti,\Tj,\Tk,\Tl} P_{\Ti,\Tk;\Tj,\Tl}(\theta) \, E_\theta X^\Ti Z^\Tk \rho \, Z^{-\Tl} X^{-\Tj} E^\dag_\theta \,.
\end{align}
The summation is over loops $\Ti,\Tj,\Tk,\Tl$ on the torus.
The set $\{ X^\Ti Z^\Tk \}$ is a linear basis of operators in the logical space, and $P_{\Ti,\Tk;\Tj,\Tl}(\theta)$ are the coefficients of the channel in such basis.
$P$ obey $\sum_{\Ti,\Tj} P_{\Ti,\Tj;\Ti+\Ts,\Tj+\Tt}(\theta) = \delta(\Ts) \delta(\Tt)$ (trace-preserving map) and $P_{\Ti,\Tk;\Tj,\Tl}(\theta)$ is a positive semidefinite viewed as a $p^4 \times p^4$ matrix indexed by $(\Ti,\Tk)$ and $(\Tj,\Tl)$.

To relate the coefficients $P(\theta)$ to symmetry properties of the PEPS, it is convenient to work in the basis defined in Eq.~\eqref{eq:Xpsi_basis}.
First define the ``intersection form''
\begin{align} \label{eq:def_I}
	I_{\Ts,\Tt} \defeq \zeta_p^{\Ts \times \Tt} \,,
\end{align}
which satisfies
\begin{subequations} \label{eq:I_properties} \begin{align}
	I_{\Ts,\Tt}^\phd = I_{\Ts,-\Tt}^\ast = I_{-\Ts,\Tt}^\ast &= I_{\Tt,\Ts}^\ast \,,
\\	I_{\Ti,\Ts} I_{\Ti,\Tt} &= I_{\Ti,\Ts+\Tt} \,,
\\	\sum_{\Ti} I_{\Ts,\Ti} = \sum_{\Ti} I_{\Ti,\Ts} &= p^2 \delta(\Ti) .
\end{align} \end{subequations}
Hence, as a $p^2 \times p^2$ matrix, $\frac{1}{p}I$ is both unitary and Hermitian.
Let
\begin{align} \label{eq:P_to_Q}
	Q_{\Ti,\Ts;\Tj,\Tt}(\theta) &= \sum_{\Tk,\Tl} P_{\Ti-\Ts,\Tk \,;\, \Tj-\Tt,\Tl}(\theta) \, I_{\Ts,\Tk}^\ast \, I_{\Tt,\Tl}^{\vphantom\ast} \,.
\end{align}

Acting on $\rho_1 = \sum_{\Ts,\Tt} \rho_{\Ts,\Tt} \ketbra{\Ts}{\Tt} = \sum_{\Ts,\Tt} \rho_{\Ts,\Tt} X^\Ts \rho_{00} X^{-\Tt}$,
the quantum channel behave as
\begin{align} \label{eq:Qchannel}
	\bigl( E_\theta^\dag Q_\theta(\rho_1) E_\theta \bigr)_{\Ts',\Tt'}
	&= \sum_{\Ti,\Tj,\Tk,\Tl} \raisebox{-1ex}{$\begin{array}{l} P_{\Ti,\Tk;\Tj,\Tl}(\theta) \\[0.2ex]\quad \times \Braket{ \Ts' | X^\Ti Z^\Tk \rho_1 Z^{-\Tl} X^{-\Tj} | \Tt' } \end{array}$}
\notag\\
	&= \sum_{\Tk,\Tl,\Ts,\Tt} P_{\Ts'-\Ts\,,\,\Tk\,;\,\Tt'-\Tt\,,\,\Tl}(\theta) \, I_{\Tk,\Ts} I_{\Tt,\Tl} \, \rho_{\Ts,\Tt}
\notag\\
	&= \sum_{\Ts,\Tt} Q_{\Ts',\Ts \,;\, \Tt',\Tt}(\theta) \, \rho_{\Ts,\Tt} \,.
\end{align}
Evidently $Q_{\Ti,\Ts;\Tj,\Tt}(\theta)$ is the Choi matrix of the map $E_\theta^\dag Q_\theta E_\theta^\phd$ in the basis~\eqref{eq:Xpsi_basis}.
It is positive semidefinite ($Q(\theta) \succeq 0$) and trace-preserving $\sum_{\Ti} Q_{\Ti,\Ts;\Ti,\Tt}(\theta) = \delta(\Ts-\Tt)$.
$Q(\theta)$ completely characterizes the channel $\M\E$ acting on the TC state; tracking the evolution of logical quantum information for all possible error syndrome.

The key step in our QEC protocol is to postselect the same error syndromes for all copies.
This effectively alters the probability distribution for each $\theta$.
\begin{align}
	\rho_\text{ps} = \mathcal{S}\bigl( \M\E(\rho_1)^{\otimes n} \bigr)
	&= \frac{\sum_{\theta} \Prob(\theta)^n Q_\theta(\rho_1)^{\otimes n}}{\sum_{\theta} \displaystyle\Prob(\theta)^n}.
\end{align}

We can now relate this postselection-based QEC protocol with the expectation value of auxiliary loop operators.
The expectation value
\begin{align}
	\bigXp{ \auxX_1^{t_1} \auxX_2^{t_2} \cdots \auxX_{2n}^{t_{2n}} }
\end{align}
of the $n$-replica network of $\M\E(\rho_{00})$ is the evaluation of tensor network representing $\Tr\bigl[ \bigl( \M\E(\rho_{00}) \bigr)^{n} \bigr]$ inserting $\auxX_i^{t_i}$ in the $i$\textsuperscript{th} layer.
Recall that in the (unperturbed) toric code state, auxiliary $\auxX$-strings can be transformed into $X$-Wilson loop operators, making the evaluation of $\auxX$-strings straightforward [cf.\ Fig.~\ref{fig:TC_Wilson_aux}, Eq.~\eqref{eq:TC_Xloop_Xp}].
This is no longer the case once the error channel is introduced since the $X$-Wilson loops can be ``dressed'', or may no longer exist: $\bigXp{ \auxX_1^{t_1} \auxX_2^{t_2} \cdots } \neq \Tr X^{\Tt_1} \M\E(\rho_{00}) X^{\Tt_2} \cdots$.

Since the strings $\auxX$ acts on the auxiliary space while the channels $\M\E$ act on the physical space of the tensor network, their action actually commutes!
We can first insert strings $\auxX_1^{t_1}$, $\auxX_2^{t_2}$, etc. into the toric code density matrix $\rho_{00}$ (with periodic boundary conditions), then apply the error and measurement channels.
For the toric code PEPS, $\auxX_{2i-1}^{t_{2i-1}}$ and $\auxX_{2i}^{t_{2i}}$ can be pushed into Wilson loops acting on the ket and bra side, respectively, of the $i$\textsuperscript{th} copy of $\rho_{00}$.

\begin{widetext}
\noindent
Then the expectation values of the loop operators are given by
\begin{align}
\Braket{ \auxX_1^{\Tt_1} \auxX_2^{\Tt_2} \cdots \auxX_{2n}^{\Tt_{2n}} }
&\notag
= \Tr\Bigl[ \M\E(X^{\Tt_1} \rho_{00} X_2^{\Tt_2}) \, \M\E(X^{\Tt_3} \rho_{00} X_2^{\Tt_4}) \, \cdots \, \M\E(X^{\Tt_{2n-1}} \rho_{00} X_2^{\Tt_{2n}}) \Bigr]
	\,\Big/\, \Tr\Bigl[ \bigl( \M\E(\rho_{00}) \bigr)^n \Bigr]
\\\notag&
= {\displaystyle\sum}_{\theta}{\Prob^n} \Tr\Bigl[ Q_\theta(X^{\Tt_1} \rho_{00} X^{\Tt_2}) \, Q_\theta(X^{\Tt_3} \rho_{00} X^{\Tt_4}) \, \cdots \, Q_\theta(X^{\Tt_{2n-1}} \rho_{00} X^{\Tt_{2n}}) \Bigr]
	\,\Big/\, {\displaystyle\sum}_{\theta}{\Prob^n} \Tr\Bigl[ \bigl( Q_\theta(\rho_{00}) \bigr)^n \Bigr]
\\\notag&
=	\frac{ {\displaystyle\sum}_{\theta} \Prob^n \, {\displaystyle\sum}_{\Ti_1,\cdots,\Ti_{n}} Q_{\Ti_1,\Tt_1;\,\Ti_2,-\Tt_2} \, Q_{\Ti_2,\Tt_3;\,\Ti_3,-\Tt_4} \cdots Q_{\Ti_{n},\Tt_{2n-1};\,\Ti_{1},-\Tt_{2n}}}
	{ {\displaystyle\sum}_{\theta} \Prob^n \, {\displaystyle\sum}_{\Ti_1,\cdots,\Ti_n} Q_{\Ti_1,0;\,\Ti_2,0} \, Q_{\Ti_2,0;\,\Ti_3,0} \cdots Q_{\Ti_n,0;\,\Ti_1,0}}
\\&
=	\frac{ {\displaystyle\sum}_{\theta}{\Prob^n} \tr\bigl( q_{\Tt_1,-\Tt_2}^\phd \, q_{\Tt_3,-\Tt_4}^\phd \cdots\, q_{\Tt_{2n-1},-\Tt_{2n}}^\phd \bigr) }
	{ {\displaystyle\sum}_{\theta} \Prob^n \tr\bigl(q_{00}^n\bigr) }
\,. \label{eq:xp_XXXX_Qq}
\end{align}
\end{widetext}
Here the dependence on the syndrome configuration $\theta$ are implicit.
On the left hand side, we also omit $\braket{1} = \Tr \M\E(\rho_{00})^{n}$ in the denominator.\footnote{From this point on, we always absorb $\braket{1}$ into $\bigXp{ \auxX_1^{\Tt_1} \cdots }$.}
For convenience, we also define a set of $p^2 \times p^2$ matrices $q_{\bullet\bullet}(\theta)$ which are submatrices of $Q(\theta)$:
\begin{align}
	\bigl(q_{\Ts,\Tt}(\theta)\bigr)_{\Ti,\Tj} \defeq Q_{\Ti,\Ts;\Tj,\Tt}(\theta) \,.
\end{align}

We refer to the object $\bigXp{ \auxX_1^{\Tt_1} \auxX_2^{\Tt_2} \cdots \auxX_{2n}^{\Tt_{2n}} }$ (the collection of numbers) as the \textbf{$X$-loop expectation tensor}.
Equation~\eqref{eq:xp_XXXX_Qq} is the keystone connecting the bulk topological order (as characterized by the boundary SPT order) to the quantum information content of a mixed state.
In Sec.~\ref{sec:ClassifyMixed}, we leverage this connection to extract information about the types of Wilson loop operators and the geometry to the state space.
Prior to that, we first find the set of solutions to the equation.

\subsection{Enumerating the possible topological orders}
\label{sec:enum_boundary_SPT}

In this subsection, we enumerate the possible $n$-replica topological phases which are descendants of the $\Z_p$ toric code for $p$ primes,
and associate them to quantum channels acting on the logical space of the toric code.

\begin{theorem} \label{thm:TCdesc_classification}
Let $p$ be a prime number.
Subject to the constraints listed below, there are exactly $p+3$ classes of solutions to Eq.~\eqref{eq:xp_XXXX_Qq} for $n \geq 2$.
The results are summarized in Table~\ref{tab:TCcondensation}.
\end{theorem}
Our result is that the boundary state can admit only one of $p+3$ possible SPT orders, each of these order representing a distinct mixed state topological order.
For each case the error channel (followed by syndrome measurement and postselection) can be described by a specific class of quantum channels.
When $p$ is not a prime, the SPT orders/quantum channels found in this section remain valid, however there are additional solutions not covered in Table~\ref{tab:TCcondensation}.

The constraints are as follows.
\begin{itemize}
\item
    In sums over $\theta$ [e.g.\ Eq.~\eqref{eq:xp_XXXX_Qq}], we only consider $\theta$ with $\Prob(\theta) > 0$.

    This allows us to ignore terms with $\Pr(\theta)^n = 0$ where $Q$ cannot be constrained.  However, those zero terms do not contribute to any expectation values or physical processes.

\item Hermiticity and positivity.
	\begin{align} Q_{\Ti,\Ts\,;\,\Tj,\Tt}^\ast(\theta) = Q_{\Tj,\Tt\,;\,\Ti,\Ts}^\phd(\theta) \end{align}
	and
	\begin{align} \label{eq:Q_positive} Q(\theta) \succeq 0 \end{align}
	when $Q$ is written as a matrix with rows indexed by $(\Ti,\Ts)$ and columns indexed by $(\Tj,\Tt)$.

	The Hermiticity condition is $q_{\Ts\Tt}^\dag = q_{\Tt\Ts}^\phd$ phrased in terms of the $q$-matrices.
	The positivity condition means that any principal submatrix of $Q$ is also positive.
	For example, $q_{\Ts\Ts} \succeq 0$, and
	\begin{align} \label{eq:q2_positive} \begin{bmatrix} q_{\Ts\Ts} & q_{\Ts\Tt} \\ q_{\Tt\Ts} & q_{\Tt\Tt} \end{bmatrix} \succeq 0 . \end{align}
	which in turn implies%
	\footnote{$\operatorname{im} A$ denote the image of matrix $A$; the vector space spanned by the columns of $A$.}
	\begin{align} \label{eq:column_inclusion} \operatorname{im} q_{\Ts\Tt} \subseteq \operatorname{im} q_{\Ts\Ts} \,. \end{align}
	(See Lemma~\ref{lem:column_inclusion}.)

\item Normalization (trace).
    \begin{align} \tr q_{\Ts\Tt}(\theta) = \delta(\Ts-\Tt) . \end{align}

\item SPT order.
As discussed in Sec.~\ref{sec:PEPS}, the $n$-replica PEPS network admits $\Z_p^{2n}$ symmetry in its auxiliary degrees of freedom.
Denote the symmetry generator on the $k$\textsuperscript{th} copy of the TC state as $x_k$, that is, $\Z_p^{2n}$ is the Abelian group generated by $x_1, \dots, x_{2n}$ such that $x_k^p = 1$.

Let $\Tt_1, \Tt_2, \dots, \Tt_{2n} \in \Z_p[\Th] \oplus \Z_p[\Tv]$ be loops on the torus.
Decompose $\Tt_k = a_k \Th + b_k \Tv$ into horizontal and vertical components.
Let $a = \prod_k x_k^{a_k} \in \Z_p^{2n}$ and $b = \prod_k x_k^{b_k} \in \Z_p^{2n}$.%
    \footnote{For the Abelian group of loops $\Z_p[\Th] \oplus \Z_p[\Tv]$, we use the additive notation $(+)$ to denote group composition.  For the symmetry group of the tensor network $\Z_p^{2n}$ (also Abelian), we use the multiplicative notation.}
The $X$-loop expectation tensor is given by
\begin{align} \label{eq:xp_XXXX_sptorder} \quad
    \bigXp{ \auxX_1^{\Tt_1} \cdots \auxX_{2n}^{\Tt_{2n}} }
    &= \begin{cases} \chi(a,b) & \text{$a \in G$ and $b \in G$},
        \\ 0 & \text{$a \notin G$ or $b \notin G$}, \end{cases}
\end{align}
for some group $G \subseteq \Z_p^{2n}$.
$\chi: G \times G \to \mathbb{C}^\ast$ is an alternating bimultiplicative function, that is,
$\chi(ab,c) = \chi(a,c) \chi(b,c)$, $\chi(a,bc) = \chi(a,b) \chi(a,c)$, and $\chi(a,a) = 1$.

These conditions implies that $\chi(a,b)^p = 1$ for all $a,b \in G$,
with $\chi(1,a) = +1$ and $\chi(a,b) = \chi(b,a)^{-1} = \chi(b,a)$.
There is a one-to-one correspondence between functions $\chi$ and SPT orders of $G$.
When there are only logical strings in one direction---that is, $\{\Tt_1,\dots,\Tt_{2n}\}$ consists of only scalar multiples of some element $\Tl$---the expectation value only depends on whether $a$ (or $b$) is in $G$, and not on the SPT class.

\end{itemize}

\begin{lemma} \label{lem:Qrestrictions}
Subject to the constraints, the following are true.
\begin{enumerate}[label={\it(\arabic*)}, itemsep=0pt, parsep=4pt, topsep=-2pt]
\item Cyclic symmetry.
If $x_{i_1} x_{i_2} \cdots x_{i_m} \in G$,
then $x_{i_1+2k} x_{i_2+2k} \cdots x_{i_m+2k} \in G$ for all $k\in\Z$.
(Here we use cyclic labelling for group generators $x_i = x_{i+2n\Z}$.)

\item $x_1^\phd x_2^{-1} x_3^\phd x_4^{-1} \cdots x_{2n}^{-1} \in G$.

\item If $x_1^\phd x_2^{-1} \in G$, then $q_{\Ts\Ts}(\theta) = q_{00}(\theta)$ for all loops $\Ts$.

\item If $x_1^\phd x_2^{-1} \notin G$, then the set $\bigl\{ q_{\Ts\Ts}(\theta) \,\big|\, \text{loops $\Ts$} \bigr\}$ are mutually orthogonal projectors.

\item $x_k^m \notin G$ for any $k$ and $1 \leq m < p$.

\item If $x_2^\phd x_3^{-1} \notin G$, then $q_{\Ts\Tt}(\theta) = 0$ for $\Ts\neq\Tt$.

\end{enumerate}
\end{lemma}

\noindent
In what follows, let
\begin{align}
    D \defeq \sum_\theta \Prob^n \tr q_{00}^n
\end{align}
denote the denominator in~\eqref{eq:xp_XXXX_Qq}.
Crucially, $D$ is positive.
We will often use the property that if two matrices $A,B \succeq 0$, then
\begin{align} \label{eq:2positive0trace}
    \tr AB = 0 \;\Leftrightarrow\; BA = 0 \;\Leftrightarrow\; \operatorname{im}A \subseteq \ker B .
\end{align}

\begin{proof}[Proof to \ref{lem:Qrestrictions}(1)]
This is evident by considering the form of the expectation tensor $\bigXp{ \auxX_1^{\Tt_1} \cdots \auxX_{2n}^{\Tt_{2n}} }$ combining Eq.~\eqref{eq:xp_XXXX_Qq} and Eq.~\eqref{eq:xp_XXXX_sptorder}.
\end{proof}

\begin{proof}[Proof to \ref{lem:Qrestrictions}(2)]
As $\tr q_{\Ts\Ts} = 1$ for each $\Ts$ and $\theta$, we have the bound $1 \geq \tr q_{\Ts\Ts}^n \geq p^{2(1-n)}$.
Hence $\bigXp{\auxX_1^\Ts \auxX_2^{-\Ts} \auxX_3^{\Ts} \cdots \auxX_{2n}^{-\Ts}} = \frac{1}{D} \sum_\theta \Prob^n \tr q_{\Ts\Ts}^n \geq p^{2(1-n)}$ is positive.
Via Eq.~\eqref{eq:xp_XXXX_sptorder}, the expression must in fact be equal to $1$ and $x_1^\phd x_2^{-1} x_3 \cdots x_{2n}^{-1} \in G$.
\end{proof}

\begin{proof}[Proof to \ref{lem:Qrestrictions}(3)]
By {\it(1)}, $x_{2j-1}^\phd x_{2j}^{-1} \in G$ for all $j\in\Z$.
Hence
$\bigXp{\auxX_1^{\Ts_1} \auxX_2^{-\Ts_1} \auxX_3^{\Ts_2} \auxX_4^{-\Ts_2} \cdots \auxX_{2n-1}^{\Ts_n} \auxX_{2n}^{-\Ts_n}} = +1$ for any selection of $\Ts_k \in \{0,\Th\}$,
which implies $\sum_\theta \Prob^n \tr\bigl[ q_{\Ts_1\Ts_1} q_{\Ts_2\Ts_2} \cdots q_{\Ts_n\Ts_n} \bigr] = D$.
The expressions
\begin{align} \label{eq:lem_Qrestrict_c2} \begin{aligned}
    \textstyle \sum_\theta \Prob^n \tr\bigl[ (q_{00}-q_{\Th\Th})^2 q_{00}^{n-2} \bigr] &= 0 ,
\\  \textstyle \sum_\theta \Prob^n \tr\bigl[ (q_{00}-q_{\Th\Th})^2 q_{\Th\Th}^{n-2} \bigr] &= 0 .
\end{aligned} \end{align}
follow because they are linear combinations of the previous above.
As $q_{00} \succeq 0$, $q_{\Th\Th} \succeq 0$, and $(q_{00}-q_{\Th\Th})^2 \succeq 0$,
every term in the summand of~\eqref{eq:lem_Qrestrict_c2} must be non-negative and therefore zero.
$(q_{00}-q_{\Th\Th})^2$ must be orthogonal [cf.~\eqref{eq:2positive0trace}] to both $q_{\bullet\bullet}$ while also part of the algebra that they generate, and hence $(q_{00}-q_{\Th\Th})^2 = 0$.

Repeating the same argument replacing $\Th$ for any other loop, we conclude that $q_{00}(\theta) = q_{\Ts\Ts}(\theta)$ for all $\Ts$.
\end{proof}

\begin{proof}[Proof to \ref{lem:Qrestrictions}(4)]
When $p$ is a prime, $x_1^\phd x_2^{-1} \notin G$ implies that $x_1^{\vphantom{-}m} x_2^{-m} \notin G$ for all $1 \leq m < p$ (since $m$ is invertible in the field $\Z_p$.)
Combining $x_1^\phd x_2^{-1} \notin G$ and {\it(2)}, $x_3^\phd x_4^{-1} \cdots x_{2n-1}^\phd x_{2n}^{-1} \notin G$.
So $0 = \bigXp{\auxX_1^\Ts \auxX_2^{-\Ts} \auxX_3^\Tt X_4^{-\Tt} \cdots \auxX_{2n}^{-\Tt}}$
$= \frac{1}{D} \sum_\theta \Prob^n \tr( q_{\Ts\Ts}^\phd q_{\Tt\Tt}^{n-1} )$ when $\Ts \neq \Tt$.
Since $q_{\Ts\Ts} \succeq 0$, each individual summand $\tr \bigl[ q_{\Ts\Ts}^\phd q_{\Tt\Tt}^{n-1} \bigr]$ must vanish.
Thus for each $\theta$, the $q_{\Ts\Ts}(\theta)$s are mutually orthogonal.
Since there are $p^2$ such $q$ matrices, each of which are $p^2 \times p^2$ matrices with unit trace, they must be orthogonal projectors.

The diagonal blocks take the general form
\begin{align} \label{eq:orthogonal_blocks} \begin{aligned}
    q_{00} &= U_\theta^\phd \begin{bmatrix}
    1&0&\cdots&0\\0&0&\cdots&0\\[-1ex]\vdots&\vdots&\ddots&\vdots\\0&0&\cdots&0 \end{bmatrix} U_\theta^\dag ,
\\
    q_{\Ts\Ts} &= U_\theta^\phd \begin{bmatrix}\phd
    \ddots&&&&\\&0&&&\\&&1&&\\&&&0&\\[-1ex]&&&&\ddots \end{bmatrix} U_\theta^\dag ,
\end{aligned} \end{align}
where $U_\theta$ are syndrome-dependent $p^2 \times p^2$ unitary matrices.
\end{proof}

\begin{proof}[Proof to \ref{lem:Qrestrictions}(5)]
It suffices to prove that $x_1 \notin G$, since $p$ is a prime.
Let $\Ts \neq 0$.
Observe that $\bigXp{\auxX_1^\Ts}
= \frac{1}{D} \sum_\theta \Prob^n \tr q_{\Ts0}q_{00}^{n-1}
= \frac{1}{D} \bigl[ \sum_\theta \Prob^n \tr q_{0\Ts} q_{00}^{n-1} \bigr]^\ast = \bigXp{\auxX_2^{-\Ts}}^\ast$
and hence $x_1 \in G \Leftrightarrow x_2 \in G$.

[Proof by contradiction.]
Suppose $x_k \in G$ for some $k$.
By {\it(1)} either $x_1$ or $x_2$ is also in $G$, and hence all $x_1, x_2, x_1^\phd x_2^{-1} \in G$.
From {\it(3)}, $q_{00} = q_{\Ts\Ts}$, and so [cf.\ Eq.~\eqref{eq:q2_positive}]
\begin{align}
    \begin{bmatrix} q_{00} & q_{0\Ts} \\ q_{\Ts0} & q_{00} \end{bmatrix} \succeq 0 .
\end{align}
By Lemma~\ref{lem:block2tr}, $\bigl| \tr\bigl( q_{\Ts0}^\phd q_{00}^{n-1} \bigr) \bigr| \leq \tr q_{00}^n$, with equality iff $q_{\Ts0} = e^{i\phi} q_{00}$ for some phase $\phi$.
However, equality cannot occur since $\tr q_{\Ts0} = 0 \neq 1 = \tr q_{00}$.
Therefore,
\begin{align}
\notag
	1 = \bigXp{\auxX_1^\Ts} &= \frac{1}{D} \mathord{\sum}_\theta \Prob^n \tr\bigl( q_{\Ts0}^\phd q_{00}^{n-1} \bigr)
\\\notag
	&\leq \frac{1}{D} \mathord{\sum}_\theta \Prob^n \Bigl| \tr\bigl( q_{\Ts0}^\phd q_{00}^{n-1} \bigr) \Bigr|
\\\notag
&   < \frac{1}{D} \sum_\theta \Prob^n \tr q_{00}^n
\\&= 1 ,
\end{align}
which is a contradiction.
\end{proof}

\begin{proof}[Proof to \ref{lem:Qrestrictions}(6)]
Let $\Ts \neq \Tt$.
Observe that the expectation value
$\bigXp{ \auxX_1^{+\Ts} \auxX_2^{-\Tt} \auxX_3^{+\Tt} \auxX_4^{-\Ts} \auxX_5^{+\Ts} \cdots \auxX_{2n}^{-\Ts} }
= \frac{1}{D}\sum_\theta \Pr^n \tr(q_{\Ts\Tt}q_{\Tt\Ts}q_{\Ts\Ts}^{n-2})$ vanishes because $x_2^\phd x_3^{-1} \notin G$.

Because $q_{\Ts\Ts}$ and $q_{\Ts\Tt}q_{\Tt\Ts}$ are both positive semidefinite, each summand $\tr(q_{\Ts\Tt}q_{\Tt\Ts}q_{\Ts\Ts}^{n-2})$ must non-negative and hence zero.
This means that the image of $q_{\Ts\Tt}q_{\Tt\Ts}$ is both orthogonal [cf.~\eqref{eq:2positive0trace}] and subspace [cf.~\eqref{eq:column_inclusion}] to $\operatorname{im} q_{\Ts\Ts}$, and must be zero as a result.
Therefore, $q_{\Ts\Tt} = 0$.
\end{proof}

\begin{table*}
\caption{Descendant phases of the $\Z_p$ TC: boundary SPT orders and logical quantum channels} \label{tab:TCcondensation}
\begin{minipage}{\textwidth} \renewcommand{\arraystretch}{1.4} \begin{tabular}{l @{\quad\;\;} c @{\quad} c @{\quad} c @{\quad} c}
\hline\hline \\[-3ex]
	& {\parbox{30ex}{ $X$-loop expectation tensor \\ $\bigXp{ \auxX_1^{\Tt_1} \auxX_2^{\Tt_2} \cdots \auxX_{2n-1}^{\Tt_{2n-1}} \auxX_{2n}^{\Tt_{2n}} }$ }} & {\parbox{22ex}{boundary symmetry\\subgroup generators}}
		& {\parbox{22ex}{code space channel\\$\R\M\E(\cdot)$}} & \parbox{13ex}{decoherence\\type}
\\[1.5ex]\hline
	Case~1
	& $ \delta(\Tt_1+\Tt_2) \delta(\Tt_3+\Tt_4) \cdots \delta(\Tt_{2n-1}+\Tt_{2n})$ & $\{x_{2i-1}^\phd x_{2i}^{-1}\}$
		& $q_{00}$ & $e+m$
\\
	Case~2
	& $\delta(\Tt_1+\Tt_2)\delta(\Tt_2+\Tt_3) \cdots \delta(\Tt_{2n}+\Tt_1)$ & $ x_1^\phd x_2^{-1} \cdots x_{2n-1}^\phd x_{2n}^{-1}$
		& $\frac{1}{p^2}\sum_{\Ts} Z^\Ts (\cdot) Z^{-\Ts}$ & $e$
\\
	Case~3a
	& $\delta(\Tt_1+\Tt_2+\dots+\Tt_{2n})$ & $\{x_1^\phd x_i^{-1}\}$
		& $\frac{1}{p^2}\sum_{\Ts} X^\Ts (\cdot) X^{-\Ts}$ & $m$ 
\\
	Case~3b
	& $ \delta(\Tt_1+\Tt_2+\dots+\Tt_{2n}) \prod_{i<j} I_{\Tt_i,\Tt_j}^c$ & $\{x_1^\phd x_i^{-1}\}$
		& $\frac{1}{p^2}\sum_{\Ts} X^{\Ts} Z^{-c\Ts} (\cdot) Z^{c\Ts} X^{-\Ts}$ & $e^{-c}m$
\\
	Case~4
	& $\delta(\Tt_2+\Tt_3) \delta(\Tt_4+\Tt_5) \cdots \delta(\Tt_{2n}+\Tt_1)$ & $\{ x_{2i}^\phd x_{2i+1}^{-1}\}$
		& $(\cdot)$ & $-$
\\\hline\hline
\end{tabular} \end{minipage}
\\[0.8ex]
\begin{minipage}{0.96\linewidth} \raggedright
	A classification of the possible phases that results when a local error channels is applied to the $\Z_p$ TC for $p$ prime.
	This classification is based on the boundary SPT order in the $n$-replica network representing $\Tr \M\E(\rho_\text{TC})^{n}$.
	Each case represents a boundary SPT order, except for Case~3b which encompasses $p-1$ possible solutions, one for each integer $c \in [1,p-1]$.
	For each case, we provide the effective quantum channel acting on the code space under optimal recovery.
	The column under ``decoherence type''---consistent with labelling of Ref.~\onlinecite{2023BaoBoundarySPT}---indicate the type(s) of Wilson loop that appear as Kraus operators.
\end{minipage}
\end{table*}

We now consider the four cases based on the presence/absence of $x_1^\phd x_2^{-1}$ and $x_2^\phd x_3^{-1}$ in the group $G$.
\begin{align*}
	\begin{array}{|c|c|c|} \hline & x_2^\phd x_3^{-1} \notin G & x_2^\phd x_3^{-1} \in G \\\hline x_1^\phd x_2^{-1} \notin G & \text{Case~2} & \text{Case~4} \\\hline x_1^\phd x_2^{-1} \in G & \text{Case~1} & \text{Case~3} \\\hline \end{array}
\end{align*}
Cases~1, 2, and 4 each give one solution for $\bigXp{ \auxX_1^{\Tt_1} \cdots }$, while Case~3 contains $p$ total solutions.
We summarize the results in Table~\ref{tab:TCcondensation}. 

\vspace{1em}\noindent\textbf{Case~1:} $x_1^\phd x_2^{-1} \in G$ and $x_2^\phd x_3^{-1} \notin G$.

From Lemma~\ref{lem:Qrestrictions}\textit{(3,6)}
the diagonal blocks are all identical while the off diagonal blocks vanish.
The channel $Q$ becomes
\begin{align}
	Q(\theta) = \begin{bmatrix} q_{00}&&\text{\raisebox{-5pt}{\Large 0}} \\ &\ddots& \\ \text{\Large 0}&&q_{00} \end{bmatrix} .
\end{align}
This channel takes every state to $q_{00}$, completely erasing any logical information held in the code.

Directly applying~\eqref{eq:xp_XXXX_Qq} gives
\begin{align} \label{eq:xp_case_triv} \begin{aligned}
&   \bigXp{ \auxX_1^{\Tt_1} \auxX_2^{\Tt_2} \cdots \auxX_{2n}^{\Tt_{2n}} }
\\  &= \delta(\Tt_1+\Tt_2) \, \delta(\Tt_3+\Tt_4) \,\cdots\, \delta(\Tt_{2n-1}+\Tt_{2n}) .
\end{aligned} \end{align}
Thus, the boundary state symmetry group is $G = \bigl\langle \{x_{2k-1}^\phd x_{2k}^{-1} \,|\, 1\leq k\leq n \} \bigr\rangle \cong \Z_p^{n}$ in its trivial SPT phase.

As we will show in the following section, this is corresponds to a trivial topological mixed state.

\vspace{1em}\noindent\textbf{Case~2:} $x_1^\phd x_2^{-1} \notin G$ and $x_2^\phd x_3^{-1} \notin G$.

From Lemma~\ref{lem:Qrestrictions}\textit{(4,6)}, the diagonal blocks are rank-1 orthogonal matrices while the off-diagonal blocks vanishes.
Writing Eq.~\eqref{eq:orthogonal_blocks} in terms of indices:
\begin{align} \label{eq:Q_case_e}
    Q_{\Ti,\Ts;\Tj,\Tt}(\theta) &= (U_\theta^\phd)_{\Ti,\Ts}^\phd \delta(\Ts-\Tt) (U_\theta^\ast)_{\Tj,\Tt}^\phd
\end{align}
for unitaries $U_\theta$.
The quantum channel [see Eq.~\eqref{eq:Qchannel}] thus strips $\rho$ of its off-diagonal elements [in the $\ket{\Ts}$ basis~\eqref{eq:Xpsi_basis}] and then applies $U_\theta$.
The action of this channel is a measurement of $Z$-loop operators,
\begin{align} \label{eq:Qmap_case_e}
    Q(\theta)[\rho] = U_\theta^\phd \left( \frac{1}{p^2}\sum_{\Ts} Z^\Ts \rho \, Z^{-\Ts} \right) U^\dag_\theta \,.
\end{align}
 
Via Eq.~\eqref{eq:xp_XXXX_Qq},
\begin{align} \label{eq:xp_case_e}
    \bigXp{ \auxX_1^{\Tt_1} \auxX_2^{\Tt_2} \cdots \auxX_{2n}^{\Tt_{2n}} }
    = \begin{cases} 1 & (-1)^k\Tt_k \text{ are all equal} , \\ 0 & \text{otherwise} . \end{cases}
\end{align}
Therefore, the boundary state has $G \cong \Z_p$ symmetry generated by $x_1^\phd x_2^{-1} \cdots x_{2n-1}^\phd x_{2n}^{-1}$.

Later we show that this is corresponds to a classical topological mixed state.

\vspace{1em}\noindent\textbf{Case~3:} $x_1^\phd x_2^{-1} \in G$ and $x_2^\phd x_3^{-1} \in G$.

From Lemma~\ref{lem:Qrestrictions}{\it(1)}, $x_k^\phd x_{k+1}^{-1} \in G$ for all $1 \leq k < 2n$, and hence $x_i^\phd x_j^{-1} \in G$ for all $1 \leq i < j \leq 2n$.
From Lemma~\ref{lem:Qrestrictions}{\it(5)}, $x_i^m \notin G$ for $0 < m < p$. Therefore, $G \cong \Z_p^{2n-1}$ is the group of all elements $x_1^{a_1} \cdots x_{2n}^{a_{2n}}$ such that $\sum_k a_k \equiv 0 \pmod{p}$.

The SPT phase of $G$ are classified by the second cohomology group $H^2(G,\mathrm{U}(1)) \cong \Z_p^{(2n-1)(n-1)}$.
We can label elements of the cohomology group by specifying $\chi(a,b)$ for every pair of elements in a minimal generating set for the Abelian group $G$.

\begin{proposition} \label{prop:xp_case_mf_Im_form}
Any $X$-loop expectation tensor consistent with an SPT order [cf.\ Eq.~\eqref{eq:xp_XXXX_sptorder}] must take on the form
\begin{align} \label{eq:xp_case_mf_Im_form}
    \bigl\langle \auxX_1^{\Tt_1} \auxX_2^{\Tt_2} \cdots \auxX_{2n}^{\Tt_{2n}} \bigr\rangle
    &= \delta(\Tt_1 {\mkern1mu+\mkern1mu} \Tt_2 {\mkern1mu+\mkern1mu} \dots {\mkern1mu+\mkern1mu} \Tt_{2n}) \prod_{i<j} \bigl(I_{\Tt_i,\Tt_j}\bigr)^{m_{i,j}}
\end{align}
for some $2n \times 2n$ antisymmetric integer matrix $m$ with vanishing diagonal.
[$I$ is the intersection form defined in Eq.~\eqref{eq:def_I}.]
\end{proposition}

\begin{proof}
Let $\Tt_1, \Tt_2, \dots, \Tt_{2n} \in \Z_p[\Th] \oplus \Z_p[\Tv]$ be loops.
and $\Tt_k = a_k\Th + b_k\Tv$.
Define $a = \prod_k x_k^{a_k}$ and $b = \prod_k x_k^{b_k}$.
Group elements $a, b \in G$ iff $\sum_k a_k = \sum_k b_k = 0$ iff $\sum_k \Tt_k = 0$.
Thus
\begin{align}
    \bigXp{ \auxX_1^{\Tt_1} \auxX_2^{\Tt_2} \cdots \auxX_{2n}^{\Tt_{2n}} }
    &= \delta(\Tt_1+\Tt_2+\dots+\Tt_{2n}) \, \chi(a,b) .
\end{align}

Let $( g_2=x_2^\phd x_1^{-1}, g_3=x_3^\phd x_1^{-1}, \dots, g_{2n}=x_{2n}^\phd x_1^{-1})$ be a generating set for $G$.
If $a \in G$, then it can be expressed as $a = \prod_{k=2}^{2n} g_k^{a_k}$ (likewise for $b)$.
Thus for any pair $a,b, \in G$,
the bimultiplicative function is
\begin{align} \begin{aligned}
    \chi(a,b)
    &= \prod_{i=2}^{2n}\prod_{j=2}^{2n} \chi(g_i,g_j)^{a_ib_j}
\\  &= \prod_{2 \leq i < j \leq 2n} \mkern-6mu \chi(g_i,g_j)^{a_ib_j-b_ia_j} \mkern2mu.
\end{aligned} \end{align}
In addition, observe that the intersection matrix can be written as
$I_{\Tt_i,\Tt_j} = \zeta_p^{a_ib_j-b_ia_j}$.

[Proof by construction.]
Given an alternating bimultiplicative function $\chi: G \times G \to \{1, \zeta_p, \dots, \zeta_p^{-1}\}$,
Assign $m_{i,j} = 0$ if either $i$ or $j = 1$, and $\zeta_p^{m_{i,j}} = \chi(g_i,g_j)$ otherwise.
A direct calculation combining the previous two equations shows that $\chi(a,b) = \prod_{1 \leq i < j \leq 2n} \bigl(I_{\Tt_i,\Tt_j}\bigr)^{m_{i,j}}$.
\end{proof}

We note that there are multiple choices of the matrix $m$ which produce the same expectation tensor~\eqref{eq:xp_case_mf_Im_form}.
Indeed, there are more possible $m$-matrices (with $\binom{2n}{2}$ independent $\Z_p$ elements) than the possible $G$-SPT orders (with $\binom{2n-1}{2}$ $\Z_p$ generators for the second cohomology group of $G$) (see App.~\ref{sec:TC_case_mf}).

\begin{proposition} \label{prop:case_mf_same_m}
There are only $p$ solutions of the form~\eqref{eq:xp_case_mf_Im_form} consistent with Eq.~\eqref{eq:xp_XXXX_Qq} and the positivity of $Q$.
They are characterized by the matrix
\begin{align}
	m_{i,j} &= \begin{cases} c & i<j, \\ -c & i>j, \\ 0 & \text{otherwise}, \end{cases}
\end{align}
for $0 \leq c < p$.
\end{proposition}

The proof of this proposition is in App.~\ref{sec:TC_case_mf}.
Below we give the corresponding quantum channels for the solutions.
For $c = 0$, which we categorize as Case~3a, the boundary state is in the trivial SPT phase of $G$.
Case~3b deals with the non-trivial classes when $c \neq 0$.

\vspace{1ex}\noindent\textbf{Case~3a.}

In the trivial SPT phase $c = 0$, the expectation tensor is
\begin{align} \label{eq:xp_case_m}
    \bigXp{ \auxX_1^{\Tt_1} \auxX_2^{\Tt_2} \cdots \auxX_{2n}^{\Tt_{2n}} } = \delta(\Tt_1 + \Tt_2 + \ldots + \Tt_{2n}) \,.
\end{align}

This case can be obtained from that of Case~2 via a unitary transformation.
The general form of a unitary transformation takes
\begin{align} \begin{aligned}
	Q_{\Ti,\Ts\,;\,\Tj,\Tt} &\to \sum_{\Ts',\Tt'} Q_{\Ti,\Ts';\,\Tj,\Tt'} \, A_{\Ts,\Ts'}^\phd A_{\Tt,\Tt'}^\ast \,,
\\
	\bigXp{ \auxX_1^{\Tt_1} \cdots } &\to \sum_{\{\Ts\}} \bigXp{ \auxX_1^{\Ts_1} \cdots } \prod_{i=1}^n \bigl( A_{\Tt_{2i-1},\Ts_{2i-1}} A_{-\Tt_{2i},-\Ts_{2i}}^\ast \bigr) ,
\end{aligned} \end{align}
for some unitary matrix $A$.

With the choice of $A = \frac{1}{p}I$ [cf.\ Eqs.~\eqref{eq:I_properties}], we can transform between the expectation tensor of this case~\eqref{eq:xp_case_m} and that of Case~2~\eqref{eq:xp_case_e}.
\begin{align} \begin{aligned}
&	\sum_{\{\Ts_i\}}\left( \prod_{i=1}^{n} \frac{I_{\Tt_{2i-1},\Ts_{2i-1}}}{p} \times \frac{I_{-\Tt_{2i},-\Ts_{2i}}^\ast}{p} \right) \times \text{Eq.~\eqref{eq:xp_case_e}}
\\	&=	\sum_{\{\Ts_i\}}\left( \prod_{i=1}^{2n} \frac{ I_{\Tt_i , -(-1)^{i}\Ts_i} }{p} \right) \times \begin{cases} 1 & (-1)^{i}\Ts_i \text{ are all equal} \\ 0 & \text{otherwise} \end{cases}
\\	&= \sum_{\Ts} \prod_{i=1}^{2n} \frac{I_{\Tt_i,-\Ts}}{p}
\\	&= \frac{1}{p^{2(n-1)}} \delta(\Tt_1+\Tt_2+\dots+\Tt_{2n}) .
\end{aligned} \end{align}
As a result, the corresponding quantum channel is that of Case 2~\eqref{eq:Q_case_e} conjugated by $\frac{1}{p} I$,%
    \footnote{Let $Q^{(e)}$ be the $Q$ from case~2~\eqref{eq:xp_case_e}, and $S = \frac{1}{p}I$.  Then:
    \\$Q^{(m)}_{\Ti,\Ts\,;\,\Tj,\Tt} = \sum_{\Ts',\Tt'} S_{\Ts,\Ts'} S_{\Tt,\Tt'}^\ast Q^{(e)}_{\Ti,\Ts';\,\Tj,\Tt'}$
    \\$= \sum_{\Ts',\Tt'} S_{\Ts,\Ts'} S_{\Tt,\Tt'}^\ast (U_\theta^\phd)_{\Ti,\Ts'} \delta(\Ts'-\Tt') (U_\theta^\ast)_{\Tj,\Tt'}$
    \\$= \sum_{\Ti',\Tj',\Ts',\Tt'} S_{\Ts,\Ts'} S_{\Tt,\Tt'}^\ast (U_\theta^\phd S)_{\Ti,\Ti'} S_{\Ti',\Ts'} \delta(\Ts'-\Tt') (U_\theta^\ast S^\ast)_{\Tj,\Tj'} S_{\Tj',\Tt'}^\ast$
    \\$= \sum_{\Ti',\Tj',\Ts'} S_{\Ts,\Ts'} S_{\Tt,\Ts'}^\ast (U_\theta^\phd S)_{\Ti,\Ti'} S_{\Ti',\Ts'} (U_\theta^\ast S^\ast)_{\Tj,\Tj'} S_{\Tj',\Ts'}^\ast$
    \\$= \frac{1}{p^2} \sum_{\Ti',\Tj'} (U_\theta^\phd S)_{\Ti,\Ti'} (U_\theta^\phd S)^\ast_{\Tj,\Tj'} \delta(\Ts-\Tt+\Ti'-\Tj')$.
    }
and takes the form
\begin{align} \begin{aligned}
    Q_{\Ti,\Ts\,;\,\Tj,\Tt} &= \sum_{\Ti',\Tj'}
    (U_\theta^\phd)_{\Ti,\Ti'} \frac{\delta(\Ts-\Tt+\Ti'-\Tj')}{p^2} (U_\theta^\ast)_{\Tj,\Tj'}
\end{aligned} \end{align}
for some unitaries $U_\theta$.

This transformation turns eigenstates of $Z$-Wilson loop operators [Eq.~\eqref{eq:Xpsi_basis}] into eigenstates of $X$-Wilson loop operators and vice versa; effectively switching between $e$ and $m$ anyons of the TQFT model.
Therefore, the action of this channel is a measurement of $X$-loop operators,
\begin{align} \label{eq:Qmap_case_m}
    Q(\theta)[\rho] = U_\theta^\phd \left( \frac{1}{p^2}\sum_{\Ts} X^\Ts \rho \, X^{-\Ts} \right) U^\dag_\theta \,.
\end{align}
As this case is related to Case~2 via unitary transformations, this phase yields CTO states.

\vspace{1ex}\noindent\textbf{Case~3b.}

Now, consider the nontrivial representation
\begin{align} \label{eq:xp_case_em}
    \bigXp{ \auxX_1^{\Tt_1} \auxX_2^{\Tt_2} \cdots \auxX_{2n}^{\Tt_{2n}} }
    &= \delta(\Tt_1+\Tt_2+\dots+\Tt_{2n}) \prod_{i<j} I_{\Tt_i,\Tt_j}^c \,,
\end{align}
for $c$ an integer in $[1,p-1]$.

This case corresponds to the channel of the form.
\begin{align} \label{eq:Qmap_case_em}
    Q(\theta)[\rho] = U_\theta^\phd \left( \frac{1}{p^2}\sum_{\Ts} X^\Ts Z^{-c\Ts} \rho \, Z^{c\Ts} X^{-\Ts} \right) U^\dag_\theta \,.
\end{align}

Indeed, we can show that this gives the expectation tensor~\eqref{eq:xp_case_em}.\footnote{Here we only derive $\eqref{eq:Qmap_case_em} \Rightarrow \eqref{eq:xp_case_em}$ but not the converse statement.}
Comparing to~\eqref{eq:def_Ptheta}, $P$ can be written as
\begin{align}
	P_{\,\Ti,\Tk\,;\,\Tj,\Tl} &= \delta(\Tk+c\Ti) \delta(\Tl+c\Tj) \delta(\Ti-\Tj) / p^2
\end{align}
(We can ignore $U_\theta$ in this calculation since it will not affect the expectation values.)
Plugging into~\eqref{eq:P_to_Q},
\begin{align}
	p^2 Q_{\Ti,\Ts;\Tj,\Tt} &= {\textstyle\sum}_{\Tk,\Tl} P_{\, \Ti-\Ts,\Tk \,;\, \Tj-\Tt,\Tl} \, I_{\Ts,\Tk}^\ast \, I_{\Tt,\Tl}^{\vphantom\ast} \notag
\\	&= \delta(\Ti-\Ts-\Tj+\Tt) \, I_{\Ts, c(\Ts-\Ti)}^\ast I_{\Tt, c(\Ts-\Ti)} \notag
\\	&= \delta(\Tj-\Ti+\Ts-\Tt) \, I_{\Ts-\Ti,\Ts-\Tt}^c \,.
\end{align}
The expression
\begin{align} \begin{aligned}
	Q_{\Ti_i,\Tt_{2i-1} ;\, \Ti_{i+1},-\Tt_{2i}} &= \delta\bigl( \Ti_{i+1} - \Ti_i + \Tt_{2i-1}+\Tt_{2i} \bigr) \\ &\qquad \times I_{ \Tt_{2i-1}-\Ti_i , \Tt_{2i-1}+\Tt_{2i} }^c / p^2
\end{aligned} \end{align}
enforces $\Ti_{i+1} = \Ti_i - (\Tt_{2i-1}+\Tt_{2i})$ and so $\Ti_{i} = \Ti_1 + \bigl( \Tt_1+\Tt_2+\dots+\Tt_{2i-2} \bigr)$ when these are chained together as in~\eqref{eq:xp_XXXX_Qq}.
Hence
\begin{align}
	\bigXp{ \auxX_1^{\Tt_1} \cdots } &= \delta\Bigl( {\textstyle\sum}_{k=1}^{2n} \Tt_{k} \Bigr) \times \prod_{i=1}^n I_{ \Tt_1+\Tt_2+\dots+\Tt_{2i-1} , \Tt_{2i-1}+\Tt_{2i} }^c \notag
\\	&= \text{Eq.~\eqref{eq:xp_case_em}}.
\end{align}

We show, in the following section, that theses phases may be CTO or QTO, depending on $p$.

\vspace{1em}\noindent\textbf{Case~4:} $x_1^\phd x_2^{-1} \notin G$ and $x_2^\phd x_3^{-1} \in G$.

Invoking Lemma~\ref{lem:Qrestrictions}\textit{(4)}, the diagonal blocks $q_{ss}$ are rank-1 orthogonal matrices.
Let $Q^\circ$ be $Q$ written in a basis such its diagonal blocks take on the form of~\eqref{eq:orthogonal_blocks} without the unitaries.
Specifically,
$Q_{\Ti,\Ts\,;\,\Tj,\Tt} = \sum_{\Ti',\Tj'} (U_\theta)_{\Ti,\Ti'} (U_\theta^\ast)_{\Tj,\Tj'} Q^\circ_{\Ti',\Ts\,;\,\Tj',\Tt}$.
In this basis, the submatrices $q^\circ_{\Ts\Tt} = Q^\circ_{\bullet,\Ts\,;\,\bullet,\Tt}$ have rank at most 1 [cf.~\eqref{eq:column_inclusion}]
and its only potential nonzero element must occur at $(\Ti,\Tj) = (\Ts,\Tt)$ due to its positivity.
Hence $Q^\circ$ is constrained to the form
\begin{align}
	Q^\circ_{\Ti,\Ts\,;\,\Tj,\Tt} = \delta(\Ti-\Ts) \delta(\Tj-\Tt) M_{\Ts,\Tt} \,,
\end{align}
with $M_{\Ts,\Ts} = 1$ and $M \succeq 0$.

Equation~\eqref{eq:xp_XXXX_Qq} relating $\bigXp{\auxX_1^{\Tt_1}\cdots}$ to the quantum channel remains valid substituting $Q$ for $Q^\circ$.
Since $x_2^\phd x_3^{-1} \in G$,
we have $0 \neq \bigXp{\auxX_1^\Ts \auxX_2^{-\Tt} \auxX_3^\Tt \auxX_4^{-\Ts} \auxX_5^\Ts \cdots \auxX_{2n}^{-\Ts}}
= \frac{1}{D} \sum_\theta \Prob^n \tr\bigl[ q^\circ_{\Ts\Tt} q^\circ_{\Tt\Ts} (q^\circ_{\Ts\Ts})^{n-2} \bigr]
= \frac{1}{D} \sum_\theta \Prob^n |M_{\Ts,\Tt}|^2$,
and $D = \sum_\theta \Prob^n \tr (q^\circ_{00})^n = \sum_\theta \Prob^n$.
We conclude that the former expression must be $+1$, and that $|M_{\Ts,\Tt}(\theta)| = 1$.
We can always use the residual degrees of freedom in $U_\theta$ such that $M_{0,\Ts} = 1$ for all $\Ts$.
Via Lemma~\ref{lem:phaseMx_positive}, the positivity of $M$ then constrains all other elements to be unity: $Q^\circ_{\Ti,\Ts\,;\,\Tj,\Tt} = \delta(\Ti-\Ts) \delta(\Tj-\Tt)$.
Hence $Q$ always take the form
\begin{align}
	Q_{\Ti,\Ts\,;\,\Tj,\Tt}(\theta) &= (U_\theta)_{\Ti,\Ts} (U_\theta^\ast)_{\Tj,\Tt}
\end{align}
for some unitary matrix $U_\theta$.
The expectation tensor
\begin{align} \label{eq:xp_case_quantum} \begin{aligned}
&	\bigXp{ \auxX_1^{\Tt_1} \auxX_2^{\Tt_2} \cdots \auxX_{2n}^{\Tt_{2n}} }
\\	&= \delta(\Tt_2+\Tt_3) \delta(\Tt_4+\Tt_5) \cdots \delta(\Tt_{2n}+\Tt_1) .
\end{aligned} \end{align}
The boundary state has symmetry $G \cong \Z_p^n$ generated by $\bigl\{ x_2^\phd x_3^{-1}, \dots, x_{2k}^\phd x_{2k+1}^{-1}, \dots, x_{2n}^\phd x_1^{-1} \bigr\}$.

As a quantum channel,
\begin{align}
    Q(\theta)[\rho]  = U_\theta^\phd \rho \, U^\dag_\theta \,.
\end{align}
For each syndrome outcome $\theta$, the quantum channel $Q(\theta)$ can be reversed by application of the inverse unitary [along with $E_\theta^\dag$ from Eq.~\eqref{eq:Qchannel}].
Therefore, under optimal decoding and postselection, the noise channel $\mathcal{M}\E$ is reversible; the quantum information encoded in the logical space can be retrieved by a suitably constructed map recovery $\R$.
On a practical level, for this to be an effective QEC this requires the map $U_\theta$ to be efficiently computable on a classical computer---which is challenge beyond the scope of this work.
Importantly, our derivation proves that in this phase such recovery map always exists.

As we will demonstrate in the next section, this topological phase corresponds to a QTO matching that of the original pure state TC.

\section{Classifying mixed states descendants of the toric code}
\label{sec:ClassifyMixed}

\begin{table*}
\caption{Topological classification of $\Z_p$ toric code descendant mixed states} \label{tab:TCdesc_classification}
\begin{minipage}{\linewidth} $\begin{array}{c @{\qquad} c c c c c c}
\hline\hline
    & \text{Case~1} & \text{Case~2} & \text{Case~3a} & \multicolumn{2}{c}{\text{Case~3b}} & \text{Case~4}
\\ & \makebox[14ex][c]{} & \makebox[14ex][c]{} & \makebox[14ex][c]{} & \makebox[14ex][c]{$p=2$} & \makebox[14ex][c]{$p>2$ (prime)} & \makebox[14ex][c]{}
\\\hline \\[-2.8ex]
    \R\M\E(\cdot)
    & q_{00}
    & \frac{1}{p^2} \!\sum_\Ts Z^\Ts(\cdot)Z^\Ts
    & \frac{1}{p^2} \!\sum_\Ts X^\Ts(\cdot)X^\Ts
    & \multicolumn{2}{c}{ \frac{1}{p^2} \!\sum_\Ts X^\Ts Z^{-c\Ts} (\cdot) Z^{c\Ts} X^{-\Ts} }
    & (\cdot)
\\[1ex]	\text{decoherence type}
	& e+m & e & m & \psi {\,\scriptstyle(=\,em)} & e^{-c}m & -
\\[1ex]
	\DMxS{n}(\rho_\E)
	& \raisebox{0.20ex}{$\scriptstyle\bullet$} & \Delta^{p^2-1} & \Delta^{p^2-1} & \Delta^3 & \mathcal{B}(p) & \mathcal{B}(p^2)
\\[0.5ex]\hline \textbf{Boundary SPT order}\hfill
\\
	x_1^\phd x_2^{-1} \in G & \checkmark && \checkmark & \checkmark^\ast & \checkmark^\ast &
\\
	x_2^\phd x_3^{-1} \in G &&& \checkmark & \checkmark^\ast & \checkmark^\ast & \checkmark
\\
	\bigXp{ \auxX_1^{\Tt_1} \cdots \auxX_{2n}^{\Tt_{2n}} }
	& \text{Eq.~\eqref{eq:xp_case_triv}}
	& \text{Eq.~\eqref{eq:xp_case_e}}
	& \text{Eq.~\eqref{eq:xp_case_m}}
	& \multicolumn{2}{c}{\text{Eq.~\eqref{eq:xp_case_em}}}
	& \text{Eq.~\eqref{eq:xp_case_quantum}}
\\\hline \textbf{Wilson loops}\hfill
\\
	\qquad ( X^\Ts )_\E \hfill &&& \text{C} &&& \text{Q}
\\
	\qquad ( Z^\Ts )_\E \hfill && \text{C} &&&& \text{Q}
\\
	\qquad ( X^{\Ts} Z^{b\Ts} )_\E \quad{\scriptstyle(0<b<p)} \hfill
	&&&& \text{C} & -/\text{Q}^\dag & \text{Q}
\\\hline
    \text{\parbox[c][8mm]{30mm}{Mixed state\\topological order}}
    & \text{Trivial} & \text{Classical} & \text{Classical} & \text{Classical} & \text{Quantum} & \text{Quantum}
\\\hline\hline
\end{array}$ \end{minipage}
\\[1ex]
\begin{minipage}{0.95\linewidth} \begin{flushleft}
On the row $\DMxS{n}(\rho_\E)$, 
	the dot denote a set with a single point, $\Delta^k$ is a $k$-simplex,
	$\mathcal{B}(d)$ is the set of normalized density matrices in a $d$-dimensional Hilbert space [cf.\ Eq.~\eqref{eq:qudit_space}].
In the rows under Boundary SPT order, a checkmark denote that the group element is in the boundary symmetry group $G$.
The checkmark star indicates that the group element has nontrivial cocycles (i.e., SPT order).
In the rows under Wilson loop operators, C and Q denote that the operator is classical and quantum operator respectively.
For the entry with $\dag$, the operator is quantum if $b = c$ and non-valid if $b \neq c$.
\end{flushleft} \end{minipage}
\end{table*}

In the previous section, we show the one-to-one correspondence between the boundary SPT phases and the code space quantum channels for descendants of the $\Z_p$ toric code.
In this section, we will use this correspondence as a bridge to connect two definitions of the mixed state topological order we present in Sec.~\ref{sec:TopoMixed}.
Denote $\rho_\E = \M\E(\rho_1)$ to be the mixed state that results after applying the error and measurement channels to our initial state.
(Henceforth in this section, we will ``absorb'' $\M$ into the definition of $\E$ to simplify notation.)
Recall that the first definition is based on the geometric structure of the locally indistinguishable density matrix space $\DMxS{n}(\rho_\E)$, while the second one determines the topological order based on the properties of the Wilson loop operators acting on $\DMxV{n}(\rho_\E)$.

A key feature of our analysis relies on the property that the combined error and measurement process $\E$ is \textbf{sujective}, that is, the image of $\E$ acting on the TC $\DMxS{n}(\rho_\text{TC})$ is equals to $\DMxS{n}(\rho_\E)$.
More colloquially, it means that every mixed state that is locally indistinguishable from $\rho_\E$ can be constructed via $\E$ starting from some state in the original TC logical space.

The geometric definition of the mixed state topological order is directly related to the code space quantum channel.
In our postselection-based error correcting process, the initial pure state topological ordered space $C_1 = \DMxS{n}(\rho_\text{TC})$ is mapped via the channel $\E$ to a noisy mixed state space $C_2 = \DMxS{n}(\rho_\E) = \E(C_1)$, which is the space we wish to characterize.
The decoding and recovery steps of QEC map $C_2$ back to the original code space $C_1$.
If the decoding and recovery are optimal, then $\R: C_2 \to C_1$ is injective.
(See Fig.~\ref{fig:qubit_geo} for an illustration of this idea.)
Thus the image of optimal error correction process $\R\circ\E$ is equivalent to $C_2$, and revealing the topological order of $\rho_\E$.

The connection between the boundary SPT order and the Wilson loop operator definition of topological order is more subtle.
Unlike the toric code (and other exactly solvable models), the Wilson loop operators of the noisy mixed states $\rho_\E$, if they exist, usually take complicated forms and may be difficult to find.
Our strategy is to track evolution of Wilson loop operators from the TC model, and see how they behaves after passing through $\E$ into the mixed state.
To be specific, we take a Wilson loop operator $W$ of the TC model~(Fig.~\ref{fig:TC_Wilson}), and ask if we can define an operator $W_\E$ of the mixed state such that $\E(W R) = W_\E \E(R)$.
This turns out to be the same as testing for the condition $\E(WR_1) = \E(WR_2) \text{ whenever } \E(R_1) = \E(R_2)$.
If $W_\E$ is well-defined, we then want to determine if the operator is classical or quantum, to arrive at a classification of the phase.

We show these two definitions yield the same topological classification of the descendants of TC model.
This establishes---at least for the examples studied here---an equivalence between the two definitions of topological order from Sec.~\ref{sec:TopoMixed}.
Our results are summarized in Table~\ref{tab:TCdesc_classification}.
We conclude this section with remarks regarding specific microscopic realizations of the mixed states, and discusses the relationship between our work and anyon condensation.

\subsection{Geometric classification}
\label{sec:ClassificationGeo}

Here, for each boundary SPT class, we consider the image of the map
\begin{align}
	\M\E(\cdot) = \sum_\theta \Prob(\theta) \, Q_\theta(\cdot)
\end{align}
which takes a state from the toric code $\DMxS{n}(\rho_\text{TC})$ into $\DMxS{n}(\rho_\E)$.
While the exact form of $Q_\theta[\rho]$ depends on the error syndrome $\theta$, the maps are all equivalent up to a quantum channel%
	\footnote{For any two $\theta$, $\theta'$, there exists a (nonlocal) channel $C_{\theta,\theta'}$ such that $Q_\theta(\cdot) = F_{\theta,\theta'} Q_{\theta'}(\cdot)$.}
	and so it suffices to examine a single $Q_\theta$ as a representative of the entire phase $\rho_\E$.

\vspace{1ex}\noindent
\textbf{Case~1}.
$Q(\rho) = q_{00}$.
In this case, the code space channel maps all the state to a fixed point, akin to the complete depolarization channel.
Thus, $\DMxS{n}(\rho_\text{TC})$ is a single point and this space and the state has trivial topological order.

\vspace{1ex}\noindent
\textbf{Case~2} and \textbf{Case~3a}.
Both of these cases (up to a unitary) can be written in the form
\begin{align} \label{eq:Q_Kraus_form}
	Q(\rho) = \frac{1}{p^2} \sum_\mu K_\mu^\phd \, \rho \, K_\mu^\dag
\end{align}
where the Kraus operators $K_\mu$ form a complete projective measurement.
Effectivly, the channel dephases the mixed state in the the $Z$- (case~2) or $X$- (case~3a) basis,
the resulting state is always a ``classical sum'' of orthogonal states.

The geometry of the resulting space is a $(p^2-1)$-simplex $\Delta^{p^2-1}$.
States are characterized by $p^2$ nonnegative coefficients summing to 1.
These phases is CTO.

\vspace{1ex}\noindent
\textbf{Case~3b} for $p = 2$.
This case is identical to Cases~2 and~3a, except dephasing occurs in the $XZ$-basis.
The resulting space is $\Delta^3$, that is, a tetraherdron.

\vspace{1ex}\noindent
\textbf{Case~3b} for $p \geq 3$ prime.
This channel (up to a unitary) also takes the Kraus form~\eqref{eq:Q_Kraus_form}.
Notice, in contrast to the previous cases, that the Kraus operators $X^\Ts Z^{-c\Ts}$ do not commute among each other,
\begin{align} \label{eq:noncommute_case3b}
	\bigl( X^\Ts Z^{-c\Ts} \bigr) \bigl( X^\Tt Z^{-c\Tt} \bigr) &= I_{\Ts,\Tt}^{-2c} \bigl( X^\Tt Z^{-c\Tt} \bigr) \bigl( X^\Ts Z^{-c\Ts} \bigr) .
\end{align}
In fact, the Kraus operators form the $p \times p$ matrix ring over complex numbers (that is, the algebra of $p \times p$ matrices).

We can show this explicitly by a unitary transformation.
In the basis of $\ket{j_1\Th+j_2\Tv}$ the Wilson loop operators are represented by%
	\footnote{The Pauli matrices obey the algebra $Z^p = X^p = {\bf1}$ and $ZX = \zeta_p XZ$.}
\begin{align} \begin{aligned}
	Z^\Th &\cong {\bf1} \otimes Z ,
&&&	X^\Th &\cong X \otimes {\bf1} ,
\\	Z^\Tv &\cong Z^{-1} \otimes {\bf1} ,
&&&	X^\Tv &\cong {\bf1} \otimes X .
\end{aligned} \end{align}
Let $\bar{a}a \equiv 1 \pmod{p}$ (so for example, $\bar2 \equiv \frac{p+1}{2}$ such that $\bar2 + \bar2 \equiv 1$).
Consider the unitary transformation $O \mapsto UOU^\dag$ that maps
\begin{align} \begin{aligned}
	Z \otimes {\bf1} &\mapsto X^{\bar2\bar{c}} \otimes X^{\bar2\bar{c}}
,\\
	X \otimes {\bf1} &\mapsto Z^{-c} \otimes Z^{-c}
,\\
	{\bf1} \otimes Z &\mapsto Z \otimes Z^{-1}
,\\
	{\bf1} \otimes X &\mapsto X^{\bar2} \otimes X^{-\bar2}
.
\end{aligned} \end{align}
This map takes $X^\Th Z^{-c\Th} \mapsto Z^{-2c} \otimes {\bf1}$ and $X^\Tv Z^{-c\Tv} \mapsto X \otimes {\bf1}$.
The operators act on a product of two $p$-dimensional Hilbert spaces, the combined space is $p^2$-dimensional.
The error channel takes the form
\begin{align}
	\frac{1}{p^2} \sum_{a,b=1}^{p} (X^a Z^b \otimes {\bf1}) \, \rho \, (X^a Z^b \otimes {\bf1})^\dag .
\end{align}
This is a partial depolarization, completely erasing information stored in the first subspace while maintaining information stored in the second:
$\rho \to \frac{1}{p}{\bf1} \otimes \Tr_1 \rho$.

The space $\DMxS{n}(\rho_\E)$ is isomorphic to $\mathcal{B}(p)$, the space of states in a $p$-dimensional Hilbert space [see Eq.~\eqref{eq:qudit_space}].
As quantum information can still be retained, this is a QTO phase!

\vspace{1ex}\noindent
\textbf{Case~4}.
Up to a unitary, $Q(\rho) = \rho$.
Since the resulting space of the code space channel is exactly the same with the original logical space of the TC, with geometry $\mathcal{B}(p^2)$, this case is a QTO phase.

\subsection{Wilson loop classification}
\label{sec:ClassificationWL}

\begin{figure}
	\centering
	\includegraphics[width = 0.9\linewidth]{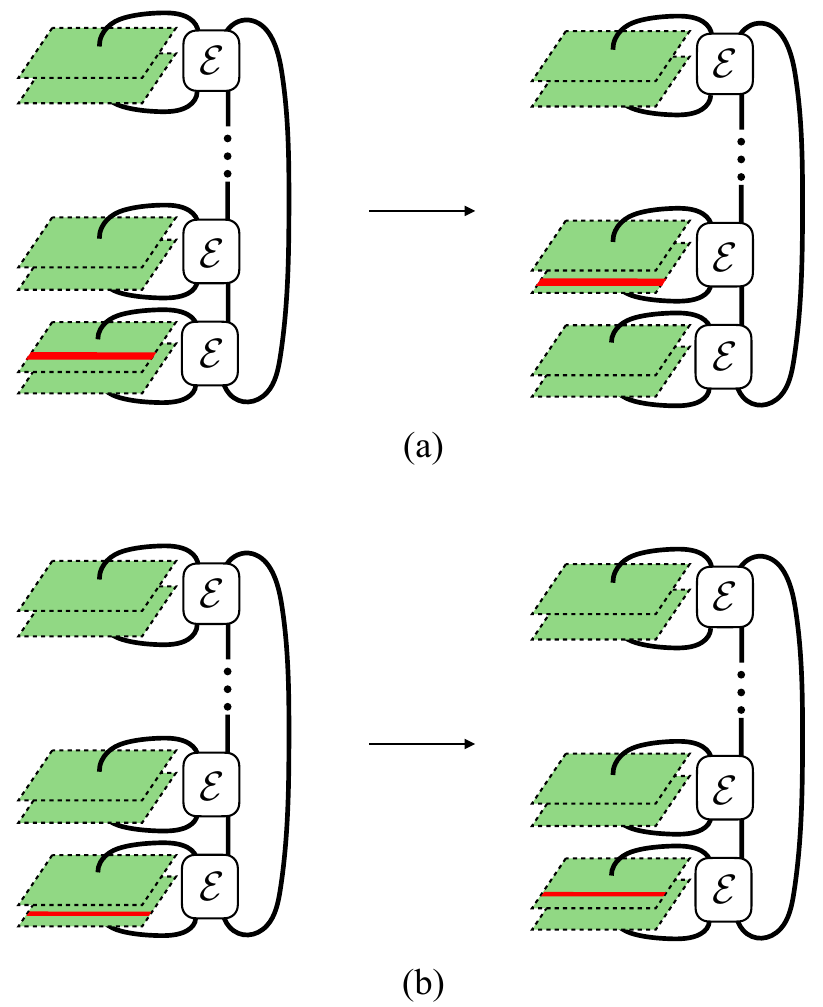}
	\caption{
	The boundary symmetry group can be visualized as operator jumping. We claim that 
	(a) $W_\E$ associate with a Wilson loop operator if the boundary operator $W$ can jump between replicas. (b) $W_\E$ associate with a classical Wilson loop operator if the boundary operator $W$ can jump within a replica.}
	\label{fig:Wjump}
\end{figure}

\begin{figure*}
	\centering
	\includegraphics[width = 0.99\linewidth]{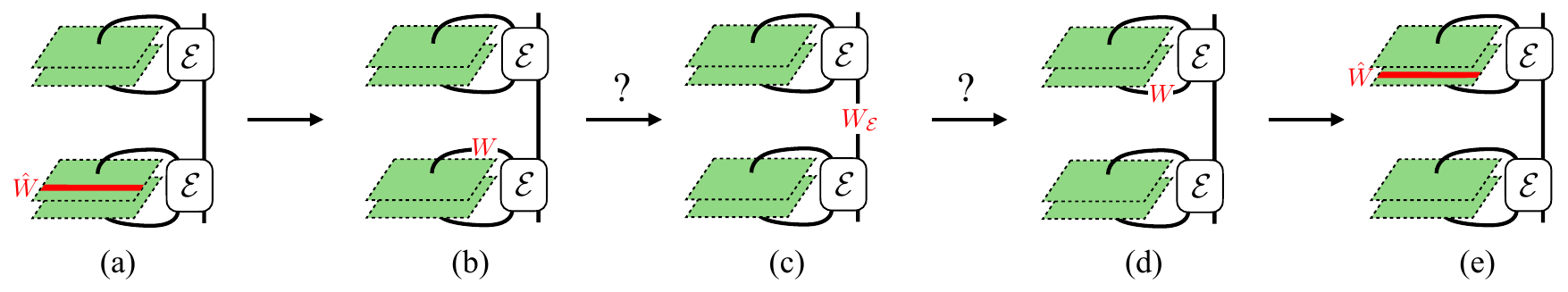}
	\caption{%
Here we break the auxiliary operator's jump into five steps.
The jumping between the auxiliary space and the physical space, (a) to (b) and (d) to (e), are always valid.
If the operator $W_\E$ is well-defined, then, by definition, it commute with the quantum channel.
That is, $\E(W R) = W_\E\E(R)$.
Therefore, if the intermediate operator $W_\E$ is well-defined, the Wilson loop operator $W$ can jump between replicas.
	} \label{fig:Wjump2}
\end{figure*}

In this section, we look for Wilson loop operators of in the mixed state phase of $\rho_\E$.
Specifically, we ask if a Wilson loop operator $W$ of the parent (pure state) phase remains well-defined after the decoherence channel; and if so the determine the operator type.
Formally, we ask if the map $W_\E$ can be defined to make the diagram commutative.
\begin{align}
	\begin{tikzpicture}
	[baseline={([yshift=-.5ex]current bounding box.center)},every text node part/.style={align=center}],
	\node[] (0) at (0pt,20pt) { $\DMxV{n}(\rho_\text{TC})$};
	\node[] (1) at (80pt,20pt) { $\DMxV{n}(\rho_\E)$};
	\node[] () at (40pt,26pt) { $\E$};
	\node[] () at (40pt,-24pt) { $\E$};
	\node[] () at (-8pt,-5pt) { $W$};
	\node[] () at (92pt,-5pt) { $W_\E$};
	\node[] (3) at (80pt,-30pt) { $\DMxV{n}(\rho_\E)$};
	\node[] (2) at (0pt,-30pt) { $\DMxV{n}(\rho_\text{TC})$};
	\draw [->,shorten >= 2pt, shorten <= 2pt] (0) edge (1) ;
	\draw [dashed,->,shorten >= 2pt, shorten <= 2pt] (1) edge (3) ;
	\draw [->,shorten >= 2pt, shorten <= 2pt] (0) edge (2) (2) edge (3);
	\end{tikzpicture} 
\end{align}
If so, we say that $W_\E$ is the \emph{induced} Wilson loop operator from $W$.

The operator $W_\E$ is well-defined if $\E(WR_1) = \E(WR_2) \text{ whenever } \E(R_1) = \E(R_2)$, that is, it consistently takes the same input to the same output.
When characterizing a Wilson loop operator, we are interested if it is proportional to the identity operator, and if it is classical or quantunm as defined in Sec.~\ref{sec:TopoMixedWL}.
The determination of these properties amounts to testing for (in)equalities between expressions of the form $\E(\cdot)$ via the $n$-Schatten norm.
When $n$ is an even integer, the norm $\bigl\lVert \E(R) \bigr\rVert_n$ consists of $n/2$ powers of $|\E(R)|^2 = \E(R^\dag) \E(R)$:
\begin{align}
	\bigl\lVert \E(R) \bigr\rVert_n &= \Tr\bigl[ \E(R^\dag) \, \E(R) \, \E(R^\dag) \cdots \E(R) \bigr]^{1/n} .
\end{align}
This is precisely the type of expression that is produced by the $n$-replica network,
and is completely characterized by the $X$-loop expectation tensor, which in turn is determined by the boundary SPT order.
To cement this connection, we introduce the following \textbf{Wilson loop tests}.

\vspace{1ex}\noindent
Suppose $W \in \DMxV{n}(\rho_\text{TC})$ is a unitary Wilson loop operator of the parent theory.
\begin{enumerate}[itemsep=0pt, parsep=2pt]
\item[W1.] $W_\E$ is a well-defined Wilson loop operator if for all $R_i \in \DMxV{n}(\rho_\text{TC})$,
	\begin{align} \label{eq:WilsonCondition} \quad \begin{aligned}
	& \Tr\bigl[ \E(R_1{W^\dag}) \, \E(W R_2) \, \E(R_3) \cdots \E(R_n) \bigr] 
	\\&\mkern90mu = \Tr\bigl[ \E(R_1) \, \E(R_2) \cdots \E(R_n) \bigr] .
	\end{aligned} \end{align}
	This condition means the Wilson loop operator $W$ can jump between replicas, as shown in Fig.~\ref{fig:Wjump}(a).
\item[W2.] $W_\E$ is a non-identity Wilson loop operator if $W_\E$ is a Wilson loop operator,
	and there exists a nonzero $R \in\mathcal{V}^{n}(\rho)$ such that
	\begin{align} \label{eq:WilsonCondition_NonId}
	\Tr\bigl[ \E(R^\dag) \, \E(WR) \, |\E(R)|^{n-2} \bigr] = 0 .
	\end{align}
	If $R$ is Hermitian, if suffices to test if $\Tr\bigl[ \E(WR) \, \E(R)^{n-1} \bigr]$ vanishes.
\item[WC.] $W_\E$ is a classical Wilson loop operator if $W_\E$ is a Wilson loop operator and for all $R_i \in\mathcal{V}^{n}(\rho_\text{TC})$,
	\begin{align} \label{eq:WilsonCondition_Classical} \begin{aligned}
	\quad& \Tr\bigl[ \E(W R_1 W^\dag) \, \E(R_2) \cdots \bigr]
	\\&\mkern130mu = \Tr\bigl[ \E(R_1) \, \E(R_2) \cdots \bigr] .
	\end{aligned} \end{align}
	This condition also means the Wilson loop operator $W$ can jump within a replica, as shown in Fig.~\ref{fig:Wjump}(b).
\item[WQ.]
	Suppose $(W_1)_\E$, $(W_2)_\E$ are well-defined nonzero operators.
	If $W_1 W_2 = \gamma W_2 W_1$ for some phase $\gamma \neq 1$, then $(W_1)_\E$ and $(W_2)_\E$ are both non-identity operators and quantum.
\end{enumerate}

We first show why these four conditions are true.
As examples, we then apply these conditions to Case~3 and Case~4 of the $\Z_p$ toric code model to classify their topological order.
The results regarding Wilson loops are summarized in Table~\ref{tab:TCdesc_classification}.

\begin{widetext}
To show W1, suppose $W$ satisfies the first condition~\eqref{eq:WilsonCondition} and $n \in 2\Z$.
Thus,
\begin{align}
\bigl\lVert \E(W R_1) - \E(W R_2) \bigr\lVert_n^n
& = \Tr\Bigl[ \bigl( \E(R_1^\dag W^\dag - R_2^\dag W^\dag) \, \E(W R_1 - W R_2) \bigr)^{n/2} \Bigr]
\notag\\ &\; = \Tr\Bigl[ \bigl( \E(R_1^\dag - R_2^\dag) \, \E(R_1 - R_2) \bigr)^{n/2} \Bigr]
= \bigl\lVert \E(R_1) - \E(R_2) \bigr\lVert_n^n \,.
\end{align}
We use the linearity of $\E$ for the first equality, and then~\eqref{eq:WilsonCondition} to eliminate $(W^\dag, W)$ pairs.
Hence $\E(R_1) = \E(R_2) \Rightarrow \E(W R_1) = \E(W R_2)$, and $W_\E$ is a well-defined map.
The map is also linear: $W_\E (\E(R_1)+\E(R_2)) = W_\E \E(R_1) + W_\E \E(R_2)$ and $(W_1)_\E + (W_2)_\E = (W_1+W_2)_\E$ since $\E$ is itself linear.

To show W2, now assume that $W$ satisfies both Eq.~\eqref{eq:WilsonCondition} and~\eqref{eq:WilsonCondition_NonId}.
Then,
\begin{align}
&	\bigl\lVert W_\E\E(R)-\alpha\E(R) \bigr\rVert_n^2 \, \bigl\lVert \E(R) \bigr\rVert_n^{n-2}
\notag\\ &\quad \geq \Tr\Bigl[ \bigl|\E(WR)-\alpha\E(R)\bigr|^2 \, \bigl|\E(R)\bigr|^{n-2} \Bigr] ,
	\mkern70mu \text{\scriptsize(via Corollary~\ref{lem:nSchatten_trace_ineq})}
\notag\\ &\quad = \Tr\Bigl[ \Bigl(\!
	\mathord{\underbrace{ \E(R^\dag W^\dag) \E(W R) }_\text{ $\E(R^\dag) \E(R)$ via~\eqref{eq:WilsonCondition} }}
	- \alpha \, \mathord{\underbrace{ \E(R^\dag W^\dag) \E(R) }_\text{ $0$ via~\eqref{eq:WilsonCondition_NonId} conj. }}
	- \alpha^\ast \mathord{\underbrace{ \E(R^\dag) \E(WR) }_\text{ $0$ via~\eqref{eq:WilsonCondition_NonId} }}
	+ |\alpha|^2 \E(R^\dag) \E(R)
	\Bigr) |\E(R)|^{n-2} \Bigr]
\notag\\ &\quad = \bigl(1+|\alpha|^2\bigr) \Tr\bigl[ |\E(R)|^n \bigr] \quad > 0 .
\end{align}
Thus $W_\E \E(R) \neq \alpha\E(R)$ (for all $\alpha$); $W_\E$ is a non-identity operator.

To show WC, conditions~\eqref{eq:WilsonCondition} and~\eqref{eq:WilsonCondition_Classical} implies that $W$ and $W^\dag$ can be passed between the bra/ket side of any replica.
$\Tr \E(WR_1) \dots = \Tr \cdots \E(WR_i) \cdots = \Tr \cdots \E(R_iW) \cdots$.
Hence for $n \in 2\Z$,
\begin{align}
\bigl\lVert W_\E^\phd \E(R) W_\E^\dag - \E(R) \bigr\rVert_n^n
& = \bigl\lVert \E(W\!RW^\dag)-\E(R) \bigr\rVert_n^n
= \Tr\Bigl[ \bigl( \bigl[\E(W\!R^\dag W^\dag)-\E(R^\dag)\bigr] \bigl[\E(W\!RW^\dag)-\E(R)\bigr] \bigr)^{\tfrac{n}{2}} \Bigr]
= 0 .
\end{align}

Finally to show WQ, assume that $W_1 W_2 = \gamma W_2 W_1$ for $\gamma \neq 1$.
Then,
\begin{align}
	(W_1)_\E (W_2)_\E \E(R) = \E(W_1 W_2 R) = \gamma \E(W_2 W_1 R) = \gamma (W_2)_\E (W_1)_\E \E(R)
\end{align}
and hence $(W_1)_\E$, $(W_2)_\E$ do not commute.
Invoking Prop.~\ref{prop:noncommuting_W}, both operators must be quantum.
\end{widetext}

Now, let's apply these tests to the $\Z_p$ toric code model.
Since $\DMxV{n}(\rho_\text{TC})$ is spanned by states $X^\Tj \rho_{00} X^\Tk$ where $\rho_{00} = \ketbra{0}{0}$,
and the expression $\Tr\bigl[ \E(R_1) \cdots \E(R_n) \bigr]$ is linear in each argument, it suffices to check the three conditions against the set of basis states.
Of course, Eq.~\eqref{eq:xp_XXXX_Qq} relates the trace to the $X$-loop expectation tensor.
\begin{align}
	\Tr\bigl[ \E(X^{\Tt_1}\rho_{00}X^{\Tt_2}) \, \E(X^{\Tt_3}\rho_{00}X^{\Tt_4}) \cdots \bigr] = \bigXp{ \auxX_1^{\Tt_1} \auxX_2^{\Tt_2} \cdots } .
\end{align}

\vspace{1ex}\noindent\textbf{Case~3}.
Recall that the $X$-loop expectation tensor for this case is given by~\eqref{eq:xp_case_em}:
\begin{align*}
    \bigXp{ \auxX_1^{\Tt_1} \auxX_2^{\Tt_2} \cdots \auxX_{2n}^{\Tt_{2n}} }
    &= \delta(\Tt_1+\Tt_2+\dots+\Tt_{2n}) \prod_{i<j} I_{\Tt_i,\Tt_j}^c
\end{align*}
for $0 \leq c < p$.
We check if $W = X^{a\Ts} Z^{b\Ts}$ satisfies conditions W1, W2, WC, and WQ above.

Evaluating the LHS of Eq.~\eqref{eq:WilsonCondition},
\begin{align} \label{eq:W1test_case3}
	& \Tr\bigl[ \E(R_1 W^\dag) \E(W R_2) \E(R_3) \cdots \bigr]
	\notag\\ &\;
	= I_{\Ts,\Tt_2+\Tt_3}^b \tr\Bigl[ \E\bigl( X^{\Tt_1} \rho_{00} X^{\Tt_2-a\Ts} \bigr) \E\bigl( X^{a\Ts+\Tt_3} \rho_{00} X^{\Tt_4} \bigr) \cdots \Bigr]
	\notag\\ &\;
	= I_{\Ts,\Tt_2+\Tt_3}^b \bigXp{ \auxX_1^{\Tt_1} \auxX_2^{\Tt_2-a\Ts} \auxX_3^{a\Ts+\Tt_3} \auxX_4^{\Tt_4} \cdots \auxX_{2n}^{\Tt_{n}} }
	\notag\\ &\;
	= I_{\Ts,\Tt_2+\Tt_3}^b I_{-a\Ts,\Tt_3}^c I_{\Tt_2,a\Ts}^c \bigXp{ \auxX_1^{\Tt_1} \auxX_2^{\Tt_2} \cdots \auxX_{2n}^{\Tt_{n}} }
	\notag\\ &\;
	= I_{\Ts,\Tt_2+\Tt_3}^{b-ac} \Tr\bigl[ \E(R_1) \E(R_2) \E(R_3) \cdots \bigr] .
\end{align}
Via the Wilson loop algebra~\eqref{eq:TC_WL_algebra}, we convert
$R_1 W^\dag = X^{\Tt_1} \rho_{00} X^{\Tt_2} Z^{-b\Ts} X^{-a\Ts} = I_{\Ts,\Tt_2}^b X^{\Tt_1} \rho_{00} X^{\Tt_2-a\Ts}$
and $W R_2 = X^{a\Ts} Z^{b\Ts} X^{\Tt_3} \rho_{00} X^{\Tt_4} = I_{\Ts,\Tt_3}^b X^{a\Ts + \Tt_3} \rho_{00} X^{\Tt_4}$ from the first to second line,
From the third to fourth line, we use the explicit form of the expectation tensor [Eq.~\eqref{eq:xp_case_em}].
Hence, we see the $W_\E$ (induced from $X^{a\Ts}Z^{b\Ts}$) are well-defined Wilson loop operators if $b \equiv ac \pmod{p}$.
In other words, $W_\E$ must be of the form $( X^\Ts Z^{c\Ts} )_\E^a$.

$( X^{a\Ts} Z^{b\Ts} )_\E$ are non-identity Wilson loop operators from the W2 test.
Let $R = \rho_{00}$, $\Tr\bigl[ \E( X^{a\Ts} Z^{b\Ts} R) \E(R)^{n-1} \bigr]$ vanishes when $a \not\equiv 0$ due to the delta function in the expectation tensor.

With a similar calculation, we evaluate the LHS of Eq.~\eqref{eq:WilsonCondition_Classical},
\begin{align}
	& \notag
	\Tr\bigl[ \E(W R_1 W^\dag) \E(R_2) \cdots \bigr]
	\\\notag &\;
	= I_{\Ts,\Tt_1+\Tt_2}^b \bigXp{ \auxX^{a\Ts+\Tt_1} \auxX^{\Tt_2-a\Ts} \auxX^{\Tt_3} \cdots \auxX_{2n}^{\Tt_{n}} }
	\\ &\;
	= I_{\Ts,\Tt_1+\Tt_2}^{b+ac} \Tr\bigl[ \E(R_1) \E(R_2) \cdots \bigr] .
\end{align}
Thus, $W_\E$ are classical Wilson loop operators when $b+ac \equiv 0$ or equivalently $2ac \equiv 0 \pmod{p}$. 
Moreover, from the calculation~\eqref{eq:noncommute_case3b}, we see that $X^{\Ts} Z^{c\Ts}$ for different directions $\Ts$ fail to commute when $2c \not\equiv 0 \pmod{p}$.
We conclude that all Wilson loops of the phase $\rho_\E$ are classical if $2c \equiv 0 \pmod{p}$, and at least one is quantum if $2c \not\equiv 0 \pmod{p}$.

Therefore, the phase is CTO if and only if $2c$ is a multiple of $p$.

\vspace{1ex}\noindent\textbf{Case~4}.
Recall that the $X$-loop expectation tensor for the Case~4 is given by~\eqref{eq:xp_case_quantum}:
\begin{align*}
	\bigXp{ \auxX_1^{\Tt_1} \cdots \auxX_{2n}^{\Tt_{2n}} } = \delta(\Tt_2+\Tt_3) \delta(\Tt_4+\Tt_5) \cdots \delta(\Tt_{2n}+\Tt_1) .
\end{align*}
For an arbitrary operator $W = X^{a\Ts}Z^{b\Ts}$, we have
\begin{align}
	& \Tr\bigl[ \E(R_1 W^\dag) \, \E(W R_2) \cdots \bigr]
	\notag\\ &\;
	= I_{\Ts,\Tt_2+\Tt_3}^b \bigXp{ \auxX_1^{\Tt_1} \auxX_2^{\Tt_2-a\Ts} \auxX_3^{a\Ts+\Tt_3} \auxX_4^{\Tt_4} \cdots \auxX_{2n}^{\Tt_{n}} }
	\notag\\ &\;
	= \bigXp{ \auxX_1^{\Tt_1} \auxX_2^{\Tt_2} \auxX_3^{\Tt_3} \auxX_4^{\Tt_4} \cdots \auxX_{2n}^{\Tt_{n}} }
	\notag\\ &\;
	= \Tr\bigl[ \E(R_1) \, \E(R_2) \cdots \bigr] .
\end{align}
The first two lines follows an identical calculation to Eq.~\eqref{eq:WilsonCondition} of Case~3; the differences manifest when we use the $X$-loop expectation tensor to simplify the expression.
For example, the term $\delta(\Tt_2+\Tt_3)$ eliminates $I_{\Ts,\Tt_2+\Tt_3}$ from the expression---a direct consequence of the fact that $x_2^\phd x_3^{-1} \in G$ and has trivial cocylce.
Hence $(X^{a\Ts}Z^{b\Ts})_\E$ is always a Wilson loop operator.
Moreover, because these operators generally do not commute, all these operators (except for when $a \equiv b \equiv 0$) are classified as quantum.

\subsection{Remarks}
\label{sec:TCdescRemarks}

\noindent\textbf{Classical topologically ordered phases}.
For $p=2$, Cases~2, 3a, and~3b are all CTO phases that results from measurement of Wilson loops of types $e$, $m$, and fermion $\psi = em$ respectively.
In all three scenarios, the density matrix space $\DMxS{n}$ is a tetrahedron.
Cases~2 and~3a can be physically realized via the $Z$-flip/dephasing, and $X$-flip respectively~\cite{2023BaoBoundarySPT, DiagnosticsRelativeEntropyNegativity}.
Because the TC has an automorphism mapping between the bosons $e \leftrightarrow m$, their phases are related by finite-depth unitary circuits.
However, reference~\onlinecite{wang2023fcondensation} finds that the Case~3b (fermion decoherence) phase has non-vanishing topological entanglement negativity,
	in contrast with the $e$- and $m$-decoherence phases.
This means that a finer classification exists beyond our geometric definition topological mixed states.

\vspace{1ex}\noindent\textbf{Chiral phases}.
Case~3b for $p > 2$ is very interesting.
When $\gcd(2c,p) = 1$, the $\Z_p$ toric code can be decomposed into a product of two separate TQFTs,
	with anyons $\mathcal{C}_+ = \{ 0, e^c m, e^{2c}m^2, \dots \}$ and $\mathcal{C}_- = \{ 0, e^{-c} m, e^{-2c}m^{-2}, \dots \}$.
These theories are chiral if $p \equiv 3 \pmod{4}$ or if the Legendre symbol $(c|p) = -1$.
The error channel depolarizes anyons in the $\mathcal{C}_-$ sector, but leaves the subtheory $\mathcal{C}_+$ intact.
The resulting QTO phase has the same geometric structure and Wilson loop algebra as that of the $\mathcal{C}_+$ theory.
This suggests that the mixed state phase is the same as the pure state $\mathcal{C}_+$ phase, in that one can find a continuous deformation connecting the two.
(And if not, it means that there are finer classification of mixed states to be discovered!)

Reference~\onlinecite{Haah_2023} provides an explicit example of operators that can generate such mixed states.
Consider a square lattice with a $p$-dimensional Hilbert space on each edge (the standard Hilbert space for the $\Z_p$ TC).
Define the Pauli operators
\begin{subequations} \begin{align} \label{eq:JH_generators1}
\JHh^\pm_{0,0} \;&=\;
\scalebox{1}{%
	\begin{tikzpicture}
	[baseline={([yshift=-.5ex]current bounding box.center)}],
	\draw[color=black!100,line width=1.pt] (-65pt,30pt) -- (35pt,30pt);
	\draw[color=black!100,line width=1.pt] (-65pt,0pt) -- (35pt,0pt);
	\draw[color=black!100,line width=1.pt] (-65pt,-30pt) -- (35pt,-30pt);
	\draw[color=black!100,line width=1.pt] (30pt,-35pt) -- (30pt,35pt);
	\draw[color=black!100,line width=1.pt] (0pt,-35pt) -- (0pt,35pt);
	\draw[color=black!100,line width=1.pt] (-30pt,-35pt) -- (-30pt,35pt);
	\draw[color=black!100,line width=1.pt] (-60pt,-35pt) -- (-60pt,35pt);
	\filldraw[black] (0pt,0pt) circle (2pt);
	\draw[color=white!100,line width=2.pt] (0pt,11pt) -- (0pt,20pt);
	\draw[color=white!100,line width=2.pt] (-30pt,11pt)--(-30pt,20pt);
	\draw[color=white!100,line width=2.pt](0pt,-12pt) -- (0pt,-21pt);
	\draw[color=white!100,line width=2.pt](-30pt,-12pt) -- (-30pt,-21pt);
	\draw[color=white!100,line width=2.pt] (6pt,0pt) -- (26pt,0pt);
	\draw[color=white!100,line width=2.pt] (-34pt,0pt) -- (-54pt,0pt);
	\draw[color=white!100,line width=2.pt] (-6pt,0pt) -- (-22pt,0pt);
	\node[black!40!red] () at (5pt,16pt){$X^{\pm1}$};
	\node[black!40!red] () at (16pt,1pt){$X^{\pm1}$};
	\node[red] () at (5pt,-16pt){$X^{\mp1}$};
	\node[blue] () at (-14pt,1pt){${Z}^{2c}$};
	\node[black!40!red] () at (-25pt,16pt){$X^{\pm1}$};
	\node[red] () at (-44pt,1pt){$X^{\mp1}$};
	\node[black!40!red] () at (-25pt,-16pt){$X^{\pm1}$};
	\end{tikzpicture} }\;,
\end{align}
and
\begin{align} \label{eq:JH_generators2}
	\JHv^\pm_{0,0}
\;&=\;
	\scalebox{1}{\begin{tikzpicture}
	[baseline={([yshift=-.5ex]current bounding box.center)}],
	\draw[color=black!100,line width=1.pt] (-35pt,60pt) -- (35pt,60pt);
	\draw[color=black!100,line width=1.pt] (-35pt,30pt) -- (35pt,30pt);
	\draw[color=black!100,line width=1.pt] (-35pt,0pt) -- (35pt,0pt);
	\draw[color=black!100,line width=1.pt] (-35pt,-30pt) -- (35pt,-30pt);
	\draw[color=black!100,line width=1.pt] (30pt,-35pt) -- (30pt,65pt);
	\draw[color=black!100,line width=1.pt] (0pt,-35pt) -- (0pt,65pt);
	\draw[color=black!100,line width=1.pt] (-30pt,-35pt) -- (-30pt,65pt);
	\filldraw[black] (0pt,30pt) circle (2pt);
	\draw[color=white!100,line width=2.pt] (0pt,11pt) --(0pt,20pt);
	\draw[color=white!100,line width=2.pt](0pt,-10pt) -- (0pt,-19pt);
	\draw[color=white!100,line width=2.pt] (6pt,0pt) --(26pt,0pt);
	\draw[color=white!100,line width=2.pt] (-6pt,0pt) -- (-26pt,0pt);
	\draw[color=white!100,line width=2.pt] (6pt,30pt) --(26pt,30pt);
	\draw[color=white!100,line width=2.pt] (-6pt,30pt) -- (-26pt,30pt);
	\draw[color=white!100,line width=2.pt] (0pt,41pt) -- (0pt,50pt);
	\node[blue] () at (3pt,16pt){${Z}^{2c}$};
	\node[red] () at (16pt,1pt){$X^{\mp 1}$};
	\node[black!40!red] () at (5pt,-14pt){$X^{\pm1}$};
	\node[red] () at (-16pt,1pt){$X^{\mp1}$};
	\node[red] () at (16pt,31pt){$X^{\mp1}$};
	\node[black!40!red] () at (-16pt,31pt){$X^{\pm1}$};
	\node[red] () at (5pt,46pt){$X^{\mp1}$};
	\end{tikzpicture} }\;.
\end{align} \end{subequations}
The black dots indicate the origin of the lattices.
$\JHh^\pm_{l,m} / \JHv^\pm_{l,m}$ are defined to be translation of $\JHh^\pm_{0,0} / \JHv^\pm_{0,0}$ by $l$ units horizontally, $m$ units vertically.
The set $\bigl\{ \JHh^+_{l,m}, \JHv^+_{l,m} \bigr\}$ commutes with the set $\bigl\{ \JHh^-_{l,m}, \JHv^-_{l,m} \bigr\}$, and together they generate the full algebra of the Hilbert space~\cite{Haah_2021, Haah_2023}.%
	\footnote{That is, every (finitely supported) lattice operator can be written as a finite combination of the Pauli operators $\JHh^\pm_{l,m}$ and $\JHv^\pm_{l,m}$.}
Construct the Hamiltonians $H_\pm = -\sum_{l,m} \sum_{k=1}^p \big(h^\pm_{l,m}\big)\phd^k$ as a sum of commuting Pauli operators
\begin{subequations} \begin{align}
	h^+_{0,0} &= \left( (\JHv_{\bar{1},0}^+)^{\dag} (\JHh_{0,\bar{1}}^+) (\JHv_{0,0}^{+}) (\JHh_{0,0}^+)^{\dag} \right)^{\bar{2}}
\notag\\
	&= 
	\scalebox{0.9}{\begin{tikzpicture}
	[baseline={([yshift=-.5ex]current bounding box.center)}],
	\draw[color=black!100,line width=1.pt] (-45pt,40pt) -- (45pt,40pt);
	\draw[color=black!100,line width=1.pt] (-45pt,0pt) -- (45pt,0pt);
	\draw[color=black!100,line width=1.pt] (-45pt,-40pt) -- (45pt,-40pt);
	\draw[color=black!100,line width=1.pt] (40pt,-45pt) -- (40pt,45pt);
	\draw[color=black!100,line width=1.pt] (0pt,-45pt) -- (0pt,45pt);
	\draw[color=black!100,line width=1.pt] (-40pt,-45pt) -- (-40pt,45pt);
	\filldraw[black] (0pt,0pt) circle (2pt);
	\draw[color=white!100,line width=2.pt] (0pt,15pt) -- (0pt,25pt);
	\draw[color=white!100,line width=2.pt](0pt,-15pt) -- (0pt,-25pt);
	\draw[color=white!100,line width=2.pt](-40pt,-15pt) -- (-40pt,-25pt);
	\draw[color=white!100,line width=2.pt] (15pt,0pt) -- (25pt,0pt);
	\draw[color=white!100,line width=2.pt] (-5pt,0pt) -- (-35pt,0pt);
	\draw[color=white!100,line width=2.pt] (-13pt,-40pt) -- (-25pt,-40pt);
	\node[red] () at (0pt,20pt){$X$};
	\node[red] () at (20pt,0pt){$X$};
	\node[] () at (0pt,-20pt){\textcolor{red}{$X^\dag$}\textcolor{blue}{$Z^{c}$}};
	\node[] () at (-20pt,1pt){\textcolor{red}{$X^\dag$}\textcolor{blue}{$Z^{-c}$}};
	\node[blue] () at (-39pt,-20pt){${Z}^{-c}$};
	\node[blue] () at (-19pt,-40pt){${Z}^{c}$};
	\end{tikzpicture}} \;
	\sim A_+^\phd B_\square^{c} \,,
\\	
	h^-_{0,0} &= \left( (\JHv_{\bar{1},0}^-)^\dag(\JHh_{0,\bar{1}}^-)(\JHv_{0,0}^{-})(\JHh_{0,0}^-)^\dag \right)^{\bar{2}}
	\sim A_+^\dag B_\square^{c} \,.
\end{align} \end{subequations}
(Here $\bar{2} = \frac{p+1}{2}$, $A_+$ and $B_\square$ are the vertex and plaquette stabilizers of the $\Z_p$ toric code, detecting the presence of $e$ and $m$ anyons respectively.)
Reference~\onlinecite{Haah_2023} points out that the Hamiltonians $H_+ / H_-$ are commuting projector Hamiltonians whose excitation spectrum realizes the chiral $\mathcal{C}_+ / \mathcal{C}_-$ theories respectively.
Physically, it means that there is an explicit set of lattice operators which realizes the splitting of the $\Z_p$ TC to $\mathcal{C}_+ \times \mathcal{C}_-$ on a microscopic level.

Each operator $O$ defines a quantum channel $\Phi[O](\sigma) = (1-t)\sigma + t O \sigma O^\dag$ using $O$ as a Kraus operator.
They commute for (product of) Pauli operators: $\Phi[P_1] \circ \Phi[P_2] = \Phi[P_2] \circ \Phi[P_1]$.
The mixed state phase is then realized by applying the set of all $\Phi\bigl[\JHh^-_{l,m}\bigr]$ and $\Phi\bigl[\JHv^-_{l,m}\bigr]$ to the toric code ground state.
At $t=1/2$, the resulting state can be written as $\lim_{\beta \to \infty} \exp\bigl[ -\beta H^+ \bigr]$; the ground state of the $\mathcal{C}^-$ sector but a maximally mixed state in the $\mathcal{C}^+$ sector.
These states (for all values of $t$) admit a 2D tensor network representation (as projected entangled pair operators) with finite bond dimension.

\vspace{1ex}\noindent\textbf{Anyon decoherence}.
It is apparent that for each of the cases studied in the present work, the error channel is of the form (or combinations of channels of the form)
\begin{align}
	Q_\theta(\rho) \propto \sum_\Ts W_\Ts(a) \, \rho \, W_{-\Ts}(a)
\end{align}
for some anyon $a$, where $W_\Ts(a)$ is the Wilson loop operator wrapping $a$ around the loop $\Ts$.
(Here we drop the $\theta$-dependent unitaries $U_\theta$.)
Hence, we can also refer to these channels as \emph{anyon-decoherence processes}, the corresponding anyon(s) for each case are listed in Table~\ref{tab:TCdesc_classification}.
This is essentially the anyon condensation picture proposed in Ref.~\onlinecite{2023BaoBoundarySPT}.
The resulting channel behaves like a dephasing channel when the Wilson loop operators among different directions commute;%
	\footnote{For an Abelian anyon $a$, its Wilson loops commute if and only if $S_{a,a} / S_{0,0} = 1$.}
	and acts like a (partial) depolarization channel when the operators span a matrix ring algebra.
Under such channel, the Wilson loop operators that commutes with all the Kraus operators survive, that is, they induces well-defined operators in $\DMxV{n}(\rho_\E)$.
The algebra generated by such operator determines the classification of the mixed state.
The phase is trivial if the only Wilson loop operator is the identity,
	CTO if all the Wilson loop operators commute,
	and QTO if the Wilson loop algebra is noncommutative.

We can generalize this observation to an arbitary Abelian theory.
Suppose the parent theory is a pure state with topological order described by the TQFT $\mathcal{C}$.
For any subset of anyons $\mathcal{F}$ which includes the identity and is closed under fusion, we can write the ``$\mathcal{F}$-decoherence'' channel
\begin{align}
	Q(\rho) = \frac{1}{|\mathcal{F}|^2} \sum_{a,b \in \mathcal{F}} W_\Th(a) W_\Tv(b) \, \rho \, W_\Tv(b)^\dag W_\Th(a)^\dag .
\end{align}
In an anyon-decoherence process, the surviving operators of the descendant theory are generated by Wilson loop operators which commutes with the Kraus operators.
A simple test for such operators is as follows: $W_{\Ts \neq 0}(c)$ commutes with all the Kraus operators if and only if $S_{a,c} = S_{0,0}$ for all $a \in \mathcal{F}$.
We can view the descendant phase as being described by a subtheory of $\mathcal{C}$ possessing only anyons which ``commutes'' with $\mathcal{F}$:
\begin{align}
	\mathrm{C}_{\mathcal{C}}(\mathcal{F}) \defeq \bigl\{ c \in \mathcal{C} \,\big|\, \forall a\in\mathcal{F}, \, S_{a,c} = S_{0,0} \bigr\} .
\end{align}
(This viewpoint assumes that the decoherence map is surjective, which may not be true if $\mathcal{C}$ is non-Abelian.)
Finally, the phase is trivial if $|\mathrm{C}_{\mathcal{C}}(\mathcal{F})| = 1$,
	CTO if the $S$-matrix restricted to $\mathrm{C}_{\mathcal{C}}(\mathcal{F})$ is elementwise positive real,
	and QTO otherwise.

For example, consider the double semion model $\mathcal{C} = \{0,s,\bar{s},b\}$, where $s/\bar{s}$ have topological spins $\pm i$ and $b = s\bar{s}$ is the boson.
There are 5 subtheories: $\{0\}$, $\{0,s\}$, $\{0,\bar{s}\}$, $\{0,b\}$, and $\mathcal{C}$ itself.
Since $\mathcal{C}$ decomposes as $\{0,s\} \times \{0,\bar{s}\}$, both $s$- and $\bar{s}$-decoherence channels would lead to a QTO mixed state.
However, for the $b$-decoherence process: as $S_{b,b} = S_{0,0} = \frac{1}{2}$, we have $\mathrm{C}_\mathcal{C}(\mathcal{F}) = \mathcal{F} = \{0,b\}$.
Only the Wilson loop operators $W_\Ts(b)$ survives the decoherence process, these operators mutually commute, and so the phase is classical topologically ordered.
Notably, our classification of this phase differs from that of Ref.~\onlinecite{2023BaoBoundarySPT}.

\vspace{1ex}\noindent\textbf{Relation between different replica orders}.
The Wilson loop tests give us a direct link between the $X$-loop expectation tensor~\eqref{eq:xp_XXXX_Qq} and topological order of the mixed phase.
In determining the properties of Wilson loop operators, we only needed to use a portion of the information contained within the expectation tensor/boundary SPT order.
For instance, whether $(X^\Ts)_\E$ is a Wilson loop operator depends only on whether $x_{2n}^\phd x_1^{-1} \in G$ and has trivial cocycles (see Tab.~\ref{tab:TCdesc_classification}).
This suggests that the auxiliary symmetries on some subset of layers can be ignored, which can be thought of as grouping a bunch of replica together as if it was one mixed state.

Furthermore, if $\rho$ and $\sigma$ are $mn$-indistinguishable, then $\rho^m$ and $\sigma^m$ are then $n$-indistinguishable (the converse statement is not obviously true).
Is there a connection between the topological order of $(\rho_\E)^m$ and the topological order of $\rho_\E$?
Here we offer some tantalizing, albeit incomplete, arguments towards such connection.

First, let $\mathcal{X}_n^{(C)}$ denote the $n$-replica $X$-loop tensor for Case~$C$ (all listed in Tab.~\ref{tab:TCcondensation}).
Observe that for each of the $p+3$ possible boundary SPT order,
\begin{align} \label{eq:Xtensor_contract} \begin{aligned}
& \sum_{\Tt_2,\Tt_3} \delta(\Tt_2+\Tt_3) \, \mathcal{X}_{n}^{(C)}(\Tt_1, \Tt_2, \Tt_3, \dots, \Tt_{2n})
\\& \mkern120mu = \kappa_C \mathcal{X}_{n-1}^{(C)}(\Tt_1, \Tt_4, \Tt_5, \dots, \Tt_{2n}) ,
\end{aligned} \end{align}
for
\begin{align}
	\kappa_C = \begin{cases} 1 & \text{Cases 1, 2}, \\ p & \text{Cases 3, 4}. \end{cases}
\end{align}
Basically, the $n$-replica $X$-loop tensor, after tracing over $(n-k)$ pairs of $(\Tt_{2i},\Tt_{2i+1})$, reduces to the $k$-replica $X$-loop tensor of the same case.

Inspired by the identity above, consider the following quantum channel
\begin{align}
	\E_m(X^\Ts \rho_{00} X^\Tt) = \frac{1}{\mathcal{J}_m} \sum_{\Tk_1,\dots,\Tk_{m-1}}
	\left[\begin{array}{@{} l @{}} \E(X^\Ts \rho_{00} X^{-\Tk_1}) \\{} \times \E(X^{\Tk_1} \rho_{00} X^{-\Tk_2}) \\{} \times \cdots \\{} \times \E(X^{\Tk_{m-1}} \rho_{00} X^\Tt) \end{array}\right] .
\end{align}
The channel can also be defined recursively: $\E_{\ell+m}(X^\Ts \rho_{00} X^\Tt) \propto \sum_\Tk \E_\ell(X^\Ts \rho_{00} X^{-\Tk}) \, \E_m(X^\Tk \rho_{00} X^\Tt)$, with $\E_1 = \E$.

While it is not clear that such map is ``local''---that is, can be approximated by a finite-depth quantum channel circuit---%
the map is trace preserving (with the appropriate choice of $\mathcal{J}_m$), and positive.
Leveraging the PEPS construction, the resulting states are locally indistinguishable from $\E(\rho_{00})^m = (\rho_\E)^m$.

By construction, the $n$-replica tensor network for $\E_m(\rho)$ consists of $mn$ copies of $\rho_\E$.
The $n$-replica $X$-loop expectation tensor for $\E_m(\rho)$ is given by $\mathcal{X}_{mn}^{(C)}(\Tt_1, \dots, \Tt_{2mn})$, after tracing over all pairs $(\Tt_{2i},\Tt_{2i+1})$ except for $(\Tt_{2m},\Tt_{2m+1}), (\Tt_{4m},\Tt_{4m+1}), \dots, (\Tt_{2mn},\Tt_{1})$.
Applying~\eqref{eq:Xtensor_contract}, this simplifies to $\mathcal{X}_{n}^{(C)}$, in the same case~$C$ as the $mn$-replica tensor for $\rho_\E$.

In other words, our hypothesis is that the $n$-replica topological order of $(\rho_\E)^m$ is the same as the $mn$-replica topological order of $\rho_\E$.

\section{Numerical results}
\label{sec:numerics}
\begin{figure*}
    \centering
    \includegraphics[width = \linewidth]{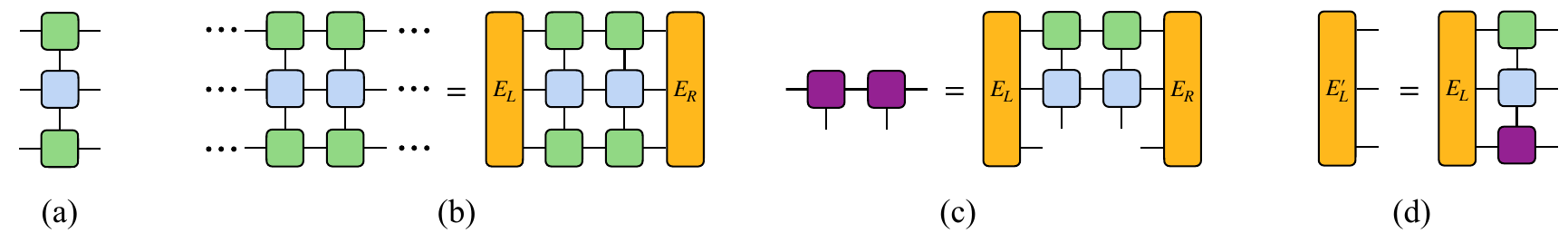}
    \caption{(a) Unit cell operator of the overlap tensor network. (b) The infinite overlap is equivalent to the finite overlap with the environments. The left and right environments are left and right dominant eigenvectors of the unit cell operator. (c) The new two-site tensor is given by the corresponding overlap two-site tensor with appropriate SVD decomposition. (d) The left and right environments are updated accordingly.}
    \label{fig:overlap}
\end{figure*}

In this section, we numerically calculate the boundary expectation tensor of the $\Z_2$ toric code under different quantum channels.
We aim to provide numerical support to our theoretical arguments regarding mixed-state topological order.
We begin by reproducing the well-known error thresholds for the $X$- and $Z$-flip error channels, and then we explore a more interesting situation with the amplitude damping channel.
This exploration demonstrates that PEPS is a powerful tool for both theoretical analysis and numerical calculation.

In our calculation, we employ the standard iPEPS and iMPS tensor network algorithms to handle the infinitely large toric code pure states and their boundary states~\cite{JordanPhysRevLett2008PEPS,PhysRevLettInfitniteMPSVidal,PhysRevBInfitnitePEPSOrusVidal}.
The key step in our calculation is obtaining the dominant eigenstate(s) of the transfer operator.
For this purpose, we use the straightforward power method.
Simply speaking, to find the dominant eigenvector of an operator $\mathcal{O}$, we start with a random vector and repeatedly multiply it by $\mathcal{O}$, normalizing at each step.
This process leads to a convergence to the dominant eigenvector if the operator has the largest eigenvalue (in magnitude).
Specifically, in our case, we start with a completely symmetry breaking state and repeatedly multiply it by the transfer operator.
We can determine if a symmetry operator is restored by taking its expectation values with respect to the symmetry generators.

\begin{figure}
    \centering
    \includegraphics[width = \linewidth]{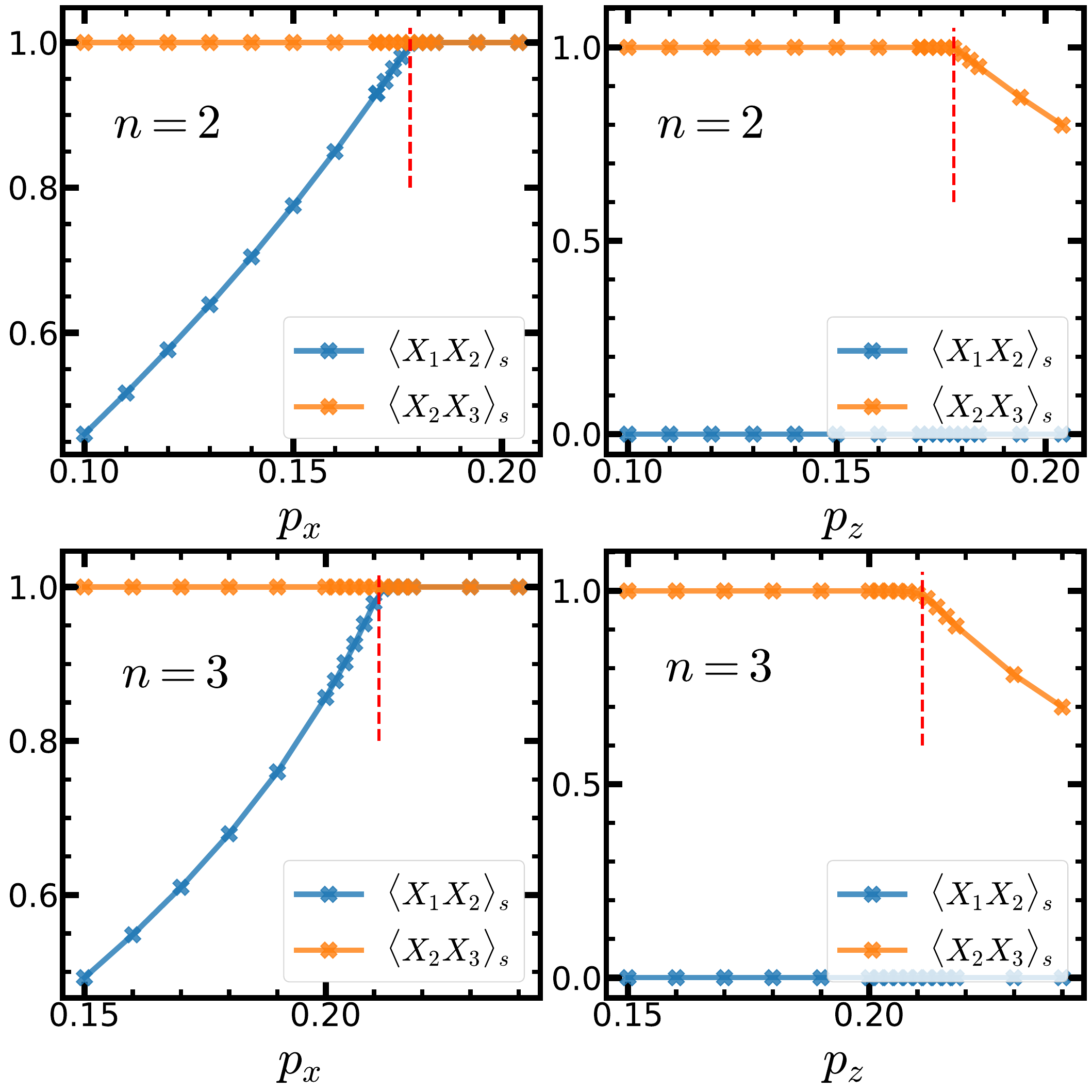}
    \caption{The expectation values per site of the boundary symmetries for the $\Z_2$ toric code model under the $X$- and $Z$-flip quantum channels. The red dashed lines indicate the analytical solutions ($p^{(2)}_c \approx 17.8\%, \ p^{(3)}_c \approx 21.1\%$).  }
    \label{fig:XandZflip}
\end{figure}

The situation becomes slightly more complicated when dealing with an infinite size system.
In our case, the transfer operator is represented by an iMPO, and the eigenstate is represented by a iMPS with the unit cell of two (the bond dimensions of the iMPS are fixed to be 10 in this section). 
Here we use the two-site update method to apply an iMPO to an iMPS (see Fig.~\ref{fig:overlap}).
The idea is that the state $\ket{b}$, which maximizes the quotient $\braket{b|\mathcal{O}|a}/\braket{b|b}$, is given by $\ket{b} = \mathcal{O}\ket{a}$.
The inner product $\braket{b|\mathcal{O}|a}$ is represented by a overlap tensor (Fig.~\ref{fig:overlap}(a)).
To evaluate this overlap, we substitute the infinite overlap with a finite one that sandwiched by environments.
Here the environments are dominant eigenstates of the unit cell of the overlap (Fig.~\ref{fig:overlap}(b)).
We then update the iMPS of $\ket{b}$ by replacing two sites of it with the corresponding sites in the overlap tensor network (Fig.~\ref{fig:overlap}(c)).
This is followed by an appropriate SVD decomposition and normalization to maintain the canonical form of the iMPS.
The left and right environments are also updated accordingly (Fig.~\ref{fig:overlap}(d)).
This two-site update process is repeated until the Frobenius distance between the old and new two-site tensors falls below a set threshold ($10^{-7}$ in our case), or until the total number of updates exceeds the predefined maximum.
After each multiplication step, we record the quotient $\braket{b|\mathcal{O}|b}/\braket{b|b}$, and repeat using $\ket{b}$ as $\ket{a}$ for the next multiplication step.
We terminate the power method when the difference between quotients drops below the threshold ($10^{-7}$ in our case), or when the number of multiplications exceeds the maximum allowance.
Using this iMPO-iMPS multiplication procedure, the iMPS converges to the dominant eigenvector of the transfer operator, i.e., the boundary state of the $n$-replica network.

After obtaining the boundary state $\ket{b}$, the expectation value $\braket{\auxX_i^\Th\auxX_j^\Th}$ is given by the overlap tensor network $\braket{b|\auxX_i^\Th\auxX_j^\Th|b}$.
Recall that the symmetry generators $\auxX^\Th$ is a product of $\auxX$ applied to each ``site'' of the boundary MPS.
Here we calculate the expectation value per site $\braket{\auxX_i \auxX_j}_s$ of the iMPS, which can be an arbitrary number between 0 and 1.
The expectation value of the symmetry generator
\begin{align}
     \bigXp{\auxX_i^\Th \auxX_j^\Th} = \bigl( \bigXp{\auxX_i \auxX_j}\!_s \bigr)^L
\end{align} is either 0 or 1 in the thermodynamic limit $L \to \infty$.
Thus, we can determine the boundary state symmetry group $G$ by calcuating the two boundary expectations $\braket{\auxX_1\auxX_2}_s$ and $\braket{\auxX_2\auxX_3}_s$ (for some arbitrary direction):
	specifically $x_1x_2 \in G$ if $\braket{\auxX_1\auxX_2}_s = 1$, and $x_2x_3 \in G$ if $\braket{\auxX_3\auxX_3}_s = 1$.
The topological order of the mixed state can then be read from Tab.~\ref{tab:TCdesc_classification}.
The subclasses of Case~3 can be distinguished by calculating the SPT order of the boundary MPS using the method presented in Ref.~\onlinecite{Pollmann2012DetectingSPT}.

We first reproduce the results of the $X$- and $Z$-flip error channels by calculating the boundary expectations.
In these cases, the on-site quantum channel is the composition of
\begin{subequations} \begin{align}
	\epsilon_x(\sigma) &= (1-p_x)\sigma + p_x X \sigma X ,
\\	\epsilon_z(\sigma) &= (1-p_z)\sigma + p_z Z \sigma Z . \label{eq:Zerr_channel}
\end{align} \end{subequations}
(Recall that $X$ creates a pair of $m$ anyons, and $Z$ creates a pair of $e$ anyons.)
Note that the two error channels commute with $\M$, so it's unnecessary to include the measurement channels in the tensor network as $\M\E(\rho_\text{TC}) = \E\M(\rho_\text{TC}) = \E(\rho_\text{TC})$.
For the 2- and 3-replica cases, the TC models subject to these two error channels are exactly solvable via a mapping to the random bound Ising model.
The thresholds are given by $p^{(2)}_c = \left(1-\sqrt{\sqrt{2}-1}\right)/2 \approx 17.8\%$, $p^{(3)}_c =\left(3-\sqrt{3}\right)/6\approx 21.1\%$~\cite{dennis2002topologicalMemory,Jean_Marie_Maillard_2003_RBIM,DiagnosticsRelativeEntropyNegativity}.
We present the boundary expectations in Fig.~\ref{fig:XandZflip}, finding that the boundary phase transitions align with previous results.
Observe that for the $X$-flip channel, the boundary symmetry ${x_1x_2}$ is restored across the transition point.
Furthermore, it is straightforward to verify that the projective representation for this case is trivial.
Thus, this is a phase transition from Case~4 to Case~3a.
For the $Z$-flip channel, the boundary symmetry ${x_2x_3}$ breaks across the transition point, indicating a phase transition from Case~4 to Case~2.
These boundary symmetry enhancement/breaking patterns are consistent with our analysis in Tab.~\ref{tab:TCcondensation}.

\begin{figure}
    \centering
    \includegraphics[width = \linewidth]{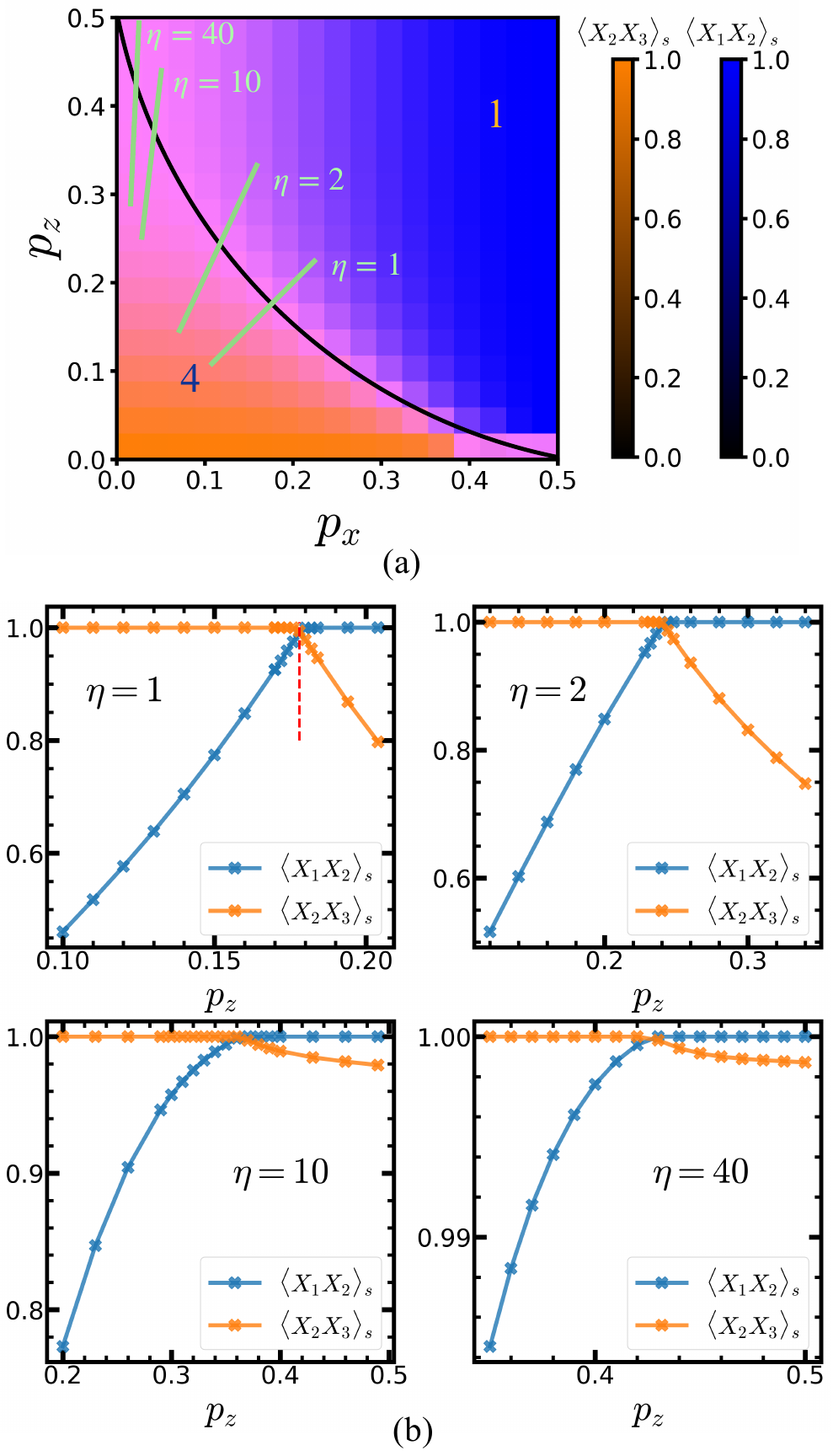}
    \caption{%
(a) Phase diagram based on the boundary expectation values of the 2-replica XZZX code.
(b) The boundary expectation values versus the $Z$-flip error for different error biases. Here the biases is defined as $\eta = p_z/p_x$.
In the unbiased case where $p_z = p_x$, the critical point is again $p_c^{(2)} \approx 17.8\%$.
The error threshold, $p_z^{(2)}$,  rises correspondingly as the bias $\eta = p_z/p_x$ increases.
}
    \label{fig:data_hadamard}
\end{figure}

Next, we consider the XZZX code~\cite{WenPhysRevLett2003Plaquette,bonilla2021xzzx}, which is a variation of the standard TC model.
The unit cell of the XZZX code is obtained by applying Hadamard gates on every horizontal sites.
The unit cell of the PEPS is
\begin{align} \label{eq:toric_code_PEPS_Hadamard_A}
\begin{tikzpicture}
	[baseline={([yshift=-.5ex]current bounding box.center)},rnode1/.style={circle,inner sep=1pt, draw=orange!100,fill = white!100,line width=1.5pt,minimum size = 10pt}],
	\draw[color=black!100,line width=1.6pt] (-20pt,0pt) -- (25pt,0pt);
	\draw[color=black!100,line width=1.6pt] (0pt,-25pt) -- (0pt,15pt);
	\draw[color=black!100,line width=2pt] (0pt,0pt) -- (18pt,-18pt);
	\draw[color=black!100,fill=white!100,rounded corners=4pt,line width=1.5pt] (-10pt,-10pt) rectangle (10pt,10pt);
	\node[] at (0pt, 0pt) {$A$};
\end{tikzpicture}
\mkern10mu = \mkern10mu
\begin{tikzpicture}
	[baseline={([yshift=-.5ex]current bounding box.center)},rnode1/.style={circle,inner sep=1pt, draw=orange!100,fill = white!100,line width=1.5pt,minimum size = 10pt,path picture={\draw[orange]
       (path picture bounding box.south) -- (path picture bounding box.north) 
       (path picture bounding box.west) -- (path picture bounding box.east);
      }}],
	\draw[color=black!100,line width=1.6pt] (-20pt,0pt) -- (45pt,0pt);
	\draw[color=black!100,line width=1.6pt] (0pt,-25pt) -- (0pt,15pt);
	\draw[color=orange!100,line width=1.6pt] (12.5pt,0pt) -- (38.5pt,-26pt);
	\draw[color=orange!100,line width=1.6pt] (0pt,-12.5pt) -- (10pt,-22.5pt);
	\draw[color=black!100,fill=white!100,rounded corners=2pt,line width=1.5pt] (22.5pt,-20.5pt) rectangle (33pt,-10pt);
	\node[]() at (27.75pt,-15.25pt){$H$};
	\node[rnode1] (0) at (12.5pt,0pt) {};
	\node[rnode1] (1) at (0pt,-12.5pt) {};
\end{tikzpicture} \ .
\end{align}
In the XZZX code, an $X$ error creates a pair of $m$ anyons on vertical sites and a pair of $e$ anyons on horizontal sites; conversely, a $Z$ error generates a pair of $e$ anyons on vertical sites and a pair of $m$ anyons on horizontal sites.
As seen in Fig.~\ref{fig:data_hadamard}, the Hadamard transformation leads to two significant consequences. 

Firstly, the equal amounts of $e$ and $m$ errors means that there is a direct phase transition from Case~4 to Case~1.
Thus, there are no CTO phases in this scenario.

Secondly, since the vertical and horizontal sites have different error types, this model has a larger error threshold for $X$- and $Z$- type errors.
In the unbiased case where $p_x = p_z$, the error threshold is identical to that of the $\mathbb{Z}_2$ TC model because there is an equal number of $m$- and $e$-type errors.
As the bias $\eta = p_z/p_x$ increases to 2, 10, and 40, the error threshold, $p_z^{(2)}$, correspondingly rises to approximately 0.24, 0.36, and 0.43, respectively.
We suspect the error threshold increases to 0.5 as the bias approaches infinity.%
	\footnote{For the 1-replica (conventional) error correction protocol, the error threshold is 0.5 as the bias approaches infinity~\cite{bonilla2021xzzx,xiao2024exactXZZX}.}

\begin{figure}
    \centering
    \includegraphics[width = \linewidth]{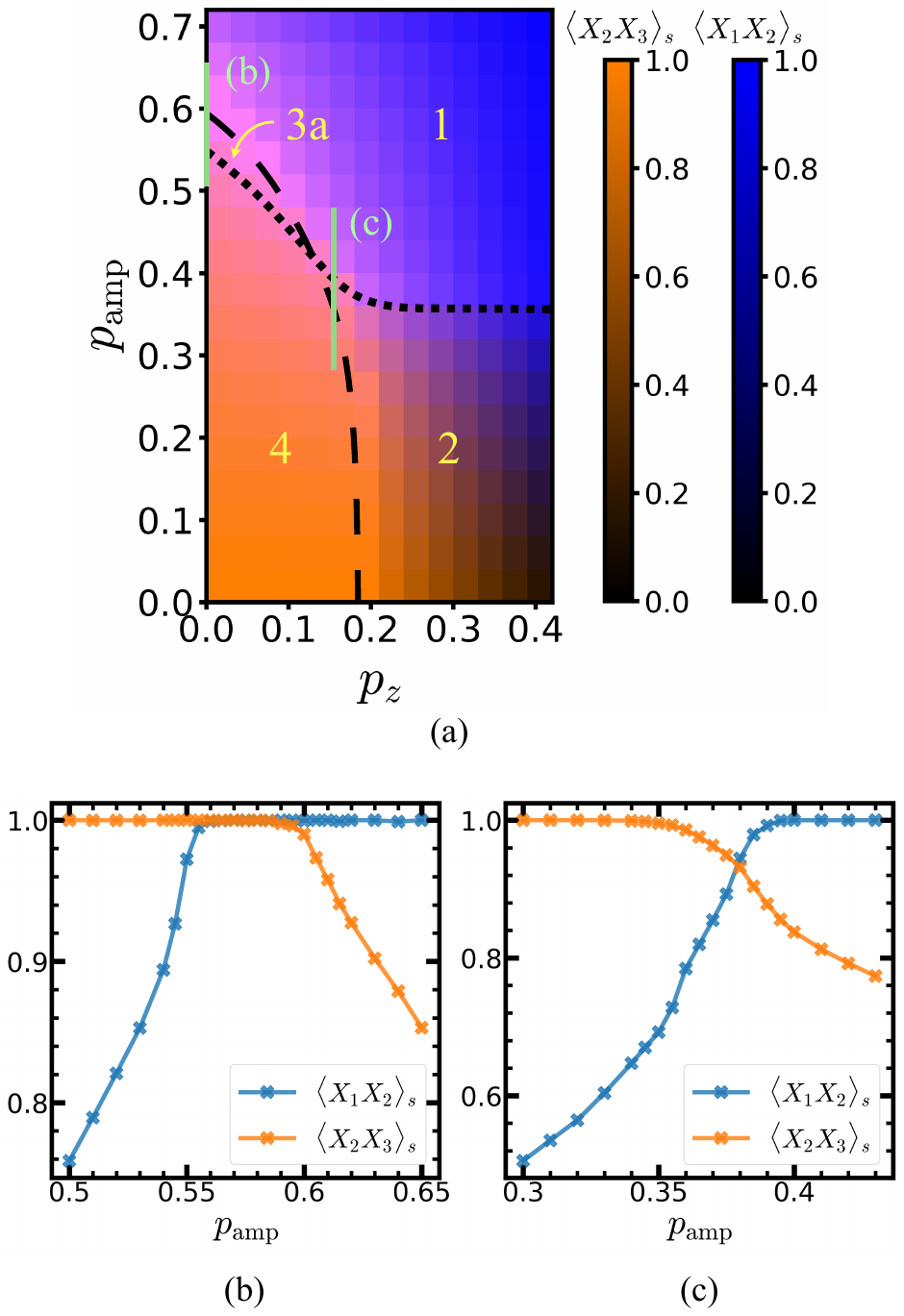}
    \caption{(%
a) Phase diagram of the 2-replica $\Z_2$ TC model subjects to both the $Z$-flip channel and the amplitude damping channel. The phase boundaries are set by the boundary expectations values. The $e$- and the $m$-decoherence transitions are denoted by the dashed and dotted lines respectively. Both lines on the plot are approximate.
(b) and (c) Two different line cuts through the phase diagram which shows that the two transition lines crosses.}
    \label{fig:zflip_deamp}
\end{figure}

We choose the amplitude damping and dephasing channel as our last example.
The amplitude damping channel is non-Pauli, and is given by
\begin{align}
    \epsilon_\mathrm{amp}(\sigma) &=   M_1\sigma M_1^\dag +  M_2\sigma M_2^\dag \,,
\end{align}
where
\begin{align}
    &M_1 = \begin{pmatrix} 1& 0 \\
    0 &\sqrt{1-p_\mathrm{amp}}
    \end{pmatrix},& M_2 = \begin{pmatrix} 0& \sqrt{p_{\mathrm{amp}}} \\
    0 &0
    \end{pmatrix}.&
\end{align}
The dephasing channel is identical to the $Z$-error channel Eq.~\eqref{eq:Zerr_channel}.
Unlike in the previous two scenarios, the amplitude damping channel does not commute with $\M$, specifically the portion that measures $A_+$ (comprised of four $X$'s) [cf.\ Eq.~\eqref{eq:M_channel}].
Here the vertex syndrome measurements must be explicited inserted in the $n$-replica tensor network as a four-site quantum channel,
resulting in a larger bond dimension in comparison to the previous error models.

Fig.~\ref{fig:zflip_deamp} illustrates the phase diagram of the 2-replica $\Z_2$ TC subjects to both the amplitude damping channel and the $Z$-flip channel.
Notably, even with the presence of the amplitude damping channel alone, the system undergoes two types of phase transitions: initially from Case~4 (QTO) to Case~3a (CTO), followed by a transition from Case~3a to Case~1 (trivial).

\section{Summary and Discussion}
\label{sec:discussion}

In this paper, we explore various aspects of topological order in mixed states.
We began, in Sec.~\ref{sec:TopoMixed}, with two definitions of replica topological order in mixed states.
Our definitions are based on the concept of local and global indistinguishability, involving $n$ replicas of the state.
While we use the $n$-Schatten norm to define global indistinguishability [cf.\ Eq.~\eqref{eq:Schatten_norm}],
it is likely that we will arrive at an equivalent definition by utilizing instead the R\'enyi relative entropy $D^{(n)}(\rho\Vert\sigma) \defeq \frac{1}{1-n} \Tr( \rho\sigma^{n-1} ) / \Tr(\rho^n)$ used in Ref.~\onlinecite{DiagnosticsRelativeEntropyNegativity}.

We present examples, based on Kitaev's non-Abelian toric code, of classical topologically ordered phases in Sec.~\ref{sec:GaugeModel}.
We also expect thermal mixed states of the form $e^{-\beta H}$ to be trivial in our classification of mixed states~\cite{PhysRevLett_Hastings_NozeroTemperature}, however we leave the argument for such case to future work.

In this work, we classify phases by considering the subspace of locally indistinguishable states, and the operators that act on such subspace.
A natural question is to ask about their ``excitations''---states that are locally indistinguishable nearly everywhere except for a few isolated regions---to provide a complementary perspective to the effort of this work.
What is the generalization of an anyon and what is the correct mathematical formulation for their braiding, fusion, etc.\ properties?
It will also be fruitful to explore how our definitions compare with that of Ref.~\onlinecite{sang2023mixedstateRG}, which classifies phases based on their ability to be transformed into each other via local channels.

In Sec.~\ref{sec:boundarySPT}, we study mixed states in the context of quantum error correction.
Here, we classify descendant phases which results from a local quantum channel acting on the $\Z_p$ toric code, and associate each phase with an effective quantum channel in our postselection QEC protocol.
We find, when $n \geq 2$, that there are $p+3$ descendant phases: Cases~1, 2, 4 (with one each), and Case~3 (with $p$ phases).
The main technical tool connecting these ideas is the PEPS tensor network, which naturally satisfies the local indistinguishability criteria allowing topological phase to be encoded.

Specifically, we consider mixed states of the form $\M\E(\rho_\text{TC})$ where $\E$ is a local quantum channel, and $\M$ are syndrome measurements.
The PEPS encoding the $\Z_p$ toric code (as a wavefunction) has $\mathcal{G} = \Z_p$ auxiliary symmetry, its density matrix has $\mathcal{G}^2$ symmetry, and its $n$-replica network has $\mathcal{G}^{2n}$ symmetry.
However in evaluating the $n$-replica network, the auxiliary symmetry is spontaneously broken by the dominant eigenvectors of the transfer operator---which we refer to as the ``boundary''---and the state can be classified by the residue boundary symmetry $G \subseteq \mathcal{G}^{2n}$ and the SPT order of the boundary state $H^2(G,\mathrm{U}(1))$.
We show that there is a one-to-one relation between the boundary SPT order of the $n$-replica network (which results from contracting $n$ copies of $\M\E(\rho_\text{TC})$) with the quantum channel $Q$ after optimal recovery, shown in Table~\ref{tab:TCcondensation}.

In Sec.~\ref{sec:ClassifyMixed}, we relate two definitions of the mixed-state topological order to the characteristics we obtained in the previous section. 
To be more precise, we relate the code space quantum channel to the geometric structure of the mixed state density matrix space $\DMxS{n}$, and the boundary SPT order to the properties of the Wilson loop operators acting on $\DMxV{n}$.
In doing so, we demonstrate the equivalence between the two definitions of mixed state topological order, at least when applied to descendants of the $\Z_p$ toric code for $p$ prime.
We caution that some of the tools---such as the Wilson loop tests---developed in this section relies on the map between the logical spaces being surjective.
The decoherence process may be non-surjective when dealing with non-Abelian parent state, for example
	the fermion-decoherence in the Ising TQFT, or (spin-2) boson-decoherence in the $\mathrm{SU}(2)_4$ TQFT.
From the anyon condensation perspective, this is because certain topological sector ``splits'', and so not every state in the mixed state phase can be construct starting from the ground state of the parent phase.

Our QEC protocol has an explicit step $\M$ measuring all the syndrome operators on all $n$ replica.
This is such that the postselction process is well-defined, and also allows us to explicitly write down the code space following $\M$ [cf.~\eqref{eq:def_Ptheta}], an important step in relating the $X$-loop expectation tensor to the effective quantum channel~\eqref{eq:xp_XXXX_Qq}.
From the QEC perspective, $\M$ dephases in the basis of the TC stabilzers and is a lossy operation;
	the stabilzers may not be the optimal set for decoding and recovery (e.g.\ when there are coherent error sources).
From the perspective of classifying mixed states via Wilson loop definition, we only rely the boundary SPT order, which is well-defined without the $\M$ step.

Unsurprisingly, the error threshold depends on the replica index $n$.
This is attributed to the postselection process, which alters the relative probabilities of different error syndromes $\Prob(\theta) \to \Prob(\theta)^n$.
Intuitively, this enhances the likelihood of syndromes with fewer errors, which have a higher probability to be correctable.
The ``amount of postselection'' scales exponentially with system size.
For example, for a $\Z_2$ toric code with $N$ physical qubits, one among $\approx 2^{(n-1)N}$ syndrome outcomes must be postselected in this protocol; the ratio becomes unity in the limit $n \to 1$.
Extrapolating our results, for mixed states whose 1-replica topological order matches that of the original parent state, quantum information can always be recovered without any postselection.
Since postselection is not needed in such limit, a stronger conjecture is that any local error channel $\E$ preserving the 1-replica QTO admits a perfect recovery map preserving its logical quantum content, i.e., a quantum error correcting code.

Throughout this work, we let $n \geq 2$ in all the calculations and analysis.
The reason is that the 1-replica network for $\E(\rho_\text{TC})$ is identical to that of $\rho_\text{TC}$, as $\E$ is trace-preserving.
This means that $G$ is always generated by $x_1^\phd x_2^{-1}$ and the methods from Sec.~\ref{tab:TCdesc_classification} are unable to constrain the form of $Q$ or differentiate between the possible mixed state topological orders.
However, if our conjecture at the end of Sec.~\ref{sec:TCdescRemarks} holds true,
then the $n$-replica topological order of $\rho_\E$ is the same as the $2$-replica topological order of $(\rho_\E)^{n/2}$.
The vectorizing of the latter state, $\bigl| (\rho_\E)^{n/2} \bigr\rangle\mkern-5mu\bigr\rangle$, is simply the canonical purification of $(\rho_\E)^n$.
Taking the replica limit $n \to 1$, it suggests that 1-replica mixed state topological order of $\rho_\E$ is determined by the TQFT describing its canonical purification $\bigl| \sqrt{\rho_\E} \bigr\rangle\mkern-5mu\bigr\rangle$.

Mixed states also provide a route to constructing PEPS network for states with chiral topological order.
It remains an open question whether it is possible to construct an infinite PEPS representing a chiral topological state with exponential decaying correlation functions~\cite{Wahl2013PEPS_ChiralTO,Yang2015PEPS_ChiralTO,Dubail2015TN_nogo_ChiralTO,Chen2018PEPS_NonAbelian_ChiralTO,Hackenbroich2018PEPS_SU2}.
However, it appears that tensor network representations for mixed states with chiral topological order are indeed possible, as discussed in Sec.~\ref{sec:TCdescRemarks}.
Can we smoothly connect the mixed state to a pure state?  What happens to the tensor network representation?

In this work, we delve into the many interconnected relations between quantum information, topological quantum, field theory, tensor networks, and quantum error correction in studying the topological order for mixed states.
Some of the next steps on the horizon will involve finding a complete field theoretical descriptions for topological mixed states, and developing new numerical to diagonse their order.

\begin{acknowledgments}
The authors are grateful for enlightening discussions with
	Jason Alicea,
	Ehud Altman,
	Maissam Barkeshli,
	Jeongwan Haah,
	Eric Huang,
	Pablo Sala,
	Norbert Schuch,
	and Ruben Verresen.
ZL and RM acknowledge support from the National Science Foundation under award No.~DMR-1848336.
ZL also acknowledges support from the Pittsburgh Quantum Institute.
\end{acknowledgments}

\bibliography{ref}

\clearpage

\onecolumngrid
\appendix
\section{Some matrix inequalities}
\label{app:matrix_ineq}

\begin{lemma} \label{lem:column_inclusion}
Let $A$, $B$, $C$ be matrices, with $A$, $C$ square, such that the block matrix
\begin{align}
	\begin{bmatrix} A & B \\ B^\dag & C \end{bmatrix} \succeq 0 .
\end{align}
Then
\begin{align}
	\operatorname{im} B \subseteq \operatorname{im} A .
\end{align}
\end{lemma}

$\operatorname{im} B$ denotes the image of $B$---the vector space span by the columns of $B$---also known as the \emph{column space} or the \emph{range} in literature.
This lemma is proved in Ref.~\onlinecite{horn2012matrix} \S7.1.10.

\begin{lemma} \label{lem:matrix_power_trace_ineq}
Let $A,B \succeq 0$ be matrices of the same size.
Then for $0 < \mu < 1$
\begin{align} \label{eq:matrix_power_trace_ineq}
    \tr A^{1-\mu} B^\mu \leq \bigl(\tr A\bigr)^{1-\mu} \bigl(\tr B\bigr)^\mu \,.
\end{align}
Equality holds iff $A \propto B$.
\end{lemma}

\begin{proof}
Denote $P_A = \lim_{\nu \to 0^+} A^\nu$ and $P_B = \lim_{\mu \to 0^+} B^\mu$ as the projectors on to the image of $A$ and $B$ respectively.

If $A$ and $B$ are orthogonal---i.e., $P_AP_B=0$---then $\tr A^{1-\mu} B^\mu = 0$ for all $0 < \mu < 1$, and the inequality is trivially true.
Equality occurs iff either $A = 0$ or $B = 0$.

Suppose $A$ and $B$ are not orthogonal, then $\tr A^{1-\mu} B^\mu > 0$ for all $0 < \mu < 1$.
Define
\begin{align}
   f(\mu) \defeq \begin{cases} \ln\tr AP_B & \mu=0, \\ \ln\tr A^{1-\mu} B^\mu & 0<\mu<1 , \\ \ln\tr P_AB & \mu=1. \end{cases}
\end{align}
$f$ is continuous in $[0,1]$ and infinitely differentiable in $(0,1)$.
Since $P_A, P_B \preceq 1$, $f(0) \leq \ln\tr A$ and $f(1) \leq \ln\tr B$.
If $f$ is convex, by Jensen's inequality
\begin{align} \label{eq:matrix_power_trace_ineq_J}
   f(\mu) \leq (1-\mu) f(0) + \mu f(1) \leq (1-\mu)(\ln\tr A) + \mu(\ln\tr B)
\end{align}
which gives our desired inequality.

We will prove that $f$ is convex by showing its second derivative is non-negative.
Note that $\frac{\D{}}{\D{\mu}} A^{1-\mu} = (-\ln A)A^{1-\mu} = A^{1-\mu}(-\ln A)$ and $\frac{\D{}}{\D{\mu}} B^\mu = (\ln B)B^\mu = B^\mu(\ln B)$.%
    \footnote{While $\ln A$ diverges if $A$ becomes singular (i.e., some of the eigenvalues are zero), combinations like $(\ln A)A^{1-\mu}$, $A^{1-\mu}(\ln A)$, $(\ln A)A^{1-\mu}(\ln A)$, etc.\ are still finite.}
\begin{align} \begin{aligned}
    \frac{\D{f}}{\D{\mu}}
    &= \frac{\tr A^{1-\mu} (\ln B-\ln A) B^\mu}
    {\tr A^{1-\mu} B^\mu} \,,
\\
    \frac{\D[2]{f}}{\D{\mu}^2}
    &= \frac{ \bigl(\tr A^{1-\mu} B^\mu \bigr) \bigl[\tr A^{1-\mu} (\ln B-\ln A) B^\mu (\ln B-\ln A) \bigr]
    - \bigl[ \tr A^{1-\mu} (\ln B-\ln A) B^\mu \big]^2 }
    {\bigl(\tr A^{1-\mu} B^\mu \bigr)^2} \,.
\end{aligned} \end{align}
Via the Cauchy-Schwarz inequality $(\tr E^\dag F)^2 \leq (\tr E^\dag E) (\tr F^\dag F)$ with $E = A^{(1-\mu)/2} B^{\mu/2}$ and $F = A^{(1-\mu)/2} (\ln A - \ln B) B^{\mu/2}$,
we have
$\bigl[\tr B^{\mu/2} A^{1-\mu} (\ln A - \ln B) B^{\mu/2} \bigr]^2
\leq \bigl(\tr B^{\mu/2} A^{1-\mu} B^{\mu/2} \bigr)
\bigl[\tr B^{\mu/2} (\ln A - \ln B) A^{1-\mu} (\ln A - \ln B) B^{\mu/2} \bigr]$.
Indeed the numerator of $\frac{\D[2]{f}}{\D{\mu}^2}$ is the difference between the two sides of this inequality and hence non-negative.

For a convex function $f$, equality $f(\mu) = (1-\mu) f(0) + \mu f(1)$ at any $\mu \in (0,1)$ requires the function to be linear, implying equality for all $\mu \in [0,1]$.
Likewise, equality $(1-\mu) f(0) + \mu f(1) = (1-\mu)(\ln\tr A) + \mu(\ln\tr B)$ at any $\mu \in (0,1)$ requires the end points $f(0)$ and $f(1)$ to be equal to $\ln\tr A$ and $\ln\tr B$ respectively,
and hence equality of~\eqref{eq:matrix_power_trace_ineq_J} at any $\mu \in (0,1)$ implies equality for all $\mu \in [0,1]$.
Choosing the halfway point, equality of~\eqref{eq:matrix_power_trace_ineq} implies $\tr A^{1/2} B^{1/2} = \sqrt{\tr A} \sqrt{\tr B} \Rightarrow A \propto B$ again via the Cauchy-Schwarz inequality.
It is straightforward to show that $A \propto B \Rightarrow \tr A^{1-\mu} B^\mu = \bigl(\tr A\bigr)^{1-\mu} \bigl(\tr B\bigr)^\mu$.
\end{proof}

\begin{corollary} \label{lem:nSchatten_trace_ineq}
Let $\lVert \sigma \rVert_n$ denote the $n$-Schatten norm, defined via $\lVert \sigma \rVert_n \defeq \bigl[ \Tr(|\sigma|^n) \bigr]^{1/n}$.
Then for $0 < k < n$,
\begin{align}
	\bigl\lVert A \bigr\rVert_n^{k} \bigl\lVert B \bigr\rVert_n^{n-k} \geq \tr\bigl( |A|^k |B|^{n-k} \bigr) .
\end{align}
\end{corollary}

\vspace{1ex}

\begin{lemma} \label{lem:block2tr}
Let $A, B$ be square matrices such that
\begin{align}
    \begin{bmatrix} A & B \\ B^\dag & A \end{bmatrix} \succeq 0 .
\end{align}
Then for $\alpha \geq 1$,
\begin{align}
    \bigl\lvert \tr (BA^{\alpha-1}) \bigr\rvert \leq \tr A^\alpha \,,
\end{align}
with equality iff $B = e^{i\theta} A$ for some $\theta \in \mathbb{R}$.
\end{lemma}

\begin{proof}
We have $A \succeq 0$.
Let
\begin{align}
    F_\alpha &= \begin{bmatrix} A^{\alpha-1} & -A^{\alpha-1} \\ -A^{\alpha-1} & A^{\alpha-1} \end{bmatrix} ,
&   G_\theta &= \begin{bmatrix} A & e^{-i\theta}B \\ e^{i\theta}B^\dag & A \end{bmatrix} ,
\end{align}
which are both positive semidefinite for all $\alpha \geq 1$ and $\theta \in \mathbb{R}$.
Since
\begin{align}
    \frac{1}{2} \tr F_\alpha G_\theta
    = \tr A^\alpha - \operatorname{Re} \tr\bigl( e^{-i\theta} B A^{\alpha-1} \bigr) \,,
\end{align}
is always non-negative, we have $\tr A^\alpha \geq \bigl\lvert \tr\bigl( B A^{\alpha-1} \bigr) \bigr\rvert$.
Equality occurs iff $\operatorname{im} G_\theta \subseteq \ker F_\alpha$.

If $A \succ 0$ (full rank),
$\ker F_\alpha = \mathrm{span} \bigl\{ \bigl(\begin{smallmatrix} v \\ v \end{smallmatrix}\bigr) \,\big|\, v \in \mathbb{C}^{\dim A} \bigr\}$.
$\operatorname{im} G_\theta$ is spanned by the columns of $G_\theta$, which belong in $\ker F_\alpha$ iff $e^{-i\theta}B = A$, i.e., when all the blocks are identical.

If $A$ is singular (not full rank), we can make the same argument by finding an isometry $K$ (with $K^\dag K = 1$) and a positive matrix $A'$ such that $A = K A' K^\dag$.
Because $G_\theta \succeq 0$, we can also find $B'$ such that $B = K B' K^\dag$.
The lemma holds for the pair of matrices $(A',B')$ and thus also holds for $(A,B)$.
\end{proof}

\begin{lemma} \label{lem:phaseMx_positive}
Let $A$ be a positive semidefinite matrix where all its elements have unit modulus:
$A_{i,j} = e^{i\theta_{i,j}}$.
Then (a) $A_{i,k} = A_{i,j}A_{j,k}$ for all $i,j,k$, and (b) $A$ has rank 1.
\end{lemma}

\begin{proof}
The diagonal elements of a positive semidefinite matrix must be real and non-negative, and hence $A_{ii}=1$.
Consider a $3\times3$ principal submatrix of $A$ with rows/columns $\{i,j,k\}$:
\begin{align}
    A^{[i,j,k]} &= \begin{bmatrix}
    1 & e^{i\theta_{i,j}} & e^{i\theta_{i,k}}
    \\ e^{-i\theta_{i,j}} & 1 & e^{i\theta_{j,k}}
    \\ e^{-i\theta_{i,k}} & e^{-i\theta_{j,k}} & 1
    \end{bmatrix} ,
&   \det A^{[i,j,k]} &=
    -4 \biggl( \sin\frac{\theta_{i,j}-\theta_{i,k}+\theta_{j,k}}{2} \biggr)^2 .
\end{align}
Being a principal submatrix, $A^{[i,j,k]}$ must also be positive semidefinite with non-negative determinant.
This is only possible if $\theta_{i,j}-\theta_{i,k}+\theta_{j,k} \in 2\pi\Z$, and hence $e^{i\theta_{i,k}} = e^{i\theta_{i,j}} e^{i\theta_{j,k}}$.
(Alternately, we can invoke Lem.~\ref{lem:column_inclusion} here.)
Thus
\begin{align}
    A &= v \, v^\dag \,,
&   v^\dag &= \begin{pmatrix} 1 & e^{i\theta_{1,2}} & e^{i\theta_{1,3}} & \cdots & e^{i\theta_{1,n}} \end{pmatrix} .
\end{align}
\end{proof}

\section{Enumerating descendants topological orders of the toric code (Case~3)}\label{sec:TC_case_mf}

The goal of this section is to prove Prop.~\ref{prop:case_mf_same_m}, classifying the possible boundary SPT orders for the $\Z_p$ toric code when $x_1x_2^{-1}, x_2x_3^{-1} \in G$ and $x_1^k \notin G$ for $1 \leq k < p$.
While some of the results from the main text stipulates that $p$ is a prime, the results of this specific section applies for any integer $p$ greater than 1.
Our strategy is to first show that boundary SPT orders can be completely characterized by 3-point correlators.
Then we argue that the positivity of the channel $Q$ induces a number of positivity constraints on the boundary expectation values.
The positivity constraints on the 3-point correlators allow us to narrow down the possible boundary SPT orders.
The result of the section is that there are $p$ solutions.

In the $n$-replica formalism, the PEPS has $\Z_p^{2n}$ symmetry (which will be partially broken at the boundary).
It will be convenient to labeled the $2n$ ``layers'' as elements in $\Z_{2n}$.
For example, we will identify $X_k = X_{k+2n}$.

\begin{definition}
Let $n$ be a positive integer.
An \underline{X3-tensor} of size $2n$ is a map $\mathbf{X}: H_1(T^2;\Z_p)^{2n} \to \zeta_p^\Z \cup \{0\}$, such that
\begin{align} \label{eq:xp_Im_form2}
    \mathbf{X}(\Tt_1, \Tt_2, \dots, \Tt_{2n}) = \delta(\Tt_1+\Tt_2+\dots+\Tt_{2n}) \prod_{1 \leq i<j \leq 2n} \bigl(I_{\Tt_i,\Tt_j}\bigr)^{m_{i,j}}
\end{align}
for some matrix $m \in \Z_p^{2n\times2n}$ with $-m_{i,j} = m_{j,i}$ and $m_{i,i} = 0$ (antisymmetric and zero diagonal).
\end{definition}

Recall that $\zeta_p \defeq \exp\bigl[ 2\pi i/p \bigr]$,
so $\zeta_p^\Z$ denotes $\bigl\{ 1, \zeta_p, \zeta_p^2, \dots, \zeta_p^{-1} \bigr\}$.
The homology group $H_1(T^2;\Z_p) \cong \Z_p^2$, generated by $\Th$ and $\Tv$.
The intersection form $I: H_1(T^2;\Z_p)^2 \to \zeta_p^\Z$ is given as
\begin{align}
	I_{a\Th+b\Tv \,,\, c\Th+d\Tv} &= \exp\bigl[ 2\pi i (ad-bc) / p \bigr] .
\end{align}
The result of Prop.~\ref{prop:xp_case_mf_Im_form} is that the boundary expectation values of the PEPS $\Braket{ X_1^{\Tt_1} X_2^{\Tt_2} \cdots X_{2n}^{\Tt_{2n}} }$ for case~3 is an X3-tensor.%
	\footnote{While the proof for Prop.~\ref{prop:xp_case_mf_Im_form} is found in Sec.~\ref{sec:enum_boundary_SPT} where $p$ is a prime number, the proof itself is applicable for any positive integer $p$.}

The matrix $m$ is not uniquely specified by an X3-tensor $\mathbf{X}$.  Consider the transformation
\begin{align} \label{eq:m_gauge}
    \widetilde{m}_{i,j} = m_{i,j} + w(i) - w(j)
\end{align}
for $w : [1,2n] \to \Z_p$.
Indeed, the ratios between the nonzero elements of the X3-tensors computed from $m$ and $\widetilde{m}$ are
\begin{align}
    \prod_{i<j} I_{\Tt_i,\Tt_j}^{w(i)} \prod_{i<j} I_{\Tt_i,\Tt_j}^{-w(j)}
    = \prod_{i<j} I_{\Tt_i,\Tt_j}^{w(i)} \overbrace{  \prod_{j<i} I_{\Tt_j,\Tt_i}^{-w(i)} }^\text{swap $i \leftrightarrow j$}
    = \prod_{i<j} I_{\Tt_i,\Tt_j}^{w(i)} \mkern-12mu
        \overbrace{ \prod_{j<i} I_{\Tt_i,\Tt_j}^{w(i)} }^\text{use $I_{\Ts,\Tt} = 1/I_{\Tt,\Ts}$}
    = \prod_{i \neq j} I_{\Tt_i,\Tt_j}^{w(i)}
    = \overbrace{ \prod_{i=1}^{2n} \prod_{j=1}^{2n} I_{\Tt_i,\Tt_j}^{w(i)} }^\text{use $I_{\Ts,\Ts} = 1$}
    = \prod_{i=1}^{2n} I_{\Tt_i,0}^{w(i)}
    = 1 .
\end{align}
In the penultimate equality, we use the property $I_{\Ts,\Tt} I_{\Ts,\mathbf{u}} = I_{\Ts,\Tt+\mathbf{u}}$ and that $\sum_j \Tt_j = 0$ (enforced by the implicit delta function which is not shown).
As $m$ and $\widetilde{m}$ gives the same X3-tensor, we call~\eqref{eq:m_gauge} a \underline{gauge transformation}.

It is convenient to consider gauge-invariant quantities.
Define the completely antisymmetric tensor $\ccc : [1,2n]^3 \to \Z_p$:
\begin{align} \label{eq:def_c3_from_m}
    \ccc_{i,j,k} \defeq m_{i,j} + m_{j,k} + m_{k,i} \,,
\end{align}
which remain invariant under transformation~\eqref{eq:m_gauge}.
They are in fact given by the 3-point correlators%
	\footnote{The notation here is use to denote the X3-tensor evaluated with $\Tt_i = \Th$, $\Tt_j = \Tv$, $\Tt_k = -\Th-\Tv$, and $0$ for all other $\Tt_\ast$'s}
\begin{align} \label{eq:c3_from_xp}
    \bigXp{ \auxX_i^\Th \, \auxX_j^\Tv \, \auxX_k^{-\Th-\Tv} }
    = I_{\Th,\Tv}^{m_{i,j}} I_{\Th,-\Th-\Tv}^{m_{i,k}} I_{\Tv,-\Th-\Tv}^{m_{j,k}}
    = \zeta_p^{\ccc_{i,j,k}}
\end{align}
and hence $\ccc_{i,j,k}$ are completely determined by the X3-tensor.
The price we pay is that the $\ccc$'s are not independent, they are constrained by
\begin{align} \label{eq:c3_constraint4}
    \ccc_{i,j,k} - \ccc_{j,k,l} + \ccc_{k,l,i} - \ccc_{l,i,j} = 0 .    
\end{align}

\begin{proposition} \label{prop:Xtensor_group}
The following groups are isomorphic
\begin{enumerate}[label={(\roman*)}, itemsep=0pt, parsep=4pt, topsep=-2pt]
\item The set of X3-tensors of size $2n$ (group composition via element-wise multiplication).
\item The set of antisymmetric, $2n\times2n$ $\Z_p$-valued matrices with vanishing diagonal, modulo gauge relations~\eqref{eq:m_gauge}.
\item The set of completely antisymmetric tensors $\ccc_{i,j,k} \in \Z_p$ for $1 \leq i,j,k \leq 2n$ subject to conditions $\ccc_{i,i,j} = 0$ and Eq.~\eqref{eq:c3_constraint4}.
\end{enumerate}
\end{proposition}

\noindent
This proposition means that the classification of X3-tensors amounts to the classification of 3-point correlators $\ccc$.

\begin{proof}
Let $\Delta$ be a $(2n-1)$-simplex.
$\Delta$ consists of $2n$ vertices, $\binom{2n}{2}$ edges, $\binom{2n}{3}$ faces, etc.
Consider the simplicial cochain complex of $\Delta$ with $\Z_p$ coefficients.
\begin{align}
    \cdots \xrightarrow{d_0} C^1 \xrightarrow{d_1} C^2 \xrightarrow{d_2} C^3 \xrightarrow{d_3} C^4 \xrightarrow{d_4} \cdots
\end{align}
Since $\Delta$ is contractible, the cohomology groups $H^{k \geq 1}(\Delta;\Z_p)$ are trivial and this sequence is exact.
Therefore, we have an isomorphism among
\begin{align}
\operatorname{coker} d_1 \cong \operatorname{img} d_2 = \operatorname{ker} d_3 \,.
\end{align}
$C^2(\Delta;\Z_p)$ is the set of antisymmetric binary matrices with vanishing diagonal, $m \mapsto m + d_1 w$ is the gauge transformation~\eqref{eq:m_gauge}, and hence $\operatorname{coker} d_1 = C^2/\operatorname{im} d_1$ is the group~(ii).
$\operatorname{ker} d_3$ is the set of completely antisymmetric 3-tensors subject to the constraint, hence is the group~(iii).
The two groups are isomorphic, the bijection is induced by the formula~\eqref{eq:def_c3_from_m}.

Consider the following diagram.
\begin{align}
\xymatrix @C=18mm @R=4mm @M=2mm {
    C^2(\Delta;\Z_p)
    \ar[rd]_(0.5){\alpha \;=\; \text{Eq.~\eqref{eq:xp_Im_form2}}}
    \ar[dd]_{Q}
    \ar[rr]_(0.5){d_2 \;=\; \text{Eq.~\eqref{eq:def_c3_from_m}}}
&&  C^3(\Delta;\Z_p)
\\& \text{group (i)}
    \ar[ur]_(0.5){\beta \;=\; \text{Eq.~\eqref{eq:c3_from_xp}}}
    \ar[dr]_(0.4){\beta'}
\\  \operatorname{coker} d_1 = \text{group (ii)}
    \ar[ur]_(0.6){\alpha'}
    \ar[rr]_(0.5){\cong}
&&  \text{group (iii)} = \operatorname{ker} d_3
    \ar[uu]_(0.5){\iota}
}
\end{align}
$Q$ is the quotient map $C^2 \to C^2/\operatorname{im} d_1$ and $\iota$ is the inclusion map.
$\alpha'$ is well-defined because the map $\alpha$ is invariant under a gauge transformation.
Prop.~\ref{prop:xp_case_mf_Im_form} asserts that the map $\alpha$ is surjective, consequently $\alpha'$ is also surjective.
Eq.~\eqref{eq:c3_from_xp} specifies a map $\beta: \text{(i)} \to C^3$,
and the top triangle is commutative by direct computation.
Since $\alpha$ is surjective, $\operatorname{im} \beta = \operatorname{im} \beta\circ\alpha = \operatorname{im} d_2 = \operatorname{ker} d_3$,
and so $\beta'$ is well-defined.

Finally, since $Q$ is surjective, $\iota$ is injective, the left, top, right triangles, and the outer rectangle are commutative, the bottom triangle is also commutative.
As the map $\text{(ii)} \to \text{(iii)}$ is an isomorphism and that $\alpha'$ is surjective, the three groups are all isomorphic.
\end{proof}

\begin{definition} \label{def:Xtensor_char}
Let $\mathbf{X}(\Tt_1,\dots,\Tt_{2r})$ be an X3-tensor of size $2r$.
\begin{itemize}
\item
$\mathbf{X}$ is \emph{positive} if the following matrix $M[\mathbf{X}]$ is Hermitian and positive semidefinite.
Let $M[\mathbf{X}]$ be $\mathbf{X}$ reshaped as a $p^{2r} \times p^{2r}$ matrix, with rows indexed by $(\Tt_1,\Tt_2,\dots,\Tt_r)$ and columns indexed by $(-\Tt_{2r}, -\Tt_{2r-1}, \dots, -\Tt_{r+1})$.
\item
$\mathbf{X}$ is \emph{monochromatic} if its 3-tensor invariants satisfy
\begin{align}
	\ccc_{i,j,k} = \begin{cases} c & \text{$i < j < k$ or $j < k < i$ or $k < i < j$}, \\ -c & \text{$k < j < i$ or $i < k < j$ or $j < i < k$}, \\ 0 & \text{otherwise}. \end{cases}
\end{align}
for some $c \in \Z_p$.
\end{itemize}
\end{definition}

In other words,
$\mathbf{X}$ is positive if (1) $\mathbf{X}(-\Tt_{2r}, -\Tt_{2r-1}, \dots, -\Tt_1) = \mathbf{X}(\Tt_1, \Tt_2, \dots, \Tt_{2n})^\ast$ for all $\{\Tt\}$
and (2) for all $v: H_1(T^2;\Z_p)^r \to \mathbb{C}$, $\sum_{\{\Ts\},\{\Tt\}} v(\Ts_1,\dots,\Ts_r)^\ast \, \mathbf{X}(\Ts_1,\Ts_2,\dots,\Ts_r \,,\, -\Tt_r,\dots,-\Tt_2,-\Tt_1) \, v(\Tt_1,\dots,\Tt_r) \geq 0$.
Condition (1) is equivalent to the constraints $\ccc_{i,j,k} = \ccc_{2r+1-k, 2r+1-j, 2r+1-i}$ for all $(i,j,k)$.

\begin{definition}
Let $r \leq n$ be a positive integer.
A \underline{$2r$-subsequence} is an ordered list of $2r$ distinct elements of $\Z_{2n}$:
$(i_1,i_2,\dots,i_r | i_{r+1},\dots,i_{2r})$
such that $i_1 + i_{2r} = i_2 + i_{2r-1} = \dots = i_a + i_{2r+1-a} \pmod{2n}$ and is equal to an odd number.

A $2r$-subsequence is \emph{contiguous} if it takes the form $(i_0+1, i_0+2, \dots, i_0+r | i_0+r+1 , \dots, i_0+2r)$ for some $i_0$.
\end{definition}

\noindent
Examples.
\begin{itemize}[itemsep=0pt, parsep=4pt, topsep=-2pt]
\item For $2n=8$, the following are subsequences:
$(2,4|5,7)$, $(5|8)$, $(8,1|2,3)$ (contiguous), $(2,3|8,1)$, $(4,6,1,8|3,2,5,7)$.
\item For $2n=8$, the following are not subsequences:
$(1,2|4,5)$, $(4,2|5,7)$, $(11|22)$.
\end{itemize}\hspace{1em}

Given an X3-tensor $\mathbf{X}(\Tt_1,\dots,\Tt_{2n}) = \Braket{ X_1^{\Tt_1} X_2^{\Tt_2} \cdots X_{2n}^{\Tt_{2n}} }$, a $2r$-subsequence $B = (i_1,\dots|\dots,i_{2r})$ defines a subtensor $H_1(T^2;\Z_p)^{2r} \to \zeta_p^\Z \cup \{0\}$.
\begin{align}
    \mathbf{X}_B( \Ts_1, \dots, \Ts_{2r} ) &\defeq \Braket{ X_{i_1}^{\Ts_1} X_{i_2}^{\Ts_2} \cdots X_{i_{2r}}^{\Ts_{2r}} } .
\end{align}
$\mathbf{X}_B$ is a partial function (in the computer science sense) derived from $\mathbf{X}$ by fixing arguments $\Tt_j$ for $j \notin B$ to 0.
It is clear that $\mathbf{X}_B$ also an X3-tensor.

\begin{proposition}
Suppose $\mathbf{X}$ be an X3-tensor of size $2n$ be of the form
\begin{align} \label{eq:Xtensor_Q}
    \mathbf{X}(\Tt_1, \Tt_2, \dots, \Tt_{2n})
    = \sum_{\theta} P(\theta) \, \sum_{i_1,\cdots, i_{n}} 
     Q(\theta)_{\Ti_1,\Tt_1;\Ti_2,-\Tt_2} \, Q(\theta)_{\Ti_2,\Tt_3;\Ti_3,-\Tt_4} \cdots Q(\theta)_{\Ti_{n},\Tt_{2n-1};\Ti_{1},-\Tt_{2n}}
\end{align}
for some non-negative function $P(\theta)$ and completely positive maps $Q(\theta) \succeq 0$ (i.e., $\sum_{\Ti,\Ts,\Tj,\Tt} v_{\Ti,\Ts}^\ast Q(\theta)_{\Ti,\Ts;\Tj,\Tt} v_{\Tj,\Tt}^\phd \geq 0$ for all $v$).
Then $\mathbf{X}_B$ is positive for all subsequences $B$.
\end{proposition}

\begin{proof}
We first argue that the contiguous subtensor $\mathbf{X}_{F(i_0)}$ associated with the $2n$-subsequence $F(i_0) = (i_0{+}1 , i_0{+}2 , \dots , i_0{+}n \,|\, i_0{+}n{+}1 , \dots , i_0{-}1 , i_0 )$ is always positive.

Write the quantum channels $Q(\theta)$ in the operator-sum representation:
$Q(\theta)[\rho] = \sum_\mu^\phd N_\mu(\theta) \rho\, N_\mu(\theta)^\dag$,
or equivalently $Q_{\Ti,\Ts;\Tj,\Tt} = \sum_\mu (N_\mu)_{\Ti,\Ts}^\phd (N_\mu)_{\Tj,\Tt}^\ast$ (with implicit $\theta$-dependence).

Presented in the tensor network representation, Eq.~\eqref{eq:Xtensor_Q} is
\begin{align}
    \begin{tikzpicture}
        [baseline={([yshift=2ex]current bounding box.center)},
        rnode1/.style={rounded corners=5pt,minimum width = 75pt, inner sep=1pt, draw=black!100,fill = white!100,line width=1.5pt,minimum height= 25pt}],
        \draw[black!100,line width=1.6pt] plot [smooth,tension=0.6] coordinates { (-65pt,-5pt) (-65pt,10pt) (-50pt,15pt)  };
        \draw[black!100,line width=1.6pt] plot [smooth,tension=0.6] coordinates { (-45pt,-5pt) (-45pt,10pt) (-30pt,15pt)  };
        \draw[black!100,line width=1.6pt] plot [smooth,tension=0.6] coordinates { (20pt,-5pt) (20pt,10pt) (15pt,15pt)  };
        \draw[black!100,line width=1.6pt] plot [smooth,tension=0.6] coordinates { (40pt,-5pt) (40pt,10pt) (25pt,15pt)  };
        \node[rnode1]() at (-12pt,20pt)    {$\mathbf{X}$};
        \node[]() at (-63pt,-10pt)    {$\Tt_1$};
        \node[]() at (-42pt,-10pt)    {$\Tt_2$};
        \node[]() at (18pt,-10pt)    {$\Tt_{2n-1}$};
        \node[]() at (45pt,-10pt)    {$\Tt_{2n}$};
        \node[]() at (-11pt,-0pt)    {$\cdots$};
    \end{tikzpicture}
    =\mkern13mu \sum_\theta P(\theta) \mkern11mu
    \begin{tikzpicture}
        [baseline={([yshift=1.5ex]current bounding box.center)},
        rnode1/.style={rounded corners=5pt,minimum width = 15pt, inner sep=1pt, draw=black!100,fill = white!100,line width=1.5pt,minimum height= 25pt}],
        \draw[black!100,line width=1.6pt] plot [smooth,tension=0.6] coordinates { (-25pt,-5pt) (-25pt,10pt) (-10pt,15pt)  };
        \draw[black!100,line width=1.6pt] plot [smooth,tension=0.6] coordinates { (10pt,15pt)  (25pt,10pt) (25pt,-5pt)  };
        \draw[black!100,line width=1.6pt] plot [smooth,tension=0.6] coordinates { (45pt,-5pt) (45pt,10pt) (60pt,15pt)   };
        \draw[black!100,line width=1.6pt] plot [smooth,tension=0.6] coordinates { (80pt,15pt)  (95pt,10pt) (95pt,-5pt)  };
        \draw[black!100,line width=1.6pt] plot [smooth,tension=0.6] coordinates { (175pt,-5pt) (175pt,10pt) (190pt,15pt)   };
        \draw[black!100,line width=1.6pt] plot [smooth,tension=0.6] coordinates { (210pt,15pt)  (225pt,10pt) (225pt,-5pt)  };
        \draw[color=black!100,line width=1.6pt] (0pt,25pt) -- (110pt,25pt);
        \draw[color=black!100,line width=1.6pt] (160pt,25pt) -- (210pt,25pt);
        \draw[black!100,line width=1.6pt] plot [smooth,tension=0.2] coordinates {(0pt,25pt) (-25pt,27pt) (-25pt,35pt) (-10pt,38pt) (210pt,38pt)  (225pt,35pt) (225pt,27pt) (205pt,25pt) };
        \node[rnode1]() at (-12pt,20pt)    {$N$};
        \node[rnode1]() at (12pt,20pt)    {$N^\dag$};
        \node[rnode1]() at (58pt,20pt)    {$N$};
        \node[rnode1]() at (82pt,20pt)    {$N^\dag$};
        \node[rnode1]() at (188pt,20pt)    {$N$};
        \node[rnode1]() at (212pt,20pt)    {$N^\dag$};
        \node[circle,draw=black, line width = 1pt, fill=white, inner sep=0pt,minimum size=2pt] (b) at(26pt,5pt) {$\scriptstyle -$};
        \node[circle,draw=black, line width = 1pt, fill=white, inner sep=0pt,minimum size=2pt] (b) at(96pt,5pt) {$\scriptstyle -$};
        \node[circle,draw=black, line width = 1pt, fill=white, inner sep=0pt,minimum size=2pt] (b) at(226pt,5pt) {$\scriptstyle -$};
        \node[]() at (-23pt,-10pt)    {$\Tt_1$};
        \node[]() at (25pt,-10pt)    {$\Tt_2$};
        \node[]() at (47pt,-10pt)    {$\Tt_3$};
        \node[]() at (95pt,-10pt)    {$\Tt_4$};
        \node[]() at (175pt,-10pt)    {$\Tt_{2n-1}$};
        \node[]() at (225pt,-10pt)    {$\Tt_{2n}$};
        \node[]() at (135pt,25pt)    {$\bullet \ \bullet \ \bullet$};
    \end{tikzpicture},
\end{align}
where \begin{tikzpicture}
        [baseline={([yshift=-0.5ex]current bounding box.center)}]
        \draw[color=black!100,line width=1.6pt] (-10pt,0pt) -- (10pt,0pt);
        \node[circle,draw=black, line width = 1pt, fill=white, inner sep=0pt,minimum size=2pt] (b) at(0pt,0pt) {$\scriptstyle -$}; 
        \node[]() at (-15pt,0pt)    {$\Tt_1$};
        \node[]() at (15pt,0pt)    {$\Tt_2$};
    \end{tikzpicture}
$= \delta(\Tt_1+\Tt_2)$ is the negation operator.
Visually, it is straightforward to see that $M[\mathbf{X}_{F(i_0)}]$ is positive semidefinite along either of the following separations (legs on the `right half' are negated in the definition of $M$):
\begin{subequations} \begin{align}
    \begin{tikzpicture}
        [baseline={([yshift=-.5ex]current bounding box.center)},
        rnode1/.style={rounded corners=5pt,minimum width = 15pt, inner sep=1pt, draw=black!100,fill = white!100,line width=1.5pt,minimum height= 25pt}],
        \draw[black!100,line width=1.6pt] plot [smooth,tension=0.6] coordinates { (-25pt,-5pt) (-25pt,10pt) (-10pt,15pt)  };
        \draw[black!100,line width=1.6pt] plot [smooth,tension=0.6] coordinates { (10pt,15pt)  (25pt,10pt) (25pt,-5pt)  };
        \draw[black!100,line width=1.6pt] plot [smooth,tension=0.6] coordinates { (-105pt,-5pt) (-105pt,10pt) (-90pt,15pt)  };
        \draw[black!100,line width=1.6pt] plot [smooth,tension=0.6] coordinates { (-70pt,15pt)  (-55pt,10pt) (-55pt,-5pt)  };
        \draw[black!100,line width=1.6pt] plot [smooth,tension=0.6] coordinates { (45pt,-5pt) (45pt,10pt) (60pt,15pt)   };
        \draw[black!100,line width=1.6pt] plot [smooth,tension=0.6] coordinates { (80pt,15pt)  (95pt,10pt) (95pt,-5pt)  };
        \draw[black!100,line width=1.6pt] plot [smooth,tension=0.6] coordinates { (125pt,-5pt) (125pt,10pt) (140pt,15pt)   };
        \draw[black!100,line width=1.6pt] plot [smooth,tension=0.6] coordinates { (160pt,15pt)  (175pt,10pt) (175pt,-5pt)  };
        \draw[color=black!100,line width=1.6pt] (-25pt,25pt) -- (95pt,25pt);
        \draw[black!100,line width=1.6pt] plot [smooth,tension=0.2] coordinates {(-55pt,25pt) (-105pt,27pt) (-105pt,35pt) (-90pt,38pt) (160pt,38pt)  (175pt,35pt) (175pt,27pt) (125pt,25pt) };
        \node[]() at (110pt,25pt)    {$\bullet \ \bullet \ \bullet$};
        \node[]() at (-40pt,25pt)    {$\bullet \ \bullet \ \bullet$};
        \draw[dotted, red!100,line width=1.6pt] (35pt,0pt) -- (35pt,45pt);
        \node[rnode1]() at (-12pt,20pt)    {$N$};
        \node[rnode1]() at (12pt,20pt)    {$N^\dag$};
        \node[rnode1]() at (-92pt,20pt)    {$N$};
        \node[rnode1]() at (-68pt,20pt)    {$N^\dag$};
        \node[rnode1]() at (58pt,20pt)    {$N$};
        \node[rnode1]() at (82pt,20pt)    {$N^\dag$};
        \node[rnode1]() at (138pt,20pt)    {$N$};
        \node[rnode1]() at (162pt,20pt)    {$N^\dag$};
        \node[circle,draw=black, line width = 1pt, fill=white, inner sep=0pt,minimum size=2pt] (b) at(-54pt,5pt) {$\scriptstyle -$};
        \node[circle,draw=black, line width = 1pt, fill=white, inner sep=0pt,minimum size=2pt] (b) at(26pt,5pt) {$\scriptstyle -$};
        \node[circle,draw=black, line width = 1pt, fill=white, inner sep=0pt,minimum size=2pt] (b) at(44pt,5pt) {$\scriptstyle -$};
        \node[circle,draw=black, line width = 1pt, fill=white, inner sep=0pt,minimum size=2pt] (b) at(124pt,5pt) {$\scriptstyle -$};
    \end{tikzpicture} \;,
\\
    \begin{tikzpicture}
        [baseline={([yshift=-.5ex]current bounding box.center)},
        rnode1/.style={rounded corners=5pt,minimum width = 15pt, inner sep=1pt, draw=black!100,fill = white!100,line width=1.5pt,minimum height= 25pt}],
        \draw[black!100,line width=1.6pt] plot [smooth,tension=0.6] coordinates { (10pt,-5pt) (10pt,10pt) (25pt,15pt)  };
        \draw[black!100,line width=1.6pt] plot [smooth,tension=0.6] coordinates { (-15pt,15pt)  (-0pt,10pt) (-0pt,-5pt)  };
        \draw[black!100,line width=1.6pt] plot [smooth,tension=0.6] coordinates { (-53pt,-5pt) (-53pt,10pt) (-38pt,15pt)  };
        \draw[black!100,line width=1.6pt] plot [smooth,tension=0.6] coordinates { (-90pt,15pt)  (-75pt,10pt) (-75pt,-5pt)  };
        \draw[black!100,line width=1.6pt] plot [smooth,tension=0.6] coordinates { (70pt,-5pt) (70pt,10pt) (85pt,15pt)   };
        \draw[black!100,line width=1.6pt] plot [smooth,tension=0.6] coordinates { (45pt,15pt)  (60pt,10pt) (60pt,-5pt)  };
        \draw[color=black!100,line width=1.6pt] (-50pt,25pt) -- (120pt,25pt);
        \draw[black!100,line width=1.6pt] plot [smooth,tension=0.6] coordinates { (145pt,-5pt) (145pt,10pt) (160pt,15pt)  };
        \draw[black!100,line width=1.6pt] plot [smooth,tension=0.6] coordinates { (107pt,15pt)  (122pt,10pt) (122pt,-5pt)  };
        \draw[black!100,line width=1.6pt] plot [smooth,tension=0.2] coordinates {(-80pt,25pt) (-105pt,27pt) (-105pt,35pt) (-90pt,38pt) (160pt,38pt)  (175pt,35pt) (175pt,27pt) (150pt,25pt) };
        \node[]() at (135pt,25pt)    {$\bullet \ \bullet \ \bullet$};
        \node[]() at (-65pt,25pt)    {$\bullet \ \bullet \ \bullet$};
        \draw[dotted, red!100,line width=1.6pt] (35pt,0pt) -- (35pt,45pt);
        \node[rnode1]() at (-12pt,20pt)    {$N^\dag$};
        \node[rnode1]() at (22pt,20pt)    {$N$};
        \node[rnode1]() at (-92pt,20pt)    {$N^\dag$};
        \node[rnode1]() at (-37pt,20pt)    {$N$};
        \node[rnode1]() at (47pt,20pt)    {$N^\dag$};
        \node[rnode1]() at (82pt,20pt)    {$N$};
        \node[rnode1]() at (107pt,20pt)    {$N^\dag$};
        \node[rnode1]() at (162pt,20pt)    {$N$};
        \node[circle,draw=black, line width = 1pt, fill=white, inner sep=0pt,minimum size=2pt] (b) at(-74pt,5pt) {$\scriptstyle -$};
        \node[circle,draw=black, line width = 1pt, fill=white, inner sep=0pt,minimum size=2pt] (b) at(1pt,5pt) {$\scriptstyle -$};
        \node[circle,draw=black, line width = 1pt, fill=white, inner sep=0pt,minimum size=2pt] (b) at(69pt,5pt) {$\scriptstyle -$};
        \node[circle,draw=black, line width = 1pt, fill=white, inner sep=0pt,minimum size=2pt] (b) at(143pt,5pt) {$\scriptstyle -$};
    \end{tikzpicture} \;.
\end{align} \end{subequations}
The former demonstrates positivity of $\mathbf{X}_{F(i_0)}$ when $i_0$ is even; the latter when $i_0$ is odd.

For every $2r$-subsequence $B = (j_1,\dots,j_r|j_{r+1},\dots,j_{2r})$ with $j_a+j_{2r+1-a} \equiv 2J+1 \pmod{2n}$, its ``$M$-matrix'' (cf.\ Def.~\ref{def:Xtensor_char}) $M[\mathbf{X}_B]$ is a principal submatrix of $M[\mathbf{X}_{F(J)}]$ constructed from the contiguous subsequence $F(J) = (J{+}1,\dots|\dots,J{+}2n)$.
Therefore $\mathbf{X}_B$ is positive.
\end{proof}

\begin{lemma} \label{lem:Xtensor_4}
Every positive X3-tensor of size 4 is monochromatic.
\end{lemma}

\begin{proof}
For $2r = 4$, there are four $\ccc$-invariants: $\ccc_{1,2,3}$, $\ccc_{1,2,4}$, $\ccc_{1,3,4}$, $\ccc_{2,3,4}$,
with a single constraint $\ccc_{1,2,3} - \ccc_{1,2,4} + \ccc_{1,3,4} - \ccc_{2,3,4} = 0$.
So there are $p^3$ possible X3-tensors of size 4.
The Hermitian requirement of $M$ further constrains $\ccc_{1,2,3} = \ccc_{2,3,4}$ and $\ccc_{1,2,4} = \ccc_{1,3,4}$, reducing to $p^2$ possibilities.

Let $\Td = -\Th -\Tv$ and consider the principal submatrix of $M$ with
rows $(t_1,t_2) \in \bigl\{ (-\Td,0), (0,-\Td), (\Th,\Tv), (\Tv,\Th) \bigr\}$
and columns $(t_3,t_4) \in \bigl\{ (0,\Td), (\Td,0), (-\Tv,-\Th), (-\Th,-\Tv) \bigr\}$.
Its elements are
\begin{align}
	\newcommand{\pha}{{\vphantom{\ast}}}
	\begin{bmatrix}
		1 & 1 & z_{134}^\ast & z_{134}^\pha
	\\ 1 & 1 & z_{234}^\ast & z_{234}^\pha
	\\	z_{124}^\pha & z_{123}^\pha & 1 & z_{123}^\pha z_{134}^\pha
	\\	z_{124}^\ast & z_{123}^\ast & z_{123}^\ast z_{134}^\ast & 1
	\end{bmatrix}
\end{align}
where $z_{ijk} = \zeta_p^{\ccc_{i,j,k}}$.
This matrix can only be positive semidefinite when $z_{134} = z_{234}$ (cf.\ Lem.~\ref{lem:column_inclusion}) which implies $\ccc_{1,3,4} = \ccc_{2,3,4}$.
Hence at size 4, positivity implies monochromaticity.
\end{proof}

\begin{lemma} \label{lem:Xtensor_6}
Every positive X3-tensor of size 6 is uniquely identified by the three invariants $\ccc_{1,2,3}$, $\ccc_{1,2,4}$, and $\ccc_{1,3,4}$.
The tensor is monochromatic if these three invariants are equal.
\end{lemma}

\begin{proof}
At size $2r=6$, there are $\lvert \operatorname{coker} d_1 \rvert = p^{10}$ possible X3-tensors [cf.\ Prop.~\ref{prop:Xtensor_group}].

Every tensor is uniquely specified by the $m$-matrix (which generically has 15 nonzero entries).  We can gauge-fix $m$ to be of the form
\begin{align}
    m &= \begin{bmatrix}
        & \alpha & \beta &&&
    \\ -\alpha && \gamma & \chi & \phi &
    \\ -\beta & -\gamma && \psi & \chi' &
    \\ & -\chi & -\psi && \gamma' & \beta'
    \\ & -\phi & -\chi' & -\gamma' && \alpha'
    \\ &&& -\beta' & -\alpha' &
    \end{bmatrix} ,
\end{align}
with ten $\Z_p$ variables: $\alpha, \beta, \gamma, \alpha', \beta', \gamma', \phi, \psi, \chi, \chi'$.
Notice that $\ccc_{1,2,6} = \alpha$ and $\ccc_{1,5,6} = \alpha'$ [cf.~\eqref{eq:def_c3_from_m}], so Hermiticity of $M$ forces $\ccc_{1,2,6} = \ccc_{1,5,6} \Rightarrow \alpha = \alpha'$.  The same argument for pairs $\ccc_{1,3,6} = \ccc_{1,4,6}$, $\ccc_{2,3,6} = \ccc_{1,4,5}$, and $\ccc_{1,2,4} = \ccc_{3,5,6}$ allow us to equate $\beta = \beta'$, $\gamma = \gamma'$, and $\chi = \chi'$.

Let $\Td = -\Th -\Tv$ and consider the principal submatrix of $M$ with selected rows/columns
\begin{align} \begin{aligned}
    &(t_1,t_2,t_3) \in \bigl\{ (-\Td,0,0), (0,-\Td,0), (0,0,-\Td), (0,\Th,\Tv), (\Tv,0,\Th), (\Th,\Tv,0) \bigr\}, \\
    &(t_4,t_5,t_6) \in \bigl\{ (0,0,\Td), (0,\Td,0), (\Td,0,0), (-\Tv,-\Th,0), (-\Th,0,-\Tv), (0,-\Tv,-\Th) \bigr\}.
\end{aligned} \end{align}
Its elements are
\begin{align}
	\newcommand{\pha}{{\vphantom{\ast}}}
	\begin{bmatrix}
		1 & 1 & 1 & z_{145}^\ast & z_{146}^\pha & z_{156}^\ast
	\\	1 & 1 & 1 & z_{245}^\ast & z_{246}^\pha & z_{256}^\ast	
	\\	1 & 1 & 1 & z_{345}^\ast & z_{346}^\pha & z_{356}^\ast	
	\\	z_{236}^\pha & z_{235}^\pha & z_{234}^\pha & 1 & z_{234}^\pha z_{246}^\pha & z_{236}^\pha z_{256}^\ast
	\\	z_{136}^\ast & z_{135}^\ast & z_{134}^\ast & z_{134}^\ast z_{145}^\ast & 1 & z_{135}^\ast z_{156}^\ast
	\\	z_{126}^\pha & z_{125}^\pha & z_{124}^\pha & z_{125}^\pha z_{145}^\ast & z_{124}^\pha z_{146}^\pha & 1
	\end{bmatrix}
\end{align}
where $z_{ijk} = \zeta_p^{\ccc_{i,j,k}}$.
Applying Lem.~\ref{lem:column_inclusion}, we have
(fifth column) $z_{146} = z_{246} = z_{346} \Rightarrow \ccc_{1,4,6} = \ccc_{2,4,6} = \ccc_{3,4,6}$
and (sixth column) $z_{156} = z_{256} = z_{356} \Rightarrow \ccc_{1,5,6} = \ccc_{2,5,6} = \ccc_{3,5,6}$.
Since $\bigl( \ccc_{1,4,6} , \ccc_{2,4,6} , \ccc_{3,4,6} \bigr) = (\beta, \chi+\beta, \psi+\beta)$ are all equal, $\chi=\psi=0$.
Since $\bigl( \ccc_{1,5,6} , \ccc_{2,5,6} , \ccc_{3,5,6} \bigr) = (\alpha, \phi+\alpha, \chi+\alpha)$ are all equal, we also have $\phi=\chi=0$.

Thus every X3-tensor comes from an $m$-matrix of the form
\begin{align}
    m &= \begin{bmatrix}
        & \alpha & \beta &&&
    \\ -\alpha && \gamma &&&
    \\ -\beta & -\gamma &&&&
    \\ &&&& \gamma & \beta
    \\ &&& -\gamma && \alpha
    \\ &&& -\beta & -\alpha &
    \end{bmatrix} ,
\end{align}
Since $\alpha = \ccc_{1,2,4}$, $\beta = \ccc_{1,3,4}$, $\gamma = \ccc_{1,2,3} - \alpha + \beta$, every X3-tensor is uniquely identified by the three invariants $\ccc_{1,2,3}$, $\ccc_{1,2,4}$, and $\ccc_{1,3,4}$.

If $\ccc_{1,2,3} = \ccc_{1,2,4} = \ccc_{1,3,4}$, then $\alpha = \beta = \gamma$.  It is straightforward (but perhaps arduous) to check for each of the twenty possible sorted triples $1 \leq i < j < k \leq 6$, $\ccc_{i,j,k} = \alpha$, and hence the X3-tensor become monochromatic.
\end{proof}

\begin{proposition}
Let $\mathbf{X}$ be an X3-tensor of size $2n$ be of the form~\eqref{eq:Xtensor_Q}.
Then every subtensor $\mathbf{X}_B$ of $\mathbf{X}$ constructed from a contiguous subsequence $B$ is monochromatic.
\end{proposition}

\begin{proof}
[Proof by induction.]
Statement$(r)$: every subtensor from a contiguous $2r$-subsequence is monochromatic.
We will let $r=3$ be the base case, and prove that Statement$(r)$ implies Statement$(r+1)$.
Without loss of generality, we assume that all contiguous $2r$-subsequence take the form $(1,\dots,r|r+1,\dots,2r)$.

For the $2r=4$ case, see Lemma~\ref{lem:Xtensor_4}.

For the $2r=6$ case, consider the 6-subsequence $(1,2,3|4,5,6)$.
Lemma~\eqref{lem:Xtensor_6}, reduces the number of possible tensors to $p^3$, identified by $\ccc_{1,2,3}$, $\ccc_{1,2,4}$, $\ccc_{1,3,4}$.
Since $(1,2|3,4)$ is also a subsequence, by lemma~\eqref{lem:Xtensor_4}, $\ccc_{1,2,3} = \ccc_{1,2,4} = \ccc_{1,3,4}$ are all equal.
Hence all contiguous 6-subsequences produce monochromatic X3-tensors.

Consider the $2(r+1)$-subsequence
\begin{align}
    B =
    \ooalign{$\phantom{ \bigl(\, 1,2,} \overbrace{\phantom{ 3,\dots,r{+}1\,|\,r{+}2,\dots,2r,2r{+}1,2r{+}2 }}^\text{$2r$-subsequence} \phantom{\,\bigr)}
    $\cr$ \bigl(\, \underbrace{1,2,3,\dots,r{+}1\,|\,r{+}2,\dots,2r}_\text{$2r$-subsequence} ,2r{+}1,2r{+}2 \,\bigr) $}
    \,.
\end{align}
By our induction hypothesis, the subsequence $(1,\dots,2r)$ is monochromatic as is the subsequence $(3,\dots,2r+2)$.
Because $r\geq3$, the two subsequences overlap by at least 4 elements, so $\ccc_{g,h,i} = \ccc_{3,4,5} = \ccc_{j,k,l}$ for $1 \leq g<h<i \leq 2r$ and $3 \leq j<k<l \leq 2r+2$.
In other words, $\ccc_{i,j,k}$ must be equal for any sorted triplet $i<j<k$ as long as all three $i,j,k$ are greater than 2 or less than $2r+1$.

Next, consider subsequences of the form $B_a = (1,2,a \,|\, \bar{a},2r+1,2r+2)$ with $\bar{a} = 2r+3-a$ where $3 \leq a \leq r+1$.
Any distinct triplet $(i,j,k) \in [1,2r+2]$ that contains both an element in $[1,2]$ and an element in $[2r+1,2r+2]$ must be a subset of some subsequence of this form.
By Lemma~\eqref{lem:Xtensor_6}, the possible tensors $\mathbf{X}_{B_a}$ are identified by $\ccc_{1,2,a}$, $\ccc_{1,2,\bar{a}}$, and $\ccc_{1,a,\bar{a}}$.
By our earlier assertion, all these must be equal since $1,2,a,\bar{a} \in [1,2r]$.
Therefore, $\ccc_{i,j,k}$ for any sorted triplet $i<j<k$ must be equal; $\mathbf{X}_B$ constructed from the $2(r+1)$-subsequence is indeed monochromatic.
\end{proof}

Thus, X3-tensor of the form~\eqref{eq:Xtensor_Q} are completely characterized by a single invariant $c \in \Z_p$ such that $\ccc_{i,j,k} = c$ for all sorted triplets $i < j < k$.
By inspection,
\begin{align}
	m_{i,j} = \begin{cases} c & i < j , \\ -c & i > j . \end{cases}
\end{align}
is a compatible solution.
This concludes the proof to Prop.~\ref{prop:case_mf_same_m}.

\end{document}